\documentclass [invmat,numbook]{svjour} 
\usepackage[adobe-utopia]{mathdesign}
\usepackage[OMLmathbf,sfdefault=cmbr]{isomath}

\def\bos#1{{\mathbf{#1}}}
\def\mb#1{\hbox{{\mathversion{bold}$#1$}}}

\begin{document}


\catcode`\@=11 

\global\newcount\nsecno \global\nsecno=0
\global\newcount\meqno \global\meqno=1
\def\newsec#1{\global\advance\nsecno by1
\eqnres@t
\section{#1}}
\def\eqnres@t{\xdef\nsecsym{\the\nsecno.}\global\meqno=1}
\def\sequentialequations{\def\eqnres@t{\bigbreak}}\xdef\nsecsym{}

\def\draftmode{\message{ DRAFTMODE }
\writelabels

{\count255=\time\divide\count255 by 60 \xdef\hourmin{\number\count255}
\multiply\count255 by-60\advance\count255 by\time
\xdef\hourmin{\hourmin:\ifnum\count255<10 0\fi\the\count255}}}
\def\nolabels{\def\wrlabeL##1{}\def\eqlabeL##1{}\def\reflabeL##1{}}
\def\writelabels{\def\wrlabeL##1{\leavevmode\vadjust{\rlap{\smash%
{\line{{\escapechar=` \hfill\rlap{\tt\hskip.03in\string##1}}}}}}}%
\def\eqlabeL##1{{\escapechar-1\rlap{\tt\hskip.05in\string##1}}}%
\def\reflabeL##1{\noexpand\llap{\noexpand\sevenrm\string\string\string##1}
}}

\nolabels

\def\eqn#1#2{
\xdef #1{(\nsecsym\the\meqno)}
\global\advance\meqno by1
$$#2\eqno#1\eqlabeL#1
$$}

\def\eqalign#1{\null\,\vcenter{\openup\jot\m@th
\ialign{\strut\hfil$\displaystyle{##}$&$\displaystyle{{}##}$\hfil
\crcr#1\crcr}}\,}

\def\foot#1{\footnote{#1}}

\catcode`\@=12 
%

\def\a{\alpha}
\def\b{\beta}
\def\c{\chi}
\def\d{\delta}  \def\D{\Delta}
\def\e{\varepsilon} \def\ep{\epsilon}
\def\f{\phi}  \def\F{\Phi}
\def\g{\gamma}  \def\G{\Gamma}
\def\k{\kappa}
\def\l{\lambda}  \def\La{\Lambda}
\def\m{\mu}
\def\n{\nu}
\def\r{\rho}
\def\vr{\varrho}
\def\o{\omega}  \def\O{\Omega}
\def\p{\psi}  \def\P{\Psi}
\def\s{\sigma}  \def\S{\Sigma}
\def\th{\theta}  \def\vt{\vartheta}
\def\t{\tau}
\def\w{\varphi}
\def\x{\xi}
\def\z{\zeta}
\def\U{\Upsilon}
\def\CA{{\cal A}}
\def\CB{{\cal B}}
\def\CC{{\cal C}}
\def\CD{{\cal D}}
\def\CE{{\cal E}}
\def\CF{{\cal F}}
\def\CG{{\cal G}}
\def\CH{{\cal H}}
\def\CI{{\cal I}}
\def\CJ{{\cal J}}
\def\CK{{\cal K}}
\def\CL{{\cal L}}
\def\CM{{\cal M}}
\def\CN{{\cal N}}
\def\CO{{\cal O}} 
\def\CP{{\cal P}}
\def\CQ{{\cal Q}}
\def\CR{{\cal R}}
\def\CS{{\cal S}}
\def\CT{{\cal T}}
\def\CU{{\cal U}}
\def\CV{{\cal V}}
\def\CW{{\cal W}}
\def\CX{{\cal X}}
\def\CY{{\cal Y}}
\def\CZ{{\cal Z}}
%

\def\V{\mathbb{V}}
\def\E{\mathbb{E}}
\def\R{\mathbb{R}}
\def\C{\mathbb{C}}
\def\Z{\mathbb{Z}}
\def\A{\mathbb{A}}
\def\T{\mathbb{T}}
\def\L{\mathbb{L}}
\def\D{\mathbb{D}}
\def\Q{\mathbb{Q}}


\def\mJ{\mathfrak{J}}
\def\mq{\mathfrak{q}}
\def\mQ{\mathfrak{Q}}
\def\mP{\mathfrak{P}}
\def\mp{\mathfrak{p}}
\def\mH{\mathfrak{H}}
\def\mh{\mathfrak{h}}
\def\ma{\mathfrak{a}}
\def\mA{\mathfrak{A}}
\def\mC{\mathfrak{C}}
\def\mc{\mathfrak{c}}
\def\ms{\mathfrak{s}}
\def\mS{\mathfrak{S}}
\def\mm{\mathfrak{m}}
\def\mM{\mathfrak{M}}
\def\mn{\mathfrak{n}}
\def\mN{\mathfrak{N}}
\def\mt{\mathfrak{t}}
\def\ml{\mathfrak{l}}
\def\mT{\mathfrak{T}}
\def\mL{\mathfrak{L}}
\def\mo{\mathfrak{o}}
\def\mg{\mathfrak{g}}
\def\mG{\mathfrak{G}}
\def\mf{\mathfrak{f}}
\def\mF{\mathfrak{F}}
\def\md{\mathfrak{d}}
\def\mD{\mathfrak{D}}
\def\mO{\mathfrak{O}}
\def\mk{\mathfrak{k}}
\def\mK{\mathfrak{K}}
\def\mR{\mathfrak{R}}
\def\sA{\mathscr{A}}
\def\sB{\mathscr{B}}
\def\sC{\mathscr{C}}
\def\sD{\mathscr{D}}
\def\sE{\mathscr{E}}
\def\sF{\mathscr{F}}
\def\sG{\mathscr{G}}
\def\sL{\mathscr{L}}
\def\sM{\mathscr{M}}
\def\sN{\mathscr{N}}
\def\sO{\mathscr{O}}
\def\sP{\mathscr{P}}
\def\sQ{\mathscr{Q}}
\def\sR{\mathscr{R}}
\def\sS{\mathscr{S}}
\def\sT{\mathscr{T}}
\def\sU{\mathscr{U}}
\def\sV{\mathscr{V}}
\def\sW{\mathscr{W}}
\def\sX{\mathscr{X}}
\def\sY{\mathscr{Y}}
\def\sY{\mathscr{Z}}


\hyphenation{anom-aly anom-alies coun-ter-term coun-ter-terms
}

\def\tr{{\rm tr}} \def\Tr{{\rm Tr}}

\def\tilde{\widetilde} 
\def\hat{\widehat}
%

\def\grad#1{\,\nabla\!_{{#1}}\,}
\def\gradgrad#1#2{\,\nabla\!_{{#1}}\nabla\!_{{#2}}\,}
\def\ph{\varphi}
\def\psibar{\overline\psi}
\def\om#1#2{\omega^{#1}{}_{#2}}
\def\vev#1{\langle #1 \rangle}
\def\ha{{1\over2}}
\def\half{{\textstyle{1\over2}}} 
\def\roughly#1{\raise.3ex\hbox{$#1$\kern-.75em\lower1ex\hbox{$\sim$}}}

\def\rd{\partial}
\def\ha{{\textstyle{1\over2}}}
\def\fr#1#2{{\textstyle{#1\over#2}}}
\def\Fr#1#2{{#1\over#2}}
\def\fs#1{#1\!\!\!/\,}   
\def\Fs#1{#1\!\!\!\!/\,} 
\def\ato#1{{\buildrel #1\over\longrightarrow}}
\def\up#1#2{{\buildrel #1\over #2}}

\def\pr{\prime}
\def\ppr{{\prime\prime}}

\def\bari{\bar\imath}
\def\barj{\bar\jmath}
\def\mapr#1{\!\smash{\mathop{\longrightarrow}\limits^{#1}}\!}
\def\mapl#1{\!\smash{\mathop{\longleftarrow}\limits^{#1}}\!}
\def\mapbr{\!\smash{\mathop{\longrightarrow}\limits^{\bbs_+}}\!}
\def\mapbl{\!\smash{\mathop{\longleftarrow}\limits^{\bbs_-}}\!}
\def\mapd#1{\Big\downarrow\rlap{$\vcenter{#1}$}}
\def\mapu#1{\Big\uparrow\rlap{$\vcenter{#1}$}}
\def\maprd{\rlap{\lower.3ex\hbox{$\scriptstyle\bs_+$}}\searrow}
\def\mapld{\swarrow\!\!\!\rlap{\lower.3ex\hbox{$\scriptstyle\bs_-$}}}
\def\ne{\nearrow}
\def\se{\searrow}
\def\nw{\nwarrow}
\def\sw{\swarrow}
\def\etal{et al.}
\def\Ker{\hbox{Ker}\;}
\def\Im{\hbox{Im}\;}

\def\ket#1{\left|\bos{ #1}\right>}\vspace{.2in}
   \def\bra#1{\left<\bos{ #1}\right|}
\def\oket#1{\left.\bos{ #1}\right>}
\def\obra#1{\left<\bos{ #1}\right.}
\def\epv#1#2#3{\left<\bos{#1}\left|\bos{#2}\right|\bos{#3}\right>}
\def\qbvk#1#2{\bos{\left(\bos{#1},\bos{#2}\right)}}
\def\Hoch{{\tt Hoch}}
\def\rrd{\up{\rightarrow}{\rd}}
\def\lrd{\up{\leftarrow}{\rd}}


\def\mod{\hbox{ }mod\hbox{ }}


\title{Algebraic Principles
of Quantum Field Theory I}
\subtitle{Foundation and an exact solution of BV QFT}
\author{Jae-Suk Park
\thanks{This work was supported by 
the Korea Research Foundation Grant 
funded by the Korean Government (MOEHRD, Basic Research Promotion Fund) 
(KRF-2006-331-C00032).}
}
\institute{
Department of Mathematics,
Yonsei University,
Seoul 120-749, Korea.
\\
\email{jaesuk@yonsei.ac.kr}
}
\dedication{
I believe it might interest a philosopher, one who can think himself, 
to read my notes. 
For even if I have hit the mark only rarely, 
he would recognize what targets I had been ceaselessly aiming at. 
-Ludwig Wittgenstein
}


\maketitle
\begin{abstract}
This is the first in a series of papers on an attempt to understand
quantum field theory mathematically.  In this paper we shall
introduce and
study  BV QFT algebra and BV QFT as the proto-algebraic model of
quantum field theory by exploiting Batalin-Vilkovisky quantization scheme.
We shall develop a complete theory of obstruction (anomaly) to quantization of 
classical observables and propose that expectation value of quantized observable is 
certain quantum homotopy invariant.
We shall, then, suggest a new method,
bypassing Feynman's
path integrals,
of computing quantum correlation functions 
when there is no anomaly. 
An exact solution for all quantum correlation functions shall be presented  
provided that the number of equivalence classes of observables is finite for each ghost numbers.
Such a theory shall have its natural family parametrized by a smooth-formal moduli space 
in quantum coordinates, which notion generalize that of flat or special  coordinates
in topological string theories and shall be interpreted as 
an example of quasi-isomorphism of general QFT algebra.

\end{abstract}

\newsec{Introduction}

This series of papers is on a quest to gain some
mathematical understanding of  Quantum Field Theory,
hoping to arrive at  certain algebraic category  equivalent 
to the  "category of  quantum field theories".  
Such algebraic category is to be called category of QFT algebras (up to homotopy), and it shall be proposed 
that 
quantum field theory is  a study of morphisms of QFT algebras such that
two quantum field theories are physically equivalent if and only if the associated
QFT algebras are  equivalent.
As a justification of this elusive program, we will argue that it is possible 
to capture rather complete physical information by studying such algebraic category. 

Our journey begins with  setting up the prototype of QFT algebra named  BV QFT algebra
after some reflections on the widely accepted 
and  rather general scheme for quantization of classical field theory due to 
Batalin and Vilkovisky (BV) \cite{BV}.
A BV QFT shall be a BV QFT algebra 
with a natural algebraic counterpart to Feynman path integral.
We shall develop
a complete obstruction theory (theory of anomaly) of quantization of classical observables 
to quantum observables. We shall, then,  present exact solution for all quantum correlation functions for a BV QFT
without anomaly and with a finite number of physically in-equivalent observables. Such a BV QFT
always comes with its family  parametrized by a smooth formal moduli space
in "quantum coordinates". This result shall be a case study of quasi-isomorphism of QFT algebra.

A further study of our notion of quantum coordinates, which is a natural generalization of
the notion of flat or special coordinates on moduli spaces of various topological string theories 
(the flat structure of K.\ Saito \cite{Saito} or Witten-Dijkgraaf-Verlinde-Verlinde equation \cite{W,DVV}) 
to general anomaly free QFT
   , and  its natural homotopy  generalization in the context of "homotopy path integrals" shall be
    the subjects of  two sequels \cite{PII,PIII} to this paper. 
We shall, then, return to the very beginning  to face with anomaly and its fundamental  implications
to  quantization of classical field theory in $4$-th paper \cite{PIV}.
The $5$th and the final paper in this series is about correct definition and some properties
of general\foot{In this series we do not concern non-commutative QFT.} QFT algebra \cite{PV}. 
We are also planning to write few  companion papers on some
examples and applications.

In what follows, we have  prologue, summary and epilogue to this paper partly due to the nature 
of our program and partly because of rather lengthy and sometimes technical nature of its main body. 
\begin{itemize}
\item
Prologue is a description of the BV quantization scheme such as the meaning of BV classical and
quantum master equations,  classical and quantum observables 
and how the notion of expectation value of quantum observable via path integral
is realized. After some
reflections on the scheme, we shall be driven to consider deformations of
the given quantum field theory to study quantum correlators. The deformation
problem appears to be governed by Maurer-Cartan equation of certain differential
graded Lie algebra (DGLA) which is nothing but the BV quantum master equation
for family of quantum master action functional. Then we shall face arbitrariness of 
quantum correlation functions. A resolution of this conundrum shall be the basic content of this paper. 
We shall argue that such DGLA  should be regarded as a descendant structure 
of more fundamental algebraic structure. 

\item
In Summary, we sketch  the notion of BV QFT algebra and its descendant DGLA as
the prototype of QFT algebra and its descendant.
A BV QFT algebra shall be a "quantum cochain complex" together with a
super-commutative  associative  product satisfying certain conditions with respect to 
$\hbar$ (formal  Planck constant). A BV QFT shall be a BV QFT algebra with a natural algebraic
counterpart to Feynman path integral [Section 2]. 
 We shall also sketch 
a complete obstruction theory (theory of anomaly) of quantization of 
classical observables to quantum observables as certain extension
problem of classical cochain map to its quantum counterpart.
Expectation value of quantized observable shall be a  "quantum homotopy
invariant" [Section 3]. We shall, then, sketch the exact solution 
of all quantum correlation functions for a BV QFT with the assumption
that
(i) there is no anomaly in quantization of classical observables 
and 
(ii) the number of physically in-equivalent observables is finite.
As a corollary, we shall show that such a BV QFT always
comes with its family  parametrized by a smooth formal moduli space
in "quantum coordinates", which notion generalize that of flat coordinates
in moduli space of topological strings [Section 4]. 

\item
In Epilogue, we argue that the above summarized result of this paper on  quantum correlation functions
of BV QFT and quantum coordinates on its moduli space
should be interpreted as  a case study of morphism of QFT algebra and its descendant morphism.
This suggests a new method of computing quantum correlation functions 
bypassing perturbative Feynman path integrals.

\end{itemize}

I would like to thank Dennis Sullivan and John Terilla for our long standing conversations  on
mathematics of quantum field theory. I am also grateful to Bung Heup Jun for a proof-reading.
Some results in this paper have been presented at the Graduate Center of CUNY (04,06,07), 
Erwin Schr\"odinger Institute (07), and Max Planck Institute of Bonn (08). 
I would  like to thank to those institutions for invitations and hospitality. 
I was managed, after a long gestation, to type most of this manuscript during my visit  to IHES
(Jan-Feb, 2010)
to which I am very grateful for its ideal working environment and hospitality.

\subsection{Prologue}

\subsubsection{The BV Quantization Scheme}
Historically  BV quantization scheme was invented to have a consistent
path integral approach to QFT in the presence of certain gauge symmetry of 
a given classical action $S_{cl}$,  which is certain function on the space $\mL_{cl}$ 
of relevant classical fields \cite{BV}.
The presence of gauge symmetry requires Physicists to make a suitable choice of gauge fixing 
to do Feynman path integrals a la Faddev-Popov and its homological (BRST) interpretation.
BV quantization scheme unifies both Faddev-Popov and BRST procedure not only 
in a greater generality 
but also gives certain consistent condition of path integrals to be independent of 
choice of gauge fixing.
Such a consistency condition is stated in terms of so-called BV quantum master equation, 
which solution ${\bf{S}}$ is called BV quantum master action if it is related with 
the given classical action $S_{cl}$ in certain ways.

We first  recall the common setup to define the BV quantum master equation.  
One of the most important  ingredient for the recipe is the BV operator $\Delta$ of  BV-algebra, 
which is a triple $\left(\sC, \Delta, \,\cdot\,\right)$ 
satisfying the following properties: (i) $\left(\sC,\,\cdot\,\right)$ is 
a $\Z$-graded associative and super-commutative 
$\R$-algebra, and the $\Z$-grading of $\sC$ is specified by so called the ghost number.
(ii) BV operator $\Delta$ is a $\R$-linear operator 
$\Delta:\sC^i\longrightarrow \sC^{i+1}$ of ghost number $1$
satisfying $\Delta^2=0$ which  failure of being a derivation of the product $\cdot$ defines
so called BV bracket 
$\left(\hbox{ },\hbox{ }\right):\sC^{i}\otimes \sC^j\longrightarrow \sC^{i+j+1}$, 
\eqn\introaa{
(-1)^{|a|}(a,b):=\Delta ({a}\cdot {b}) - \Delta{a}\cdot {b} -(-1)^{|{a}|}{a}\cdot \Delta {b},
}
where $|a|$ denotes the ghost number of $a$,
which is  a derivation of  the product. It follows that BV bracket is a graded Lie bracket with
ghost number $1$.
One also introduce the Planck constant $\hbar$, 
regarded as a formal parameter (for our purpose), 
and extends those algebraic structures naturally and trivially (no star product)
to those on $\sC[[\hbar]]$  by the condition of $\R[[\hbar]]$-linearity and $\hbar$-adic
continuity.

In BV quantization procedure
such a BV  algebra is realized as algebra 
of functions on a graded space $\mC$  of certain {\tt fields} and their {\tt anti-fields}.
The space $\mL$ of all {\tt fields} includes the space $\mL_{cl}$ of  classical fields
and ghost fields due to gauge symmetry of $S_{cl}$, symmetry of gauge symmetry etc. 
The space $\mC$ is in the form $\mC\simeq T^{*}[-1]\mL$ such that it has 
a natural odd symplectic structure with  ghost number $-1$ and the space $\mL$ of 
all {\tt fields} is a Lagrangian subspace. Then the BV operator $\Delta$ is certain 
odd differential operator of $2$nd-order
such that its BV bracket  $\left(\hbox{ },\hbox{ }\right)$ corresponds to the graded Poisson bracket
associated with the symplectic structure on $\mC$.

A BV quantum master action 
 ${\bf{S}}=S + \hbar S^{(1)} + \cdots \in \sC[[\hbar]]^0$ is a solution
to  the {\it BV quantum master equation};
\eqn\introa{
\hbar^2 \Delta\, e^{-{\bf{S}}/\hbar}=0,
}
which is equivalent to 
\eqn\introay{
-\hbar \Delta \bos{S} +\Fr{1}{2}\left(\bos{S},\bos{S}\right)=0,
}
such that
its classical limit  $S:={\bf{S}}\bigl|_{\hbar=0} \in \sC^0$, called  BV classical master action, 
restricted to 
the space $\mL$ of  {\tt fields} is the classical action 
$S_{cl}$, i.e.,  $S\bigl|_{\mL}= S_{cl}$, and 
is supposed to have encoded complete information of the classical field theory -- the nature 
gauge symmetry of the classical action $S_{cl}$ and symmetry of the gauge symmetry  etc., such that
it satisfies the condition $(S,S)=0$, called classical BV master equation.

%


To be more concrete, let's decompose BV classical master action $S$ 
based on the condition that $S\bigl|_{\mL}=S_{cl}$
as follows
\eqn\introi{
S =S_{cl}+\G.
}
Then, the piece $\G$, which contains  ``ghosts and antis'', should
have information of gauge symmetry of the classical action $S_{cl}$ and the nature of symmetry 
etc. BV classical master equation  $\left(S,S\right)=0$ is, then, equivalent to
\eqn\introj{
\left(S_{cl}, \G\right)+\Fr{1}{2}\left(\G,\G\right)=0,
}
since  BV bracket vanishes on $\mL$ so that  $\left(S_{cl},S_{cl}\right)=0$.

Here is a simple explanation for the classical picture.
Recall that classical physics is dictated by classical equation of motion so that everything 
should be considered modulo classical equation of motion. In the framework of Batalin and Vilkovisky,
any expression in the form $\left(S_{cl},\l\right)$ vanishes by classical equation of motion
and expression in the form $\left(\G, \l\right)$ represents  action of the symmetry on $\l$.
An element $O \in \sC$ satisfying $\left(S,O\right)=0$ is called a {\it classical observable},
and two classical observables $O$ and $O^{\pr}$ are said to be classical physically equivalent
if there is some $\l \in \sC$ such that $O^{\pr}-O = \left(S,\l\right)$. 
We, then, note that  the condition $\left(S,O\right)=0$  is rewritten as follows, 
\eqn\introk{
\left(\G, O\right)=- \left(S_{cl},O\right),
}
and  means that $O$ should be invariant under the symmetry 
of $S_{cl}$ modulo classical equation of motion to be a classical observable. 
It follows that two classical observables $O$ and $O^{\pr}$ must be classical
physically equivalent (indistinguishable to a classical observer)  if their difference
$O^{\pr}-O$ can be gauge transformed away modulo equation of motion, i.e., 
$O^{\pr}-O= (\G, \l) + (S_{cl},\l)$. 
Finally the criterion  \introk\ itself  should be invariant under the symmetry
modulo the classical equation of motion, leading to 
the consistency condition \introj.

\begin{remark}
To be more faithful to  physical viewpoint,
we better say a classical observable $O$ above a classical master observable and it
is the restriction $O_{cl}$ of $O$ to $\mL$  is classical observable, which should 
depend  only on classical fields. Decomposing $O = O_{cl}+ V$ accordingly,
and the equation \introk\ contains the requirement 
$$(\G,O_{cl})\biggl|_{\mL}=-(S_{cl},V)\biggl|_{\mL}$$ 
that 
$O_{cl}$ should be invariant under the gauge symmetry 
of $S_{cl}$ modulo the classical equation of motion. 
But the above condition leads to, possibly infinite, sequence of integrability (or consistency)
conditions, all of which can be summarized by \introk. Being understood we shall maintain
to call $O$ a classical observable. 
\end{remark}

\begin{remark}
For a  peace with one more widely used classical physical terminology,
consider the operation (representing action of the symmetry) $(\G,\bullet)\bigr|_{\mL}$, 
which corresponds to the BRST operator $\d_{BRST}$. 
The classical BV master equation \introj\ has  the following leading requirement
$$\Fr{1}{2}\left(\G,\G\right)\bigl|_{\mL}=-\left(S_{cl}, \G\right)\bigl|_{\mL}$$
that (a representation represents) $\d_{BRST}^{2}=0$ modulo the classical
equation of motion. Again there is, possibly infinite, sequence of integrability 
(or consistency) conditions, all of which can be summarized by the classical master
equation \introj. By the way the term BRST quantization is a misnomer.
\end{remark}

A quantum theoretic notion of observables and their equivalence are to be based on
Feynman Path Integral. Batalin and Vilkovisky  interpreted a path integral
as an integral over the Lagrangian subspace $\mL$ of $\mC$ in the
following form
\eqn\introb{
\left<\bos{O}\right>="\int_{\mL} \!\!d\m"\; \bos{O} \cdot e^{-{\bf{S}}/\hbar},
}
where $\bos{O} \in \sC[[\hbar]]$, called a {\it quantum observable}, 
should satisfy the following condition
\eqn\introc{
\hbar^2 \Delta \left(-\Fr{1}{\hbar}\bos{O} \cdot e^{-{\bf{S}}/\hbar}\right)=0,
}
which is equivalent to 
\eqn\introcy{
-\hbar \Delta \bos{O} + \left(\bos{S},\bos{O}\right)=0.
}
The above condition \introc\ together with the condition \introa\ is
formal assurance that the path integral \introb\ 
does not depend on continuous changes of $\mL$
(homologous deformation of $\mL$ in general), which changes amount 
to the changes of gauge fixing.

\begin{remark}
A BV quantum master action functional ${\bf{S}}$ may be regarded
as a sequence of $\hbar$-corrections 
to classical BV action functional $S$, which contains information of 
the given classical action functional $S_{cl}$ 
and its gauge symmetry etc.\ etc., such that its path integral is independent of 
choice of gauge fixing. Then, being assured, one may choose a gauge suitable 
to the given situation
and proceed to study perturbative Feynman path integrals.
We should, however, emphasize  that the BV quantization scheme is more than 
obtaining quantum master action as a preparation of gauge fixing and  
subsequent computations of Feynman diagrams. 
\end{remark}

There is a fundamental  identity that a  Batalin-Vilkovisky-Feynman path integral is
supposed to be satisfied; for any $\bos{\l}=\l +\hbar \l^{(1)}+\cdots \in \sC[[\hbar]]$,
\eqn\introd{
 "\int_{\mL} \!\!d\!\m"\;\hbar\Delta\left(\bos{\l} \cdot e^{-{\bf{S}}/\hbar}\right)=0.
}
 This identity, besides from 
its practical utilities of informing us what kind of Feynman path integrals must vanish 
before gauge fixing and perturbative analysis\foot{This identity is a simultaneous generalization
of  the Schwinger-Dyson equation and the Ward identity.}, 
gives the notion of quantum physical equivalence of
observables: Assume that  
two quantum observables $\bos{O}$ and $\bos{O}^\pr$ are related
as follows;
\eqn\introda{
\bos{O}^\pr \cdot e^{-{\bf{S}}/\hbar} 
=\bos{O}\cdot e^{-{\bf{S}}/\hbar} -\hbar\Delta \left(\bos{\l}\cdot e^{-{\bf{S}}/\hbar}\right),
} 
Then the identity \introd\ implies that the two quantum observables 
must have the same  value in path integrals, 
$\left<\bos{O}^\pr\right> = \left<\bos{O}\right>$. So those observables
are said to be quantum physically equivalent.

The relation \introda\ is equivalent
to $\bos{O}^{\pr}-\bos{O}= -\hbar \Delta \bos{\l} +\left(\bos{S},\bos{\l}\right)$,
so that the classical limit $O$ and $O^{\pr}$ of $\bos{O}$ and $\bos{O}^{\pr}$ are
classical physically equivalent observables, i.e., $O^{\pr}-O = (S,\l)$ where $\l=\bos{\l}\big|_{\hbar=0}$.  
Physicist may say a classical observable $O\in \sC$, $(S,O)=0$,
is quantizable   if there is a sequence of quantum corrections 
$\bos{O}=O +\hbar O^{(1)}+\hbar^{2}O^{(2)}+\cdots \in \sC[[\hbar]]$ of it 
such that  $-\hbar \Delta \bos{O} +\left(\bos{S}, \bos{O}\right)=0$.

\begin{remark}

No Physicist would say that she or he is actually defining and doing math with the integral 
\introb, which since  is more like an artistic symbol for collective wisdom and mastery.
It should be noted, however, the finite dimensional version of Batalin-Vilkovisky-Feynman path integral
exists mathematically and all of its desired properties are theorems \cite{Schwarz}. 

\end{remark}


\begin{remark}
Physicist  may call $\left<\bos{O}\right>$ in \introb\   un-normalized expectation value
of the (quantum) observable $\bos{O}$.  The following  normalization
$$
{\left<\bos{O}\right>\over \left<{1}\right>}
={"\int_{\mL} \!d\!\m"\; \bos{O} \cdot e^{-{\bf{S}}/\hbar}
\over "\int_{\mL} \!d\!\m"\;  e^{-{\bf{S}}/\hbar}},
$$
is canonical, provided that the partition function $\left<1\right>$ is non-zero.
In general, physicist seems to assume that there is a suitable normalization such that
expectation value of every quantum observable has no negative power in $\hbar$.
We shall adopt such viewpoint throughout this paper, and expectation
value shall always mean such the normalized expectation value.
\end{remark}



\subsubsection{Conundrum:  arbitrariness of quantum correlators}


An implication of BV quantization scheme is that  one might identify path integral 
with a certain linear map from the space of equivalence classes of quantum observables
to $\R[[\hbar]]$. 
At this stage it is convenient to denote the operator $-\hbar\Delta +\left(\bos{S},\bullet\right)$
by the single letter $\bos{K}$. Then $\bos{K}$ increase the ghost number by $1$ and
satisfies $\bos{K}^2=0$ due the BV quantum master equation \introa. Let $Q$ denote
the classical limit of $\bos{K}$, $Q=\bos{K}\bigl|_{\hbar=0}=(S,\bullet)$ such that $Q^{2}=0$.
The condition \introc\ for quantum observable $\bos{O}$ is $\bos{K}\bos{O}=0$ and 
the identity \introd\ is $\left<\bos{K}\bos{\l}\right>=0$.
Also the condition \introda\
for  the two quantum observables $\bos{O}$, and $\bos{O}^\pr$ being physically equivalent 
is
$\bos{O}^\pr=\bos{O}+\bos{K}\bos{\l}$, 
so  that we have
$\left<\bos{O}^\pr\right> = \left<\bos{O}\right>+\left<\bos{K}\bos{\l}\right>
= \left<\bos{O}\right>$.
Thus path integral might be interpreted as a certain linear map from the cohomology of the 
complex $(\sC[[\hbar]],\bos{K})$.   
It might, then, be natural to study correlation functions of quantum observables by exploiting
algebraic structure in cohomology of the cochain complex  $(\sC[[\hbar]],\bos{K})$.  
%


%
Naively, correlation function of two quantum observables is the expectation value 
of the product two quantum observables. However,  the product of two quantum observables 
$\bos{O}_{1}$ and $\bos{O}_{2}$  may not even be a quantum observable in general. 
Even for the case that $\bos{K}\left(\bos{O}_{1}\cdot\bos{O}_{2}\right)=0$, 
the $\bos{K}$-cohomology class of the product $\bos{O}_{1}\cdot\bos{O}_{2}$ 
has no canonical meaning in terms of equivalence classes of $\bos{O}_{1}$ and $\bos{O}_{2}$.
This is due to the fundamental property 
\introaa\ of the BV operator $\Delta$  that it is not a derivation of the product, which implies 
that the  operator $\bos{K}:=-\hbar\Delta + ({\bf{S}},\bullet)$ is also not 
a derivation of the product -the failure of $\bos{K}$ being a derivation of 
the product is  proportional to $\hbar$ so that
that the classical limit $Q$ of $\bos{K}$ is a derivation of the product.
In the classical picture the story is  different. 
The condition for classical observable is $Q O=0$, and two classical observables
$O$ and $O^{\pr}$ are (classical) physically equivalent if $O^{\pr}=O +Q \l$. Thus
two classical observables are physically equivalent if and only if they belong to the same
cohomology class of the cochain complex $(\sC,Q)$. In the classical picture, however,
$Q$ is a derivation of the product so that there is well-defined algebra of equivalence
classes of classical observables.

There is seemingly  a natural resolution of the above vexing problem 
if the given QFT comes with certain  family \cite{Park1}.
We  observe that the two consistent conditions \introa\ and \introc\
are combined into the following single equation with certain parameter $t$
\eqn\introe{
\hbar^2 \Delta\, e^{-\Fr{1}{\hbar}\left({\bf{S}} + t \bos{O}\right)}=0 \hbox{ modulo } t^2.
}
Thus it seems natural to associate the given BV quantized field theory with 
BV quantum master action ${\bf{S}}$ to a  family of theories 
with deformed BV quantum master action $\bos{S}_{\bos{\Theta}} = \bos{S} + \bos{\Theta}$;
\eqn\introaa{
\hbar^2 \Delta\, e^{-\bos{S}_{\bos{\Theta}} /\hbar}\equiv
\hbar^2 \Delta \left( e^{-\bos{\Theta}/\hbar} \cdot e^{-\bos{S}/\hbar}\right)=0,
}
such that infinitesimal part of the deformation term $\bos{\Theta}$ is given by 
quantum observables of the initial theory.
The above equation for such a family is equivalent to the following equation
\eqn\deform{
\bos{K} \bos{\Theta} +\Fr{1}{2}\left(\bos{\Theta},\bos{\Theta}\right) =0.
}
It turns out that an all order solution $\bos{\Theta}$, if exists, 
also can be used to "quantum correct" the products 
of $n$ quantum observables  to be quantum observables for all $n=2,3,\cdots$.

 \begin{example}

Let $\{\bos{O}_i\}$, $i\in I$, denote a certain set of quantum observables, 
of our interests. 
Introduce a corresponding set of parameters
$\{t_i\}$ such that $|t_i|+|\bos{O}_i|=0$. Then 
$\bos{\Theta}=\sum_i t_i \bos{O}_i + \cdots$ is a solution to \deform\ 
modulo $(t)^2$ since $\bos{K}\bos{O}_{i}=0$.
The ability to extend the first order solution to a second order solution amounts to 
existence of certain set
 $\{ \bos{O}_{ij}\}$ in $\sC[[\hbar]]^{|{\bf O}_{i}|+|{\bf O}_{j}|}$ 
satisfying
\eqn\dema{
-(-1)^{|\bos{O}_{i}|}\left(\bos{O}_{i}, \bos{O}_{j}\right)=\bos{K}\bos{O}_{ij}.
}
Then 
$$\bos{\Theta}=\sum_i t_i \bos{O}_i + \Fr{1}{2}\sum_{i,j}t_{j} t_{i}\bos{O}_{ij} \mod t^{3}
$$
solves \deform\ modulo $(t)^3$.
The products $\bos{O}_{i}\cdot \bos{O}_{j}$of  two quantum observables, in general, are not
quantum observables - they do not belong to $\hbox{Ker }\bos{K}$ but
\eqn\demb{
\bos{K}\left(\bos{O}_{i}\cdot \bos{O}_{j}\right)
=-\hbar(-1)^{|\bos{O}_{i}|}\left(\bos{O}_{i}, \bos{O}_{j}\right).
}
Combining the above with \dema, we see that the existence of $\bos{O}_{ij}$ 
is equivalent to the ability of finding quantum correction to the products 
$\bos{O}_{a_1}\cdot \bos{O}_{a_2}$ as follows
$$
\bos{\pi}_{ij}:=\bos{O}_{i}\cdot \bos{O}_{j}-\hbar \bos{O}_{ij}
$$
such that $\bos{K}\bos{\pi}_{ij}=0$. 
Then  we might take  $\bos{\pi}_{ij}$ as a definition of $2$-point correlator
and its expectation value as $2$-point correlation function.
Assuming that $\bos{\Theta}$ can be extended to all orders, $\bos{\Theta}$ allows
quantum corrections to $n$-tuple products $\bos{O}_{i_{1}}\cdots \bos{O}_{i_{n}}$
of quantum observables to get
$n$-point correlators $\bos{\pi}_{i_1\cdots i_n}$ 
satisfying $\bos{K}\bos{\pi}_{i_1\cdots i_n}=0$
for all $n=2,3,\ldots$, simultaneously. 
Then all $n$-point correlation functions are determined algebraically 
after  fixing a $\Bbbk[[\hbar]]$-linear map
from the space of equivalence classes of quantum observables, i.e.,
$\left<\bos{\pi}_{i_1\cdots i_n}\right>$.
\end{example}

But there are serious problems in the above approach.
It is suffice to consider $2$-point correlators.
%
\begin{example}
Assume that $\bos{O}_{ij}$ solves \dema, 
so that $\bos{\pi}_{ij}:=\bos{O}_i\cdot \bos{O}_j-\hbar \bos{O}_{ij}$ 
satisfies $\bos{K}\bos{\pi}_{ij}=0$. 
Then, for
any ${\bf{X}}_{ij}$ satisfying $\bos{K}\bos{X}_{ij}=0$,
$\bos{O}^\pr_{ij}= \bos{O}_{ij} -  {\bf{X}}_{ij}$ also solves \dema\ so that
$\bos{\pi}^\pr_{ij}:=\bos{O}_i\cdot \bos{O}_j-\hbar \bos{O}^\pr_{ij}$ also satisfies 
$\bos{K}\bos{\pi}^\pr_{ij}=0$. 
Thus we obtain $\bos{\pi}^\pr_{ij}-\bos{\pi}_{ij} =\hbar {\bf{X}}_{ij}$
and 
$$
\left<\bos{\pi}^\pr_{ij}\right>-\left<\bos{\pi}_{ij}\right> 
=\hbar \left<{\bf{X}}_{ij}\right>
$$
as a consequence.
We may say that the two solutions $\bos{O}_{ij}$ and $\bos{O}^\pr_{ij}$ are equivalent 
if ${\bf{X}}_{ij}=\bos{K}\bos{\l}_{ij}$, leading to the same $2$-point correlation function.
But we can also choose  ${\bf{X}}_{ij}$  to be  an arbitrary  $\Bbbk[[\hbar]]$-linear combinations 
of non-trivial quantum observables. 
Then  the value
$\left<{\bf{X}}_{ij}\right>$ can be anything, meaning that 
$\hbar$-dependent part of $2$-point quantum correlation function (thus the quantum correction)
has essentially  {\it zero information}. 
\end{example}

The story for  higher-point correlation functions are more subtle, 
and we need a systematic way 
to remove the similarly irrelevant information.  
Thus we somehow need to look for not  any  solution 
but for a certain special  solution 
to the equation \deform, equivalently, the BV quantum master equation \introaa,
with a justification of such a choice. 

\begin{remark}
Let's call, for this paper, a graded Lie algebra $\mg$ with degree $1$ bracket $[\bullet,\bullet]$ 
and degree $1$ differential $d$, which squares to zero and a graded derivation 
of the bracket, a DGLA. The Maurer-Cartan (MC) equation of the DGLA is the equation
$d\g +\Fr{1}{2}[\g,\g]=0$ for $\g \in \mg^{0}$. 
The MC equation can be naturally
generalized  such that suitably parametrized solution can be considered
by tensoring the DGLA with appropriate  parameter
algebra. Study of MC equation of DGLA is equivalent to study of $L_{\infty}$-morphisms
up to homotopy from the cohomology of the cochain complex $(\mg, d)$ to the DGLA.
In the BV quantization procedure, there are several DGLAs with their MC equations being involved.
\begin{enumerate}

\item From a classical action $S_{cl}$ to classical BV master action $S=S_{cl}+\G$:
Consider the classical BV master equation
$$
\Fr{1}{2}(S,S)\equiv Q_{cl}\G +\Fr{1}{2}(\G,\G)=0.
$$
Here $\left(\sC, Q_{cl}, (\bullet,\bullet)\right)$ is a DGLA over $\R$,
where $Q_{cl}:=(S_{cl},\bullet)$, which satisfies $Q_{cl}^{2}=0$ since $(S_{cl},S_{cl})=0$,
and its MC equation is the classical BV master equation.

\item Classical BV master action $S$ and classical observables: Let $Q:=(S,\bullet)$,
which satisfies $Q^{2}=0$ since $(S,S)=0$. Then $\left(\sC, Q, (\bullet,\bullet)\right)$
is a DGLA over $\bos{R}$. The MC equation is then
$$
Q\Theta +\Fr{1}{2}\left(\Theta,\Theta\right)=0,
$$
which an infinitesimal solution $O$ is a classical observable.

\item From the BV algebra $\left(\sC, \Delta,\,\cdot\,\right)$ with the associated BV bracket 
$\left(\bullet,\bullet\right)$,
the triple $\left(\sC, \Delta, (\bullet,\bullet)\right)$ is also a DGLA over $\bos{R}$. 
But, we have no use
of this DGLA.

\item From the above data one can construct the triple 
$\left(\sC[[\hbar]], -\hbar \Delta, (\bullet,\bullet)\right)$,
 which is a DGLA over $\R[\hbar]]$. Its MC equation 
$$
-\hbar\Delta \bos{S} +\Fr{1}{2}\left(\bos{S},\bos{S}\right)=0
$$
is, then, the quantum BV master equation \introa.

\item 
Let $\bos{S}$ be a quantum BV master action. Then the operator
$\bos{K}:=-\hbar \Delta + \left(\bos{S},\bullet\right)$ satisfies
$\bos{K}^{2}=0$ and the triple $\big(\sC[[\hbar]],\bos{K},(\bullet,\bullet)\big)$
is a DGLA over $\R[[\hbar]]$. Its MC equation 
$$
\bos{K}\bos{\Theta}+\Fr{1}{2}\left(\bos{\Theta},\bos{\Theta}\right)=0
$$
is
the quantum BV master equation \deform\ for the deformation of quantum BV master action
$\bos{S}$ to $\bos{S} +\bos{\Theta}$, 
which an infinitesimal solution $\bos{O}$ is a quantum observable.

\end{enumerate}

An implication of our demonstration is that the various DGLAs 
should be regarded as secondary notions in quantization procedure.
Satisfying  MC equation of the relevant DGLA at each stage of quantization procedure 
should be a consequence but not the goal. 
\end{remark}

There is an alternative way to approach the problem. Instead of trying to solve
the equation \introaa, we may reduce the problem to certain extension problem
of classical observables to quantum observables, which procedure automatically
gives a special solution to \introaa. 
\begin{example}
Let $\{O_{i}\}$, $i\in I$ be a certain
set of classical observables, $Q O_{i}=0$,  in which we are interested.
Let's also assume that those classical observables are extendable to quantum observables, say 
$\{\bos{O}_{i}\}$, $i\in I$ and we know their expectation values $\{\left<\bos{O}_{i}\right>\}$.
Now we want to figure out $2$-point quantum correlation functions
between them. To simplify example we assume that the ghost numbers of 
$O_{i}$ are all zero, i.e., $O_{i}\in \sC^{0}$.
Let's assume that 
$\{O_{i}\}$ somehow form a closed algebra in the sense that there are
identities like
\eqn\brainy{
O_{i}\cdot O_{j}=\sum_{k\in I} m_{ij}{}^{k}O_{k} + Q x_{ij}
}
where $x_{ij}\in \sC^{-1}$. Then it can be easily shown that
the structure constants $m_{ij}{}^{k}$ depend only on the $Q$-cohomology classes 
of $O_{\ell}$, $\ell \in I$. Note that $x_{ij}$ above are not uniquely determined but
only up to $\Ker Q$. Note also that the commutativity 
$O_{i}\cdot O_{j}=O_{j}\cdot O_{i}$ of the product $\cdot$ implies
that $m_{ij}{}^{k}=m_{ji}{}^{k}$ as well as 
$Q x_{[ij]}=0$, where $x_{[ij]}:=\Fr{1}{2}\left(x_{ij}-x_{ji}\right)$.  
Thus the term $Q x_{ij}$ in \brainy\ can be replaced with $Q \l_{ij}$, where
$\l_{ij}=\Fr{1}{2}(x_{ij}+ x_{ji})$. 
Then the  expression
$\bos{L}_{ij}=\bos{O}_{i}\cdot \bos{O}_{j}-\sum_{k\in I} m_{ij}^{k}\bos{O}_{k} 
- \bos{K} {\l}_{ij}$ obviously satisfy $\bos{L}_{ij}=\bos{L}_{ji}$
and is divisible by $\hbar$, since $\bos{L}_{ij}\bigr|_{\hbar=0}=0$ by definition.
So we can {\it define} $\bos{O}_{ij} \in \sC[[\hbar]]^{0}$ by the following formula
\eqn\nobrainer{
\hbar \bos{O}_{ij}:= \bos{O}_{i}\cdot \bos{O}_{j}-\sum_{k\in I} m_{ij}{}^{k}\bos{O}_{k} 
- \bos{K} \bos{\l}_{ij},
}
which automatically gives us $2$-point quantum correlators 
$\bos{\pi}_{ij}:=\bos{O}_{i}\cdot \bos{O}_{j}-\hbar\bos{O}_{ij}$.
It also follows that
$\bos{\pi}_{ij}=\sum_{k\in I} m_{ij}{}^{k}\bos{O}_{k} +\bos{K} \bos{\l}_{ij}$
and $2$-points correlation functions 
$\left<\bos{\pi}_{ij}\right>= \sum_{k\in I} m_{ij}{}^{k}\left<\bos{O}_{k}\right>$
determined by the classical data $m_{ij}{}^{k}$ 
and the expectation values $\{\left<\bos{O}_{k}\right>\}_{k\in I}$.

Note that the equation \deform\ has played no roles in the above. Now let's
apply $\bos{K}$ to the both hand sides of \nobrainer\ to obtain
$\hbar\bos{K} \bos{O}_{ij}:= -\hbar\left(\bos{O}_{i}, \bos{O}_{j}\right)$, that is
\eqn\dummer{
\bos{K} \bos{O}_{ij}+\left(\bos{O}_{i}, \bos{O}_{j}\right)=0
}
Thus $\bos{\Theta}=\sum_{i} t_{i}\bos{O}_{i} +\Fr{1}{2}\sum_{ij}t^{j}t^{i}\bos{O}_{ij}\mod t^{3}$
solves the equation \deform\ modulo $t^{3}$. It is clear that not every
solution of the equation \deform\ modulo $t^{3}$ satisfies \nobrainer.
We also emphasize that the three assumptions that we have made
to have  $2$-point quantum correlators among the set $\{\bos{O}_{i}\}$, $i\in I$
of quantum observables are not sufficient conditions to have 
$3$-point quantum correlators among the set $\{\bos{O}_{i}\}$, $i\in I$.
Similarly the ability to define $n$-point quantum correlators among certain
set of quantum observables does not imply that we have $(n+1)$-point
quantum correlators among its members. This is just in the nature of quantum correlations.
\end{example}

The proper setting for the above  turns out to be the notion 
of BV QFT algebra, which produces the DGLA in \deform\ as a  {\it descendant } notion.
We shall, then, introduce new notion of master equation of BV QFT algebra,
which solution is  automatically the desired special solution to 
the MC equation \deform\ (equivalently, 
to the BV quantum master equation \introaa), while 
not every solution of \deform\ is descended from the new quantum master equation.
The corresponding obstruction theory should  directly deal with 
obstructions to those quantum corrections  
to all order products of quantum observables.  
We shall see that such obstruction theory is completely determined 
by obstruction  to extending classical observables to quantum observables.

\subsection{Summary of This Paper}

\subsubsection{BV QFT algebra and its descendant.}

Fix a ground field $\Bbbk$ of characteristic zero, $\Bbbk=\R$ for example.
Let $(\sC,\,\cdot\,)$ be a ${Z}$-graded super-commutative and associative unital
$\Bbbk$-algebra with the multiplication $\cdot$.
Let
$$
\sC[[\hbar]]=\left\{\sum_{n\geq 0} \hbar^n a^{(n)} \big| a^{(n)}\in \sC\right\}.
$$
Then $\sC[[\hbar]]$ has the canonical multiplication induced from $\sC$, which will be denoted
by the same symbol $\cdot$. Thus $(\sC[[\hbar]],\,\cdot\,)$ is a ${Z}$-graded 
super-commutative and associative unital $\Bbbk[[\hbar]]$-algebra. 
In general a $\Bbbk$-multilinear map of $\sC$ into $\sC$  canonically induces 
a $\Bbbk[[\hbar]]$ multilinear map of $\sC[[\hbar]]$ into $\sC[[\hbar]]$,
and we shall not distinguish them.
Projection of any structure parametrized by $\hbar$ on $\sC[[\hbar]]$ 
to $\sC$ will be called taking classical limit.

\begin{definition}
\label{Def1}
Let  $\bos{K}=Q +\hbar K^{(1)} +\hbar^{2}K^{(2)}+\hbar^3 K^{(3)}+\cdots$ 
be a sequence of $\Bbbk$-linear maps,
parametrized by $\hbar$, of ghost number $1$
 on $\sC$ into $\sC$
satisfying $\bos{K}^2=0$ and $\bos{K}1=0$.
Then the triple
$$
\big(\sC[[\hbar]], \bos{K}, \hbox{ }\cdot\hbox{ }\big)
$$ 
is  a BV QFT algebra if the failure of $\bos{K}$ being  a derivation of the product
$\cdot$ is {\it divisible} by $\hbar$ and the binary operation measuring
the failure is a derivation of the product.
\end{definition}
It follows that the classical limit $Q$ of $\bos{K}$ is a derivation of the
product. Thus, the classical limit 
$$
\big(\sC, Q,\,\cdot\,\big)
$$
of the BV QFT algebra  is a  super-commutative associative unital
differential graded algebra (CDGA) over $\Bbbk$. 

On $\sC[[\hbar]]$, being freely generated by $\sC$, there is natural automorphism 
$\bos{g}=1+ g^{(1)}\hbar +g^{(2)}\hbar^{2}+ \cdots$, where $g^{(\ell)}$ are ghost number 
preserving $\Bbbk$-linear maps on $\sC$ into $\sC$.  Such an automorphism
will acts on both the unary operation $\bos{K}$ and the binary operation $\cdot$
as 
$
\bos{K}\rightarrow \bos{K}^{\pr}$ such that $\bos{K}^{\pr}=\bos{g}\bos{K} \bos{g}^{-1}
$
and $\cdot \rightarrow \cdot^{\pr}$ such that 
$\cdot^{\pr}= \bos{g}\left(\bos{g}^{-1}\cdot \bos{g}^{-1}\right)$.
It is trivial that $\left(\sC[[\hbar]],\bos{K}^{\pr},\,\cdot^{\pr}\,\right)$ is also a BV
QFT algebra. Note that such automorphisms  fix the classical limit, i.e.,
$Q=Q^{\pr}:=\bos{K}^{\pr}\bigr|_{\hbar=0}$ and $a\cdot^{\pr}b= a\cdot b$
for $a,b\in \sC$.
Such an automorphism should be regarded as ``gauge symmetry'' of ``underlying QFT'', 
so that
the resulting two BV QFT algebras $\big(\sC[[\hbar]], \bos{K},\,\cdot\,\big)$ 
and $\big(\sC[[\hbar]], \bos{K}^{\pr},\,\cdot^{\pr}\,\big)$ should be regarded as equivalent.
Thus we are lead to study BV QFT algebra modulo the ``gauge symmetry'', while our 
algebraic path integral shall be ``gauge invariant''.

Let $\bos{\left(\hbox{ },\hbox{ }\right)}:\sC[[\hbar]]^{k_{1}}\otimes\sC[[\hbar]]^{k_{2}}
\longrightarrow \sC[[\hbar]]^{k_{1}+k_{2}-1}$ be the binary operation divided by $\hbar$;
\eqn\djo{
-\hbar (-1)^{|\bos{a}|} \left(\bos{a},\bos{b}\right):= \bos{K}\left(\bos{a}\cdot \bos{b} \right)
- \bos{K a}\cdot \bos{b} -(-1)^{|\a|}\bos{a}\cdot \bos{Kb}. 
}
Then the triple $\left(\sC[[\hbar]], \bos{K},(\hbox{ },\hbox{ })\right)$, after forgetting
the product, is a differential graded Lie algebra (DGLA) over $\Bbbk[[\hbar]]$. 
We emphasis that the bracket $(\hbox{ },\hbox{ })$ is a purely secondary notion 
in the definition of  BV QFT algebra. 
Thus,  the triple $\bigl(\sC[[\hbar]], \bos{K},(\hbox{ },\hbox{ })\bigr)$ 
shall be called the {\it descendant} DGLA to the BV QFT algebra
$(\sC[[\hbar]], \bos{K}, \hbox{ }\cdot\hbox{ })$. 
The classical limit 
 $\bigl(\sC, Q,\left(\hbox{ },\hbox{ }\right)\bigr)$
of the  descendant DGLA is a DGLA over $\Bbbk$
(we are abusing the notations by not distinguishing the bracket $(\hbox{ }, \hbox{ })$ \djo\
with its classical limit).
Note that not every DGLA  over $\Bbbk$ is a classical limit of the descendant DGLA of
a BV QFT algebra.  We also emphasis that the  DGLA  $
\bigl(\sC, Q,\left(\hbox{ },\hbox{ }\right)\bigr)$ has a quantum origin, 
however the secondary notion as it is. 
Under a gauge symmetry of BV QFT algebra the bracket in its descendant DGLA changes as 
$\left(\phantom{\bos{a}}, \phantom{\bos{b}}\right)^{\pr}
= \bos{g}\left(\bos{g}^{-1}, \bos{g}^{-1}\right)$,  
while its classical limit remains fixed,

\begin{remark}
A typical example of BV QFT algebra is an output of  a successful BV quantization, 
which procedure  has been briefly summarized earlier. 
Let  ${\bf{S}}=S +\hbar S^{(1)}+\cdots$ be the resulting
BV quantum master action. Then
$\big(\sC[[\hbar]], \bos{K}:=-\hbar \Delta + ({\bf{S}}, \hbox{ }), \hbox{ }\cdot\hbox{ }\big)$
is a BV QFT algebra with  with the descendant DGLA 
$\big(\sC[[\hbar]], \bos{K},(\bullet,\bullet)\big)$. The classical limit of the BV QFT
algebra is $\big(\sC, Q:=(S,\hbox{ }), \,\cdot\,\big)$. 
\end{remark}

\begin{example}
A better example could be the bare data of a classical field theory,  a classical action
$S_{cl}$ which is certain function on the space $\mL_{cl}$ of classical fields 
with zero ghost number. Assuming
that some artist can always supply a BV operator $\Delta_{cl}$ to the algebra $(\sC_{cl},\;\cdot\;)$
of functions on $T^*[-1]\mL_{cl}$, it is automatic that $\hbar^2\Delta_{cl} e^{-S_{cl}/\hbar}=0$.
Then every classical field theory to quantize gives us a BV QFT algebra
$\big(\sC_{cl}[[\hbar]], \bos{K}_{cl}= -\hbar\Delta_{cl} +\left(S_{cl},\hbox{ }\right),\;\;\cdot\;\;\big)$.
This setting shall be a starting point of the $4$-th paper in this series.
\end{example}

\subsubsection{Observables and  expectation values}

We denote cohomology of the cochain complex $\big(\sC,Q\big)$ over $\Bbbk$ by $H$, 
which is a $\Z$-graded 
$\Bbbk$-module (a graded vector space over $\Bbbk$). 
Following physics terminology we call an element $O \in \sC$ a 
{\it classical observable} if $Q O=0$. 
Two classical observable $O$ and $O^{\pr}$ are
(classical) physically equivalent if there is some $\l$ such that $O^{\pr}-O =Q\l$.
Thus a classical observable $O$ is a representative of its cohomology class $[O]$ in $H$, and
two classical observables are physically equivalent if and only if they are representatives of the
same cohomology class.  
It follows that classical observables can be organized by 
a $\Bbbk$-linear map $f:H\longrightarrow \sC$ preserving the ghost number 
such that $Q f=0$ and $\left[f\big([O]\big)\right]=[O]$.
Such a map $f$ is not unique since any map $f^\pr = f + Q s$
also satisfy $\left[f^\pr\big([O]\big)\right]=[O]$
for an arbitrary $\Bbbk$-linear map $s:H\longrightarrow \sC$  of ghost number $-1$,
and classical physics must not distinguish them.

Following  physics terminology we might say that a classical observable $O \in \sC^{|O|}$
is quantized to a {\it quantum observable} if there is an $\bos{O}\in \sC[[\hbar]]^{|O|}$ 
such that   $\bos{O}\bigr|_{\hbar=0} = O$ and $\bos{K}\bos{O}=0$. On the other hand such quantized observable
$\bos{O}$ is supposed to be  (quantum) physically equivalent to
$\bos{O}^\pr= \bos{O}+ \bos{K}\bos{\l}$ for any $\bos{\l}\in \sC[[\hbar]]^{|O|-1}$,
which classical limit $O^\pr$ is, in general, differ to $O$ by a $Q$-exact term.   
Thus quantization of a classical observable $O$ should be a statement about its cohomology
class $[O]$.

Sorting out those terminologies, classical observables are organized by a cochain
map $f$ from the cohomology $H$, regarded as a cochain complex with zero differential, 
to the cochain complex $(\sC,Q)$ 
which induces the identity map on $H$ and is defined up to homotopy. And, quantization of
classical observables is an extension of $f$ to a sequence 
$\bos{f}=f +\hbar f^{(1)} +\hbar f^{(2)}+\cdots$ of $\Bbbk$-linear and ghost number preserving 
maps, parametrized by $\hbar$, on $H$ into $\sC$ such that $\bos{f}$ satisfy $\bos{K}\bos{f}=0$.
Such an extension is not always possible and should be defined up to "quantum homotopy",
The obstruction for extending $f$ to the whole
sequence  $f, f^{(1)}, f^{(2)},\cdots$ 
as well as every possible ambiguity of the procedure
is summarized by
our first  theorem.

\begin{theorem}
\label{Propos1}
Let $f$ be a cochain map from $(H, 0)$ to $(\sC, Q)$ which induces
the identity map on the cohomology $H$.
On $H[[\hbar]]$,
{ modulo its natural automorphism},
\begin{enumerate}

\item
there is an { unique} $\Bbbk[[\hbar]]$-linear map
$\bos{\k}=\hbar \k^{(1)}+\hbar^{2}\k^{(2)}+\cdots$ of ghost number $1$
into itself,
which is induced from a sequence $0,\k^{(1)},\k^{(2)},\cdots$ of
$\Bbbk$-linear maps on $H$ into $H$,
satisfying $\bos{\k}^{2}=0$ and $\bos{\k}\bigr|_{\hbar=0}=0$,

\medskip
\item
there is a 
$\Bbbk[[\hbar]]$-linear map
$\bos{f}=f+\hbar f^{(1)}+\hbar^{2}f^{(2)}+\cdots$
of ghost number $0$
into $\sC[[\hbar]]$, 
which is induced from a sequence $f,f^{(1)},f^{(2)},\cdots$ of
$\Bbbk$-linear maps on $H$ into $\sC$, which satisfies
$$
\bos{K}\,\bos{f}=\bos{f} \,\bos{\k},
$$
and is defined
up to quantum  homotopy;
$$
\bos{f}\,\bos{\sim}\,\bos{f}^{\pr}=\bos{f} + \bos{K}\,\bos{s} +\bos{s}\,\bos{\k},
$$
where $\,\bos{s}=s+\hbar {s}^{(1)}+\hbar^{2}s^{(2)}+\cdots$
is an arbitrary sequence of
$\,\,\Bbbk$-linear maps 
of ghost number $-1$
parametrized by $\hbar$ on $H$ into $\sC$. 
\end{enumerate}

\end{theorem}
We shall sometimes refer a map $\bos{f}$ with the above stated properties a quantum
extension map. An essential content of the above theorem concerning obstruction is
that a classical observable $O$ is extendable to a quantum observable 
if and only if its cohomology class $[O]$ is annihilated by $\bos{\k}$, i.e.,
$\k^{[\ell]}\left(\left[O\right]\right)=0$ for all 
$\ell=1,2,3,\cdots$. 

\begin{remark}
The classical limit of quantum homotopy equivalence
$\bos{f}\,\bos{\sim}\,\bos{f}^{\pr}=\bos{f} + \bos{K}\,\bos{s} +\bos{s}\,\bos{\k}$
is $f\,\sim\,f^{\pr}=f + Q \,s$ since  the classical limit of $\bos{\k}$ is zero.
Thus it reduces to homotopy equivalence of cochain maps from $(H,0)$ to $(\sC, Q)$.
\end{remark}

\begin{remark}
An automorphism on $H[[\hbar]]$ is an arbitrary sequence 
$\bos{\xi}=1+\hbar \xi^{(1)}+\hbar^{2}\xi^{(2)}+\cdots$ of $\Bbbk$-linear
maps with ghost number $0$  parametrized by $\hbar$  on $H$ into itself 
satisfying
$\bos{\xi}\bigr|_{\hbar=0}=1$. Such an automorphism fix $H$ and send 
$\bos{\k}$ to $\bos{\xi}^{-1}\,\bos{\k}\,\bos{\xi}$
and $\bos{f}$ to $\bos{f}\,\bos{\xi}$. Note that every automorphism
fix $f$ as well as $\k^{(1)}$, since $\bos{\k}\bigl|_{\hbar=0}=0$. 
\end{remark}


Now we have a room to accommodate "Feynman Path Integral".

\begin{definition}
A BV QFT with ghost number anomaly $N$ is a 
BV QFT algebra 
with a sequence of $\Bbbk$-linear maps
$\bos{c}:= c^{(0)}+\hbar c^{(1)}+\hbar^{2}c^{(2)}+
\cdots$, parametrized by $\hbar$,
of ghost number $-N$
on $\sC$ into $\Bbbk$ which  satisfies $\bos{c}\,\bos{K}=0$ and 
is defined up to quantum homotopy;
$$
\bos{c}\,\bos{\sim}\,\bos{c}^{\pr}=\bos{c} + \bos{r}\,\bos{K},
$$
where $\bos{r}=r^{(0)}+\hbar {r}^{(1)}+\hbar^{2}r^{(2)}+\cdots$
is an arbitrary sequence of
$\Bbbk$-linear maps 
of ghost number $-N-1$ 
parametrized by $\hbar$ on $\sC$ into $\Bbbk$. 

\end{definition}
\begin{remark}
Note again that the ghost number of $\Bbbk$ (and $\Bbbk[[\hbar]]$) is concentrated to
zero. So the sequence $c^{(0)},c^{(1)}, c^{(2)},\cdots$ of $\Bbbk$-linear maps
should be zero maps on $\sC^{n}$ for $n\neq N$.
\end{remark}
\begin{remark}
The different choice of $\bos{c}$ within the same quantum
homotopy class is a realization of the different choice of gauge fixing.
\end{remark} 

We recall that our first theorem give
a sequence $\bos{f}:=f +\hbar f^{(1)}+\hbar^{2} f^{(2)}+\cdots$ of $\Bbbk$-linear maps
parametrized by $\hbar$ on $H$ into $\sC$ defined up to 
quantum  homotopy satisfying $\bos{K}\,\bos{f}= \bos{f}\,\bos{\k}$.
We can compose the map $\bos{f}$, regarded as a $\Bbbk[[\hbar]]$-linear map
on $H[[\hbar]]=H\otimes_{\Bbbk}\Bbbk[[\hbar]]$ into $\sC[[\hbar]]$, with
the map $\bos{c}:= c^{(0)}+\hbar c^{(1)}+\hbar^{2}c^{(2)}+
\cdots$, regarded as a $\Bbbk[[\hbar]]$-linear map
on $\sC[[\hbar]]$ into $\Bbbk[[\hbar]]$, 
to obtain a sequence
$$
\bos{\iota}:=\bos{c}\,\bos{f}=\iota^{(0)}+\hbar \iota^{(1)}+\hbar^{2}\iota^{(2)}+\cdots
$$ 
of $\Bbbk$-linear maps parametrized by $\hbar$ on $H$ into $\Bbbk$  such that
$$
\iota^{(n)}=\sum_{\ell=0}^{n}c^{(n-\ell)} f^{(\ell)}, \quad n=0,1,2,\cdots.
$$
The ambiguity of $\bos{\iota}$ due to the ambiguities of $\bos{f}$ and $\bos{c}$
up to quantum homotopy, $\bos{f}\,\bos{\sim}\,\bos{f}^{\pr}$ and 
$\bos{c}\,\bos{\sim}\,\bos{c}^{\pr}$,
is
$$
\bos{\iota}^{\pr}-\bos{\iota}\equiv \bos{c}^{\pr}\bos{f}^{\pr} -\bos{c}\,\bos{f}
= \left(\bos{c}\,\bos{s}+\bos{r}\,\bos{f} +\bos{r}\,\bos{K}\,\bos{s}\right)\bos{\k}.
$$

\begin{remark}
An automorphism $\bos{g}$ on $\sC[[\hbar]]$ sends $\bos{f}$ to
$\bos{g}\,\bos{f}$ and $\bos{c}$ to $\bos{c}\,\bos{g}^{-1}$,
since $\bos{f}$ and $\bos{c}$ are $\Bbbk[[\hbar]]$-linear maps
to $\sC[[\hbar]]$ and from $\sC[[\hbar]]$, respectively.
Thus $\bos{\iota}=\bos{c}\bos{f}$ is invariant under the automorphism of BV QFT algebra.
\end{remark}   

We  recall that a classical observable $O$ is extendable to 
a quantum observable $\bos{O}$
if and only if $\bos{\k}\left(\left[O\right]\right)=0$ and, then,
$\bos{\iota}\left(\left[O\right]\right)=\bos{\iota}^{\pr}\left(\left[O\right]\right)$.
By the way, it is the cohomology class of classical observable 
that is observable to a classical observer.  Also there is no genuine classical observable so
that every classical observation must be classical limit of quantum observation.
So we can omit the decorations ``classical'' and ``quantum'' and define an 
and their expectation values:
\begin{definition}
An observable $\mo$ is an element of the cohomology $H$ of the complex $(\sC,Q)$
satisfying $\k^{(n)}\left(\mo\right)=0$ for all $n=1,2,3,\cdots$. An element of $H$ which
is not an observable shall be called an invisible.
The expectation value
an observable $\mo$ is 
$$
\bos{\iota}
\left(\mo\right)
=\sum_{n=0}^{\infty}\hbar^{n}\sum_{\ell=0}^{n}c^{(n-\ell)}\left(f^{(\ell)}\left(\phi\right)\right),
$$
which  is a  quantum  homotopy invariant as well as invariant under the automorphism of BV
QFT algebra.
\end{definition}

\begin{remark}
Being understood, we may continue to use the notation
$\left<\bos{O}\right>$ for the expectation value of quantum observable
$\bos{O}$ if $\bos{O}=\bos{f}(\mo)$ instead of $\bos{\iota}(\mo)$.
The composition $\bos{\iota}=\bos{c}\circ\bos{f}$ is our take of 
{\it $\bos{\iota}$nt\'egrale de 
$\bos{c}$hemin de $\bos{f}$eynman}.
\end{remark}

\subsubsection{Quantum master  equations and quantum correlation functions: a case study.}

Fix a BV QFT algebra  $\big(\sC[[\hbar]], \bos{K}, \hbox{ }\cdot\hbox{ }\big)$ 
and its descendant  DGLA $\big(\sC[[\hbar]], \bos{K},(\bullet,\bullet)\big)$ with classical limits 
$\big(\sC, Q, \hbox{ }\cdot\hbox{ }\big)$ and $\big(\sC, {Q},(\bullet,\bullet)\big)$, respectively. 
 Let $H$ denote the cohomology group of the classical complex $(\sC, Q)$.
The purpose of this section is to study quantum correlations specialized to a class of BV QFTs that
$\bos{\k}=0$ on $H$ identically so that we don't need to deal with invisibles.  We shall also assume that
$H$ is finite dimensional for each ghost numbers for the sake of simplicity. Those assumptions
shall allow us to describe every quantum correlation function, thus exact solution of a BV QFT.

From the assumption that $\bos{\k}=0$ and theorem \ref{Propos1}, we have a sequence 
$\bos{f}=f +\hbar f^{(1)}+\cdots$ of 
$\Bbbk$-linear maps on $H$ into $\sC$
of ghost number zero
such that $\bos{K} \bos{f}=0$, which classical limit  
$f=\bos{f}\bigr|_{\hbar=0}$ is 
a quasi-isomorphism of complexes $f:(H, 0)\longrightarrow (\sC, Q)$,
which induces the identity map on $H$. 
From the condition $\bos{K}1=0$, thus $Q1=0$, in the definition of BV QFT algebra,
there is a distinguished element $e \in H$ corresponding to the cohomology class
$[1]$ of the unit $1$ in $(\sC,\;\cdot\;)$.  On $H$ there is also 
an unique binary product $m_2:H\otimes H \longrightarrow H$ of ghost number $0$
induced from the product in the CDGA $(\sC, Q,\;\cdot\;)$;
let $a,b \in H$ then  $m_2(a,b):=\left[f(a)\cdot f(b)\right]$ 
which is an homotopy invariant since $Q$ is a derivation of the product $\cdot$,
and $m_2(e, b)=m_2(b,e)=b$, such that $(H,0, m_2)$ is a CDGA with unit $e$ with zero differential.
It is natural to fix $f$ and $\bos{f}$ such that $f(e)=1$ and $\bos{f}(e)=1$.

It is convenient to  fix a basis $\{e_\a\}$ of $H$ such that one of its component,
say $e_0$ is the distinguished element.
Let $t_H=\{t^\a\}$ be the dual basis (basis of $H^*$) such that $|t^\a| + |e_\a|=0$, which is 
a coordinates system on $H$ with a distinguished coordinate $t^0$.
We denote $\{\rd_\a =\rd/\rd t^\a\}$ be the corresponding formal partial derivatives action
on $\Bbbk[[t_H]]$ a derivations. 
The product $m_{2}$ is specified by structure constants 
$m_{\a\b}{}^\g\in \Bbbk$ such that $m_2(e_\a, e_\b)=m_{\a\b}{}^\g e_\g$ and
$m_{0\b}{}^\g =\d_\b{}^\g$.
The binary multiplication $m_2$ on $H$ can be identified with a derivation
$m_2^\sharp= \Fr{1}{2}t^{\a_2}t^{\a_1}m_{\a_1\a_2}{}^\g \Fr{\rd}{\rd t^\g}$ 
on $\Bbbk[[t_H]]$ with ghost number $0$, where we are using and going to use 
Einstein summation convention that a repeated upper and lower index  is summed over.
Any multilinear map $m_n:S^n H\rightarrow H$ of ghost number $0$ is similarly identified
with a derivation $m_n^\sharp$ on $\Bbbk[[t_H]]$ of ghost number $0$. 
We shall also use notations $\bos{f}(e_\a) =\bos{O}_\a$ such that $\bos{K}\bos{O}_\a =0$.

Now the triple $\big(\Bbbk[[t_H]]\otimes\sC[[\hbar]], \bos{K}, \;\cdot\;\big)$ is a BV QFT
algebra, where $\bos{K}$ and $\cdot$ are the shorthand notions 
for $1\otimes \bos{K}$ and 
$
(a\otimes \bos{x})\cdot (b\otimes \bos{y})= (-1)^{|\bos{x}||{b}|}ab\otimes \bos{x}\bos{y}
$
for $a,b \in \Bbbk[[t_H]]$ and $\bos{\a},\bos{\b}\in \sC[[\hbar]]$, respectively.
We denote its descendant algebra by $\big(\Bbbk[[t_H]]\otimes\sC[[\hbar]], \bos{K}, (\hbox{ },\hbox{ })\big)$,
where 
$
\big( a\otimes \bos{x}, b\otimes \bos{y}\big)
= (-1)^{(|\bos{x}|+1)|{b}|}ab\otimes \big(\bos{x},\bos{y}\big)$. 
The symbol $\otimes$ means tensor (or completed tensor) product, which shall be omitted 
whenever possible.
Then, the following theorem contains
the complete information of quantum correlation functions;
\begin{theorem}
On $H$
there is a sequence $m_2, m_3, m_4, \cdots$ of multilinear multiplications 
$m_n:S^n H\rightarrow H$ of ghost number $0$ such that
$m_2(e_0, e_\a)= e_\a$ and $m_n(e_0,e_{\a_2},\cdots, e_{\a_{n-1}})=0$ for all $n=3,4,5,\cdots$.
And, 
there is a family of BV QFTs specified by
$$
\bos{\Theta}=\bos{\Theta}_1 + \bos{\Theta}_2+\bos{\Theta}_3+\cdots
\in \big(\Bbbk[[t_H]]\otimes \sC[[\hbar]]\big)^0,
$$
where $\bos{\Theta}_1= t^\a\bos{f}(e_\a)=t^\a\bos{O}_\a$ and
$\;\bos{\Theta}_n =\Fr{1}{n!}t^{\a_n}\cdots t^{\a_1} \bos{O}_{\a_1\cdots\a_n}
 \in \big(S^n(H^*)\otimes \sC[[\hbar]]\big)^0$,
satisfying

1. quantum master equation:
$$
\eqalign{
0=&\bos{K}\bos{\Theta}_1
,\cr
\hbar \bos{\Theta}_2  =&\Fr{1}{2}\bos{\Theta}_1\cdot \bos{\Theta}_1
- m_2^\sharp \bos{\Theta}_1 - \bos{K} \La_2
,\cr
\hbar \bos{\Theta}_3  =&
\Fr{2}{3}\bos{\Theta}_1\cdot \bos{\Theta}_2
-\Fr{1}{3}m_2^\sharp\bos{\Theta}_2
-\Fr{1}{3}\big(\bos{\Theta}_1,\La_2\big) 
-m_3^\sharp\bos{\Theta}_1
-\bos{K}\La_3
,\cr
\vdots\;&
\cr
\hbar\bos{\Theta}_{n} =& 
\sum_{k=1}^{n-1}\Fr{ k(n-k)}{n(n-1)} \bos{\Theta}_{k}\cdot  \bos{\Theta}_{n-k}
-\sum_{k=2}^{n-1}\Fr{k(k-1)}{n(n-1)}\left(  m_{k}^\sharp \bos{\Theta}_{n-k+1}
+ \big(\bos{\Theta}_{n-k},\La_{k}\big)\right)
\cr
&
- m_n^\sharp \bos{\Theta}_1-\bos{K}\La_n
,\cr
\vdots\;&
}
$$
for some $\La_2,\La_3,\cdots \in \big(\Bbbk[[t_H]]\otimes
\sC\big)^{-1}$.

2. quantum identity: $\rd_0 \bos{\Theta}=1$.

3. quantum descendant equation
$$
\bos{K}\bos{\Theta} +\Fr{1}{2}\big(\bos{\Theta},\bos{\Theta}\big)=0,
$$
as  a consequence of quantum master equation.

\end{theorem}

\begin{remark}
The quantum master equation is better  understood as definition of $\bos{\Theta}$ order by order
in the word-length in $t_H$ from the quantum extension map $\bos{f}$. To begin with
$\bos{\Theta}_1:= t^a \bos{f}(e_\a)$, implying that $\bos{K}\bos{\Theta}_1=0$, implying that
$Q{\Theta}=0$, implying that $\Theta_1\cdot \Theta_1 \in \Ker Q$, inferring  that
$\Theta_1\cdot \Theta_1=m_2^\sharp \Theta_1 + Q \La_2$ for unique $m_2^\sharp$ and
for some $\La_2$ defined modulo $\Ker Q$, implying that 
$\Fr{1}{2}\bos{\Theta}_1\cdot \bos{\Theta}_1- m_2^\sharp \bos{\Theta}_1 - \bos{K} \La_2$ 
is divisible by $\hbar$, thus we {\it define} $\bos{\Theta}_2$
by the $2$-nd of quantum master equation  and infers that
$\bos{K}\bos{\Theta}_2 +\Fr{1}{2}\big(\bos{\Theta}_1,\bos{\Theta}_1)=0$.
Then it can be shown that  the expression $\Fr{2}{3}{\Theta}_1\cdot {\Theta}_2
-\Fr{1}{3}m_2^\sharp{\Theta}_2-\Fr{1}{3}\big({\Theta}_1,\La_2\big) \in \Ker Q$ 
such that it can be expressed
as $m_3^\sharp\Theta_1 + Q \La_3$ for unique $m_3^\sharp$ and
for some $\La_3$ defined modulo $\Ker Q$, implying that
the expression 
$\Fr{2}{3}\bos{\Theta}_1\cdot \bos{\Theta}_2
-\Fr{1}{3}m_2^\sharp\bos{\Theta}_2
-\Fr{1}{3}\big(\bos{\Theta}_1,\La_2\big) 
-m_3^\sharp\bos{\Theta}_1
-\bos{K}\La_3$ is divisible by $\hbar$,  thus we {\it define} $\bos{\Theta}_3$
by the $3$rd of quantum master equation and infers that
$\bos{K}\bos{\Theta}_3 +\big(\bos{\Theta}_1,\bos{\Theta}_2)=0$, et cetera, ad infinitum.

\end{remark}

One of the immediate consequence of our main result is that
that the classical limit $\Theta$ of $\bos{\Theta}$ is a solution
to the DGLA $\big(\sC, Q.\;\cdot\;\big)$ of very special kind.

\begin{corollary}
There exists a solution 
to  the classical descendant equation
\eqn\sdeform{
\eqalign{
Q{\Theta} +\Fr{1}{2}\big({\Theta},{\Theta}\big)=0,
\qquad
{\Theta} = t^\a {O}_\a 
+ \sum_{n=2}^\infty \Fr{1}{n!}t^{\a_n}\cdots t^{\a_1} {O}_{\a_1\cdots\a_n}
 \in \big(\Bbbk[[t_H]]\otimes\sC\big)^0
}
}
such that
\begin{enumerate}
\item (versality) the set of cohomology classes $[O_\a]$ 
form a basis of cohomology $H$ of the classical complex $(\sC, Q)$

\item (quantum coordinates) $\Theta$ is the classical limit of the solution to quantum master equation

\item (quantum identity) $\rd_0 \Theta=1$.
\end{enumerate}
\end{corollary}

It is a standard fact that there is a structure of minimal $L_\infty$-algebra (an $L_\infty$-algebra with
zero-differential) on cohomology of DGLA which is quasi-isomorphic as $L_\infty$-algebra, and
such a minimal $L_\infty$-structure is the obstruction to have versal solution to its Maurer-Cartan
equation. A DGLA is called formal if the minimal $L_\infty$-algebra on its cohomology is a graded
Lie algebra, and a formal DGLA has an associated smooth moduli space if and only if 
the graded Lie algebra on its cohomology is Abelian, 
i.e., the graded Lie bracket vanishes on $H$ \cite{Kontsevich}.
Now the versal solution we have is an $L_\infty$-quasi-isomorphism of very special kind 
since not every versal solution 
of \sdeform\ arises as the classical limit of solution of quantum master equation.
Hence we conclude that an anomaly-free BV QFT has its natural family parametrized by
a smooth moduli space $\CM$, a formal super-manifold, in quantum coordinates.

It shall be argued that the notion of quantum coordinates is a proper name  and generalization 
of that of flat or special coordinates
on moduli spaces of topological strings - in the context of Witten-Dijkgraaf-Verlinde-Verlinde (WDDV)
equation \cite{W,DVV} as well as the  mirror map by Candelas-de la Ossa-Green-Parkes 
\cite{COGP,Witten}. For the mathematical side, both the pioneering works of K.\ Saito 
and Barannikov-Kontsevich
on  flat structure
on moduli space of universal unfolding of simple singularities \cite{Saito}
and the flat coordinates in differential
BV algebra satisfying  a version of $\rd\!\bar\rd$-lemma   \cite{BK},
respectively 
are also examples of
quantum coordinates.

The solution $\bos{\Theta}$ of the quantum master equation shall be used to
define generating function of all
quantum correlators by the formula
$$
e^{-\Fr{\bos{\Theta}}{\hbar}}
=1 + \sum_{n=1}^\infty \Fr{1}{n!}\Fr{(-1)^n}{\hbar^{n}}\bos{\Theta}^n
=1+\sum_{n=1}^\infty \Fr{(-1)^n}{\hbar^n} \bos{\Omega}_n,
$$
where the sequence $\bos{\Omega}_1,\bos{\Omega}_2, \cdots$  is defined 
by matching the word-lengths in $t_H$:
$$
\bos{\Omega}_n=\Fr{1}{n!}t^{\a_n}\cdots t^{\a_1}\bos{\pi}_{\a_1\cdots \a_n}
$$
where $\bos{\pi}_{\a_1\cdots \a_n}\in \sC[[\hbar]]^{|\a_1|+\cdots+|\a_n|}$
with the classical limit
$$
\bos{\pi}_{\a_1\cdots \a_n}\big|_{\hbar=0}=O_{\a_1}\cdots O_{\a_n}.
$$
Note that $\bos{\Omega}_1=\bos{\Theta}_1=t^\a \bos{O}_\a$ generates $1$-point quantum correlators.
The quantum descendant equation, which is equivalent to
$\hbar^2 \bos{K}\,e^{-\bos{\Theta}/\hbar}=0$,
implies that 
$\bos{K}\bos{\Omega}_n=0$ for all $n=1,2,\ldots$, 
that $\bos{\Omega}_{n}$ generates $n$-point quantum correlators.

The following lemma, due to quantum master equation, relates generating functions of  
$n$-points quantum correlators for every $n$ to $\bos{\Theta}_1$;

\begin{lemma}
For every $n \geq 2$, we have
$$
\bos{\Omega}_n = \bos{p}^\sharp_n \bos{\Theta}_1 + \bos{K}\bos{x}_n
$$
where $\bos{p}^\sharp_2=m^\sharp_2$, $\bos{x}_2=\La_2$, and
$$
\eqalign{
\bos{p}_n^\sharp &= 
(-\hbar)^{n-2}m_n^\sharp
+\Fr{1}{n(n-1)}\sum_{k=2}^{n-1}(-\hbar)^{k-2}k(k-1) m_{k}^\sharp \bos{p}_{n+1-k}^\sharp
,\cr
\bos{x}_n &= 
(-\hbar)^{n-2}\La_n
+\Fr{1}{n(n-1)}\sum_{k=2}^{n-1}(-\hbar)^{k-2}k(k-1)
\Big(
m_{k}^\sharp \bos{x}_{n+1-k}
+\La_k\cdot\bos{\Omega}_{n-k}\Big).
}
$$
\end{lemma}
Hence, we conclude that 
$$
\left<\bos{\Omega}_n\right> 
= \bos{p}_n^\sharp\left<\bos{\Theta}_1\right>
=\Fr{1}{n!}t^{\a_n}\cdots t^{\a_1} \bos{p}_{\a_1\cdots\a_n}{}^\g\left<\bos{O}_\g\right>,
$$
for $n\geq 2$, while $\left<\bos{\Omega}_1\right>=t^\a\left<\bos{O}_\a\right>$. 

\begin{example}
The first few quantum correlators are
$$
\eqalign{
\bos{\Omega}_1&=\bos{\Theta}_1,\cr
\bos{\Omega}_2 &=\Fr{1}{2!}\bos{\Theta}_1^2 -\hbar \bos{\Theta}_2,\cr
\bos{\Omega}_3 &=\Fr{1}{3!}\bos{\Theta}_1^3 -\hbar \bos{\Theta}_1\bos{\Theta}_2+\hbar^2\bos{\Theta}_3,
\cr
\bos{\Omega}_4 &=\Fr{1}{4!}\bos{\Theta}_1^4 -\Fr{\hbar}{2}\bos{\Theta}_1^2\bos{\Theta}_2
+\hbar^2\left(\bos{\Theta}_1\bos{\Theta}_3+\Fr{1}{2}\bos{\Theta}_2^2\right)
-\hbar^3 \bos{\Theta}_4.
}
$$
and
$$
\eqalign{
\left<\bos{\Omega}_2\right> =& m_2^\sharp \left<\bos{\Theta}_1\right>,\cr
\left<\bos{\Omega}_3\right>
=&\left(\Fr{1}{3}m_2^\sharp m_2^\sharp -\hbar m_3^\sharp \right)\left<\bos{\Theta}_1\right>
,\cr
\left<\bos{\Omega}_4\right> 
=& \left(\Fr{1}{18}m_2^\sharp m_2^\sharp m_2^\sharp-\Fr{\hbar}{6} m_2^\sharp  m_3^\sharp  
-\Fr{\hbar}{2}m_3^\sharp  m_2^\sharp  +\hbar^2 m_4^\sharp   \right)
\left<\bos{\Theta}_1\right>
.
}
$$
\end{example}

We define the generating
functional $\mb{\CZ}(t_H)$ of all quantum correlation functions 
as follows
$$
\eqalign{
\mb{\CZ}(t_H)
&:= <1> + \sum_{n=1}^\infty \Fr{(-1)^n}{\hbar^{n}}\left<\bos{\Omega}_n\right>
=\left<1\right> + \sum_{n=1}^\infty \Fr{1}{n!}\Fr{(-1)^n}{\hbar^{n}}
t^{\a_{n}}\cdots t^{\a_{1}}\left<\bos{\pi}_{\a_{1}\cdots\a_{n}}\right>
}
$$
which gives an arbitrary $n$-point correlation function by 
$$\left<\bos{\pi}_{\a_1\cdots\a_n}\right>
\equiv (-\hbar)^n \rd_{\a_1}\cdots \rd_{\a_n}\mb{\CZ}(t_H)\bigl|_{t=0}.
$$
Now the lemma $1.1$ implies that
$$
\mb{\CZ}(t_H) =\left<1\right> -\Fr{1}{\hbar} \mb{T}(t_H)^{\g}\left<\bos{O}_{\g}\right>
$$
where 
$$
\mb{T}^{\g}:=t^{\g}-\Fr{1}{2\hbar} t^\b t^\a m_{\a\b}{}^\g + 
\sum_{n=3}^\infty \Fr{1}{n!}\Fr{(-1)^{n-1}}{\hbar^{n-1}}
t^{\a_{n}}\cdots t^{\a_{1}}\bos{p}_{\a_{1}\cdots\a_{n}}{}^{\g} 
\in \Bbbk\left[\!\left[t_H,\hbar^{-1}\right]\!\right].
$$
We call $\{\mb{T}^{\g}\}$ quantum coordinates for family, which encodes the essential information
of quantum correlation functions. Detailed discussions on it is a subject of the next paper in this series.

\subsection{Epilogue: It's morphism of QFT algebra, stupid}

Quantum field theory, in general, 
has infinitely many observables. The results summarized in the above are
applicable to certain  finite super-selection sector and
to general topological field theory. Or an effective description of
quantum field theory with certain filtration of super-selection sectors
(with respect to suitable scale , etc..)
\begin{definition}
A super-selection sector of BV QFT algebra $\big(\sC[[\hbar]], \bos{K},\:\cdot\:\big)$
is a sub-algebra $\big(\sC^{\G}[[\hbar]], \bos{K},\:\cdot\:\big)$ such that
$1\in \sC^{\G}$, $\sC^{\G}\cdot \sC^{\G} \subset \sC^{\G}$, $Q \sC^{\G}\subset \sC^{\G}$,
and $K^{(\ell)}\sC^{\G} \subset\sC^{\G}$ for all $\ell\geq 1$.
A super-selection sector with super-selection rule $\G$ is called finite if the cohomology 
$H^{\G}$ of  the reduced classical complex $\big(\sC^{\G}, Q\big)$ is finite dimensional
for each ghost number.
\end{definition}

The assumption of the finite dimensionality of the space $H$ of observables  may be a technicality,
though we may hardly expect to be able to determine all possible quantum correlations between
infinitely many different observables in practice. 
At present, it is more important for us to gain  
some preliminary understanding of underlying algebraic structures of quantum field theory 
as the premise of our program to unfold besides from those possible applications,
We would, however, also like to suggest an algorithm of computing quantum correlations.
The main purpose of this epilogue is to describe such a procedure and to justify our assertion that
quantum filed theory is a study of morphisms of QFT algebras.

The essential point is to figuring out the quantization map $\bos{f}:H \rightarrow \sC[[\hbar]]$
such that $\bos{K}\bos{f}=0$, $\bos{f}(e)=1$ and its classical limit $f: H \rightarrow \sC$
a  map of choosing representative in $\sC$ of each and every element in $H$ such that $f$
is $\Bbbk$-linear and $f(e)=1$. (In Section $3$ an algorithm to extend $f$ to $\bos{f}$ is
described). Once $\bos{f}=f +\hbar f^{(1)}+\cdots$ is known algebra of quantum correlation functions is 
completely determined from it.
Let $a_1,a_2 \in H$ be any pairs of observables, i.e., $\bos{\k}a_1=\bos{\k} a_2=0$, 
 and we are interested in their two point quantum correlations.
From $Q f(a_1)=Q f(a_2)=0$, then $Q(f(a_1)\cdot f(a_2))=0$ since $Q$ is
a derivation of the product $\cdot$. 
Define 
$$
m_2(a_1,a_2)=\left[f(a_1)\cdot f(a_2)\right]\in H^{|a_1|+|a_2|},
$$
by computing the classical cohomology class of $f(a_1)\cdot f(a_2)$.
Then $f: H \rightarrow \sC$ is an algebra map up to homotopy
$$
f(a_1)\cdot f(a_2) = f(m_2(a_1,a_2)) + Q \l_2\left(a_1,a_2\right)
$$
for some $\Bbbk$-bilinear map $\l_2:S^2(H)\longrightarrow \sC$ with ghost number $-1$.
It follows that the expression 
$\bos{f}(a_1)\cdot \bos{f}(a_2) -\bos{f}\big(m_2(a_1,a_2)\big) - \bos{K} \l_2\left(a_1,a_2\right)$
is divisible by $\hbar$.  Thus the failure of $\bos{f}$ being an algebra map up to homotopy
 is divisible by $\hbar$.
Then  $\bos{\phi}_2$, defined by the following formula
$$
\hbar \bos{\phi}_2\big(a_1,a_2\big)
:=\bos{f}(a_1)\cdot \bos{f}(a_2) -\bos{f}\big(m_2(a_1,a_2)\big) - \bos{K} \l_2\left(a_1,a_2\right),
$$
is a $\Bbbk[[\hbar]]$-bilinear map on $S^2 H[[\hbar]]$ into $H[[\hbar]]$ of ghost number $0$.
Consequently we have the following $2$-point quantum correlator 
$$
\bos{\pi}_2(a_1,a_2)= \bos{f}_1(a_1)\cdot \bos{f}(a_2)-\hbar \bos{f}_2\big(a_1,a_2\big)
=\bos{f}\big(m_2(a_1,a_2)\big) + \bos{K} \l_2\left(a_1,a_2\right)
$$
with the correlation function
$$
\left<\bos{\pi}_2(a_1,a_2)\right>=\left< \bos{f}\big(m_2(a_1,a_2)\big)\right>.
$$

\begin{remark}
We note that the conditions $\bos{\k}a_1=\bos{\k} a_2=0$ do not imply
$\bos{\k}m_2(a_1,a_2)=0$. We are assuming that $\bos{\k}=0$ identically on $H$.
\end{remark}

Let
$\bos{\phi}=\bos{\phi}_1, \bos{\phi}_2, \bos{\phi}_3,\cdots$
be an infinite sequence of $\Bbbk[[\hbar]]$-multilinear maps 
$$
\bos{\phi}_n=\phi_n +\hbar \phi^{(1)}_n+\hbar^2 \phi^{(2)}_n+\cdots
: S^n(H)\longrightarrow \sC[[\hbar]], \qquad n=1,2,3,4,\cdots
$$
of ghost number zero
defined by the following recursive relations
\eqn\morphism{
\eqalign{
\bos{\phi}_1:= &\bos{f}
,\cr
\hbar\bos{\phi}_2(a_1, a_2):=
&
\bos{\phi}_1(a_1)\cdot \bos{\phi_1}(a_2)
- \bos{\phi}_1(m_2(a_1,a_2))
- \bos{K}\l_2(a_1,a_2)
,\cr
\hbar\bos{\phi}_3(a_1, a_2,a_3):=
&\bos{M}_3\big(a_1,a_2,a_3\big) - \bos{\phi}_1(m_3(a_1,a_2,a_3))
- \bos{K}\l_2(a_1,a_2,a_3)
,\cr
\vdots &
\cr
\hbar\bos{\phi}_n\big(a_1,\cdots,a_n\big)
:=&\bos{M}_n\big(a_1,\cdots,a_n\big)
 -\bos{\phi}_1\left(m_n\big(a_1,\cdots,a_n\big)\right)
-\bos{K}\l_n\big(a_1,\cdots,a_n\big)
,\cr
\vdots &
\cr
}
}
where $\bos{M}_2(a_1, a_2)=\bos{\phi}_1(a_1)\cdot \bos{\phi_1}(a_2)$
and $\bos{M}_n$ for $n\geq 3$ is
$$
\eqalign{
\bos{M}_n\big(&a_1,\cdots,a_n\big)
\cr
:=&
\sum_{k=1}^{n-1}\Fr{k(n-k)}{n(n-1)} \sum_{\s\in S_{n}}
(-1)^{|\s|}\bos{\phi}_k(a_{\s(1)},\cdots, a_{\s(k)})\cdot 
\bos{\phi}_{n-k}(a_{\s(k+1)},\cdots, a_{\s(n)})
\cr
&
-\sum_{k=2}^{n-1}\Fr{k(k-1)}{n(n-1)} \sum_{\s\in S_{n}}(-1)^{|\s|}
  \bos{\phi}_{n-k+1}
  \Big(a_{\s(1)},\cdots, a_{\s(n-k)}, m_k\big(a_{\s(n-k+1)},\cdots, a_{\s(n)}\big)\Big)
  \cr
&
-\sum_{k=2}^{n-1}\Fr{k(k-1)}{n(n-1)} \sum_{\s\in S_{n}}(-1)^{|\s|}
  \Big(
  \bos{\phi}_{n-k}\left(a_{\s(1)},\cdots, a_{\s(n-k)}\right)
      , \l_k\big(a_{\s(n-k+1)},\cdots, a_{\s(n)}\big)\Big),
}
$$
which depends only on $\bos{\phi}_1,\cdots,\bos{\phi}_{n-1}$, $m_2,\dots,m_{n-1}$ and
$\l_2,\cdots,\l_{n-1}$. 
The recursive relation implies or is based on the property that 
the classical limit ${M}_n\big(a_1,\cdots,a_n\big)$ of  the expression
$\bos{M}_n\big(a_1,\cdots,a_n\big)$ belongs to $\Ker Q$ such that
$$
{M}_n\big(a_1,\cdots,a_n\big)=f\big(m_n(a_1,\cdots,a_n)\big)
- Q\l_n(a_1,\cdots,a_n),
$$
where $m_n(a_1,\cdots,a_n):=\left[{M}_n\big(a_1,\cdots,a_n\big)\right]$, and, hence,
the expression $\bos{M}_n -\bos{f}(m_n) -\bos{K}\l_n$ is divisible by $\hbar$.
A separate proof of the above assertion is  redundant after theorem $1.2$ on
quantum master equation. 
Also quantum descendant equation implies
that
$$
\eqalign{
\bos{K}\bos{\phi}_1(a)=0
,\cr
\bos{K}\bos{\phi}_2(a_1, a_2)+ (-1)^{|a_1|}\big(\bos{\phi}_1(a_1),\bos{\phi}_1(a_2)\bigr)=0
,\cr
\vdots\phantom{00}
\cr
\bos{K}\bos{\phi}_n(a_1,\cdots, a_n)
+\Fr{1}{2}\sum_{k=1}^{n-1} \sum_{\s\in S_{n}}(-1)^{|\s|}
  \Big(
  \bos{\phi}_{k}(a_{\s(1)},\cdots, a_{\s(k)})
      , \bos{\phi}_{n-k}(a_{\s(k+1)},\cdots, a_{\s(n)})\Big)
=0.
\cr
}
$$
It should be clear that the recursive definition of $\bos{\phi}=\bos{\phi}_1, \bos{\phi}_2, \bos{\phi}_3,\cdots$ 
is nothing but an infinite sequence of $\hbar$-divisibility conditions of the quantum extension map $\bos{f}$
 with respects to the sequence $m=m_2,m_3,m_4,\cdots$ of $\Bbbk$-multilinear
multiplications on $H$ up to homotopy.

We may regard the triple $\big(H[[\hbar]], 0, m\big)$, where $m=m_2,m_3,m_4,\cdots$,
as a structure of super-commutative QFT algebra with zero differential ($\bos{\k}=0$) 
with trivial quantum descendant algebra $\big(H[[\hbar]], 0, 0\big)$. The BV QFT algebra 
$\big(\sC[[\hbar]],\bos{K}, \;\cdot\;\big)$
at the chain level is also regarded as an example of super-commutative QFT algebra with
a binary product $\cdot$ only and which quantum descendant algebra 
is a DGLA $\big(\sC[[\hbar]],\bos{K}, (\hbox{ },\hbox{ })\big)$ over $\Bbbk[[\hbar]]$.
In general quantum descendant algebra of super-commutative QFT algebra shall be
an $L_\infty$-algebra over $\Bbbk[[\hbar]]$.
Then $\bos{\phi}=\bos{\phi}_1, \bos{\phi}_2, \bos{\phi}_3,\cdots$
is automatically a morphism from the trivial quantum descendant algebra $\big(H[[\hbar]],0,0\big)$ to
the  quantum descendant algebra $\big(\sC[[\hbar]],\bos{K},(\hbox{ },\hbox{ })\big)$
as $L_\infty$-algebra over $\Bbbk[[\hbar]]$. 

We say the quantum extension map $\bos{f}$
a quasi-isomorphism from the QFT algebra $\big(H[[\hbar]], 0, m\big)$ to the QFT algebra
$\big(\sC[[\hbar]],\bos{K}, \;\cdot\;\big)$. 
We, then, we call $\bos{\phi}=\bos{\phi}_1, \bos{\phi}_2, \bos{\phi}_3,\cdots$ 
quantum descendant morphism of quasi-isomorphism $\bos{f}$ as QFT algebra.
It also follows that the classical limit $\phi=\phi_1, \phi_2, \phi_3,\ldots$ of 
quantum descendant morphism $\bos{\phi}=\bos{\phi}_1, \bos{\phi}_2, \bos{\phi}_3,\cdots$
is a  quasi-isomorphism from $\big(H,0,0,\big)$ to $\big(\sC, Q, (\hbox{ },\hbox{ })\big)$
as $L_\infty$-algebras over $\Bbbk$, since $\phi_1=f$ induces an isomorphism
on the cohomology $H$.
It should be clear that not every $L_\infty$ quasi-isomorphism from $\big(H, 0,0\big)$ to 
$\big(\sC, Q, (\hbox{ },\hbox{ })\big)$ is the classical limit of the quantum descendant
of quasi-isomorphism as QFT algebra. In case that $H$ is finite dimensional for each ghost number,
the moduli space $\CM$ defined by the MC equation of the DGLA 
$(\sC, Q, (\hbox{ },\hbox{ }))$ is smooth-formal and is equipped with quantum coordinates
due  the $L_\infty$-quasi-isomorphism $\bos{\phi}=\bos{\phi}_1, \bos{\phi}_2, \bos{\phi}_3,\cdots$ 
descended from the quasi-isomorphism $\bos{f}$ of QFT algebra.

The quasi-isomorphism $\bos{f}$ of QFT algebra 
also determines arbitrary
$n$-point quantum correlators of observables
via its quantum descendant 
$\bos{\phi}=\bos{\phi}_1, \bos{\phi}_2, \bos{\phi}_3,\cdots$ as follows: 
A partition of the set $\{1,2,\cdots, n\}$ is
a set $\pi=\{B_1, \cdots, B_{|\pi|}\}$ of nonempty subsets $B_k$, $1\leq k \leq |\pi|$
of $\{1,2,\cdots, n\}$ such that  every element in $\{1,2,\cdots, n\}$ is exactly one of these subsets.
Then $n$-point quantum correlator of observables $a_1,\cdots, a_n$ is
$$
\bos{\pi}_n(a_1,\cdots ,a_n)
:=\sum_{\pi=\{B_1, \cdots, B_{|\pi|}\}}\hbox{sign}(\pi)(-\hbar)^{n-|\pi|} 
\bos{\phi}_{|B_1|}\left(B_1\right)\cdot \bos{\phi}_{|B_2|}\left(B_2\right)
\cdots\bos{\phi}_{|B_{|\pi|}|}\left(B_{|\pi|}\right)
$$
where $\pi$ runs for all partition of the set $\{1,2,\cdots, n\}$,  
 $|B_k|$ is the size of the set $B_k$ and $\bos{\phi}_{|B_k|}\left(B_k\right)$
 means $\bos{\phi}_{|B_k|}\left(a_{i_1}, \cdots, a_{i_{|B_k|}}\right)$ for
 $B_k=\left\{{i_1}, \cdots, {i_{|B_k|}}\right\}$.
The sign $\hbox{sign}(\pi)$ of  $\pi=\{B_1, \cdots, B_{|\pi|}\}$ is determined as follows:
consider the union $B_1\cup \cdots B_{|\pi|}$  as ordered list $\{\pi(1),\pi(2),\cdots, \pi(n)\}$,
then a sign $+1$ or $-1$ can be chosen by comparing ordering of the list 
$\{a_{\pi(1)}, a_{\pi(2)},\cdots, a_{\pi(n)}\}$ with
that of  the list $\{ a_1, a_2, \cdots, a_n\}$
such that the sign is compatible with the super-commutativity of the product $\cdot$
and that of $\bos{\phi}_k$. 
For example
$$
\eqalign{
\bos{\pi}_1(a)&= \bos{\phi}_1(a)
,\cr
\bos{\pi}_2(a_1,a_2)&= \bos{\phi}_1(a_1)\cdot\bos{\phi}_1(a_2) -\hbar   \bos{\phi}_2(a_1,a_2)
,\cr
\bos{\pi}_3(a_1,a_2,a_3)&= \bos{\phi}_1(a_1) \cdot\bos{\phi}_1(a_2)\cdot \bos{\phi}_1(a_3)
\cr
-\hbar&  \bos{\phi}_1(a_1) \cdot\bos{\phi}_2(a_2,a_3)
-\hbar  \bos{\phi}_2(a_1,a_2) \cdot \bos{\phi}_1(a_3)
-\hbar(-1)^{|a_1||a_2|}\bos{\phi}_1(a_2)\cdot \bos{\phi}_2(a_1,a_3)
\cr
+\hbar&^2 \bos{\phi}_3(a_1,a_2,a_3)
,\cr
}
$$
et cetera. Then $\bos{K}\bos{\pi}_n(a_1,\cdots ,a_n)=0$ for all $n$ and
any $a_1,\cdots,a_n \in H$. A separate proof of the above assertion is  
redundant after theorem $1.3$. One may clearly notice an analogy with
the relation between $n$-point correlation function and products of connected
correlation functions in Feynman path integrals. 
 
 Now an upshot is that quantum correlation functions can be determined
 certain computation of classical cohomology. For an arbitrary $2$-point
 quantum correlator $\bos{\pi}_2(a_1,a_2)$ we have
$\left<\bos{\pi}_2(a_1,a_2)\right>= \left<\bos{f}(m_2(a_1,a_2))\right>$,
thus the quantum expectation value of the observable $m_2(a_1,a_2) \in H$,
which is the classical cohomology class $\left[f(a_1)\cdot f(a_2)\right]$ of $f(a_1)\cdot f(a_2)$.
For an arbitrary $3$-point
 quantum correlator $\bos{\pi}_3(a_1,a_2,a_3)$, we have
$$
\Big<\bos{\pi}_2(a_1,a_2,a_3)\Big>= \Big<\bos{f}\big(m_2(m_2(a_1,a_2),a_3\big)\Big>
-\hbar  \Big<\bos{f}\big(m_3(a_1,a_2,a_3\big)\Big>,
$$
where $m_3(a_1,a_2,a_3)$ is the classical cohomology class of $\left[M_3(a_1,a_2,a_3)\right]$;
$$
\eqalign{
M_3&(a_1,a_2,a_3)\cr
=&
\Fr{2}{3}\Big[\phi_1(a_{1}) \cdot\phi_2(a_{2},a_{3})
+\phi_2(a_{1},a_{2})\cdot\phi_1(a_{3})
+(-1)^{|a_{1}||a_{2}|}\phi_1(a_{2})\cdot\phi_2(a_{1},a_{3})  \Big]
\cr
&
-\Fr{1}{3}\Big[
\phi_2\big(a_{1}, m_2(a_{2},a_{3})\big)
+\phi_2\big(m_2(a_{1},a_{2}),a_{3}\big)
+(-1)^{|a_{1}||a_{2}|}\phi_2\big(a_{2}, m_2(a_{1},a_{3})\big)
\Big]
\cr
&
-\Fr{1}{3}\Big[
\big(\phi(a_{1}), \l_2(a_{2},a_{3})\big)
-\big( \l_2(a_{1},a_{2}),\phi_1(a_{3})\big)
-(-1)^{|a_{1}||a_{2}|}\big(\phi_1(a_{2}), l_2(a_{1},a_{3})\big)
\Big].
}
$$
We also note that 
$$
\phi_2(a_i, a_j)=f^{(1)}(a_i)\cdot f (a_j)+f(a_i)\cdot f^{(1)} (a_j)
-f^{(1)}\big(m_2(a_i,a_j)\big) -K^{(1)}\l_2(a_i, a_j).
$$
In general
$$
\Big<\bos{\pi}_n(a_1,\cdots ,a_n)\Big> =  \Big< \bos{f}\big(\bos{p}_n(a_1,\cdots, a_n\big)\Big>
$$
where $\bos{p}_n$ is defined recursively as follows
$$
\eqalign{
\bos{p}_n(&a_1,\cdots, a_n\big)
\cr
=
&(-\hbar)^{n-2} {m}_n(a_1,\cdots, a_n)
\cr
&
+\Fr{1}{n(n-1)}\sum_\s(-1)^{|\s|}
\sum_{k=2}^{n-1}(-\hbar)^{k-2}k(k-1)
\bos{p}_{n-1-k}\big(m_k(a_{\s(1)},a_{\s(2)}), a_{\s(3)},\cdots, a_{\s(n)}\big)
}
$$
with the initial condition that $\bos{p}_2(a_i, a_j)=m_2(a_i,a_j)$.
It follows that it suffice to determine $f, f^{(1)},\cdots, f^{(n-1)}$ to determine
$m_2,\cdots, m_n$, thus $\bos{p}_2,\cdots, \bos{p}_n$ via classical computations.

\newsec{BV QFT Algebra}

In this section we define  BV QFT algebra and its descendant DGLA with
several examples. We begin with recalling some standard algebra notions 
including DGLA, differential $0$-algebra, BV algebra and differential BV algebra
before discussing BV QFT algebra.

\subsection{Differential $0$-algebra, BV and differential BV algebras}
The contents of this subsection are standard, though things may be named differently
in literature. Fix a ground field $\Bbbk=\R, \C,\ldots$ of characteristic zero.
Every algebra in this subsection is defined over $\Bbbk$, which may be replaced with
a commutative ring.

Let $\sC$ denote a $\Z$-graded $\Bbbk$-module
$$
\sC = \bigoplus_{i\in \Z} \sC^i
$$
We say that a homogeneous element $a \in \sC^i$ carries   the ghost number $i$, and
use notation $|a|$ for the ghost number of $a$. The ground field $\Bbbk$ is assigned
to have the ghost number $0$.
Let $(\sC, \;\cdot\;)$ denote a $\Z$-graded  super-commutative associative $\Bbbk$-algebra;
$\sC$ is a $\Z$-graded $\Bbbk$-module and the product 
$$
\;\cdot\; : \sC^{i}\otimes \sC^{j}\longrightarrow \sC^{i+j}
$$
is a $\Bbbk$-bilinear map of ghost number $0$ satisfying 
the super-commutativity
$$
a\cdot b = (-1)^{|a||b|} b\cdot a
$$ 
and the associativity 
$$
a\cdot (b\cdot c)= (a\cdot b)\cdot c.
$$  
A cochain complex over $\Bbbk$ is a pairs $(\sC, Q)$ where
$Q$ is a $\Bbbk$-linear map of ghost number $1$ on $\sC$ into itself,
i.e., $Q:\sC^{j}\longrightarrow \sC^{j+1}$ for all $j$, satisfying $Q^{2}=0$.

\begin{definition}
A super-commutative differential $\bos{Z}$-graded algebra (CDGA) over $\Bbbk$ is
a triple $\left(\sC, Q,\,\cdot\,\right)$ if the pair $(\sC, \,\cdot\,)$ is
a  $\Z$-graded  super-commutative associative $\Bbbk$-algebra,
the pair $(\sC,Q)$ is a cochain complex over $\Bbbk$ and $Q$ is
a (graded) derivation of the product;
$$
Q (a\cdot b)= (Q a)\cdot b + (-1)^{|a|} a\cdot (Q b).
$$
\end{definition}

\begin{definition}
A differential graded 0-Lie algebra  over $\Bbbk$ 
is a triple $\left(\sC,Q,(\bullet,\bullet)\right)$, where the pairs $(\sC,Q)$ is a cochain complex
over $\Bbbk$ and
\begin{enumerate}
\item the bracket $
(\bullet,\bullet): \sC^{k_1}\otimes \sC^{k_2}\longrightarrow \sC^{k_1+k_2+1}$ 
is $\Bbbk$-bilinear with  ghost number $1$.
\item the bracket $(\bullet,\bullet)$  is graded-commutative 
$$
(a, b) = -(-1)^{(|a|+1)(|b|+1)}(b, a),
$$
and is a derivation of the bracket (graded-Jacobi law)
$$
(a,(b,c)) = ((a,b),c) +(-1)^{(|a|+1)(|b|+1)}(b, (a,c)).
$$
\item the differential $Q$ is a derivation of the bracket
$$
Q(a, b) = (Q a, b) +(-1)^{|a|+1}(a, Q b).
$$
A graded 0-Lie algebra  is a differential graded 0-Lie algebra with zero differential, $Q=0$.
\end{enumerate}
\end{definition}
\begin{remark}
The bracket in the standard definition of differential graded Lie algebra carries
the ghost number $0$. The differential graded 0-Lie algebras have the same properties
and utilities as differential graded Lie algebras after shifting ghost number by $1$.
The standard differential graded Lie algebra shall never appear in this paper and whenever
we use the term DGLA we refer to a differential graded 0-Lie algebra.
\end{remark}

\begin{definition}
The cohomology of CDGA or DGLA is the cohomology of the underlying cochain complex
\end{definition}
The Maurer-Cartan (MC) equation of a differential $0$-Lie algebra (DGLA)
$\left(\sC, Q, (\bullet,\bullet)\right)$ is the following equation
$$
Q\g +\Fr{1}{2}\left(\g,\g\right)=0,
$$
for $\g \in \sC^{1}$. Let $\g$ is a solution. Then the MC equation implies that
$Q_{\g}:= Q +\left(\g,\bullet\right):\sC^{j}\longrightarrow \sC^{j+1}$ satisfies
$Q_{\g}^{2}=0$ after using the graded Jacobi-law of the bracket and $Q$ being a
derivation of the bracket. Also $Q_{\g}$ is a derivation of the bracket, which property
follows from  the graded Jacobi-law. Thus  
$\left(\sC, Q_{\g},   (\bullet,\bullet)\right)$ is also a DGLA.
Solutions of MC equation comes with a natural notion of gauge equivalence:
One can check that $\left(\sC^{-1}, (\bullet,\bullet)\right)$ is a standard Lie algebra, since
the bracket has ghost number $1$. This Lie algebra acts on $\g\in \sC$ infinitesimally
by $\dot{\g}=Q\l +\left(\l,\g\right)$, which action can be exponentiated to the gauge group 
provided that the Lie algebra is nilpotent,

\begin{definition}
A quadruple $(\sC, Q,\cdot, (\bullet,\bullet))$ is a differential $0$-algebra over $\Bbbk$
if (i) the triple $\left(\sC, Q,\,\cdot\,\right)$ is a CDGA over $\Bbbk$,
(ii)  the triple $\left(\sC,Q,(\bullet,\bullet)\right)$ is a DGLA over $\Bbbk$,
and (iii) the bracket is a derivation of the product
(graded-Poisson law)
$$
(a, b\cdot c)= (a,b)\cdot c + (-1)^{(|a|+1)|b|}b\cdot (a,c).
$$
A differential $0$-algebra with the zero differential $Q=0$ is a $0$-algebra.
\end{definition}
Our standard example of differential $0$-algebra is from so-called classical BV master equation;
\begin{example}
Let  $\left(\sC, \cdot, (\bullet,\bullet)\right)$ be a $0$-algebra with an element $S \in \sC^0$
of the ghost number $0$
satisfying $$(S,S)=0.$$ 
Define $Q:=(S,\bullet)$, which is a $\Bbbk$-linear map of ghost number $1$
on $\sC$ to $\sC$,
then  the $4$-tuple $\big(\sC, Q, \cdot, (\bullet,\bullet)\big)$
is a differential $0$-algebra, since  
(i) the graded-Jacobi law of the bracket
implies $Q$ is a derivation of the bracket, (ii) the graded-Poisson law implies
that $Q$ is a derivation of the product, (iii) the condition $(S,S)=0$ and the graded-Jacobi
law implies that $Q^2=0$.
\end{example}

\begin{definition}
A BV algebra  over $\Bbbk$ is a triple  $(\sC, \cdot, \Delta)$, where $(\sC, \cdot)$ is
a $\Z$-graded super-commutative associative algebra over $\Bbbk$
and  $\Delta$ is  $\Bbbk$-linear, $\Delta \Bbbk=0$, operator of ghost number $1$,  
$\Delta:\sC^{k} \longrightarrow \sC^{k+1}$ satisfying
$\Delta^2=0$, such that
\begin{enumerate}
\item the BV operator $\Delta$ is not a derivation of the product,
which failure is measured by  so called BV bracket 
$$
(\bullet,\bullet): \sC^{k_1}\otimes \sC^{k_2}\longrightarrow \sC^{k_1+k_2+1}
$$
by the following formula
$$
(-1)^{|a|}(a, b):=\Delta(a\cdot b) - \Delta a \cdot b -(-1)^{|a|}a\cdot \Delta b,
$$
\item 
 the BV bracket is a derivation of the product (graded-Poisson law)
$$
(a, b\cdot c)= (a,b)\cdot c + (-1)^{(|a|+1)|b|}b\cdot (a,c),
$$
\end{enumerate}
\end{definition}

\begin{corollary}
For a BV algebra $\left(\sC, \cdot, \Delta\right)$ with associated 
BV bracket  $(\bullet,\bullet)$,
\begin{enumerate}
\item the pair $\left(\sC, (\bullet,\bullet)\right)$ is 
a graded $0$-Lie algebra  over $\Bbbk$, 
\item $\Delta$ is a derivation of the BV bracket;
$$
\Delta (a, b) = (\Delta a, b) + (-1)^{|a|+1}(a,\Delta b).
$$
\end{enumerate}
\end{corollary}

\begin{proof}
The proof of the above corollary is standard or may be served as a good exercise.
Here is a sketch of a proof.
The graded-commutativity of the BV bracket follows from the super-commutativity
of the product and  the graded-Jacobi identity follows from 
the graded-Poisson law after applying $\Delta$. 
Finally, $\Delta$ being a derivation of the bracket follows by
applying  $\Delta$ to the defining equation of  the BV bracket and 
use the property that $\Delta^2=0$.
\qed
\end{proof}

\begin{corollary}
Let $\big(\sC,  \Delta,\,\cdot\,\big)$ be a  BV algebra  with 
the associated BV bracket  $(\bullet,\bullet)$.
Then the  triple $\big(\sC, \,\cdot\,, (\bullet,\bullet)\big)$, 
after forgetting $\Delta$, is a $0$-algebra. 
\end{corollary}

\begin{remark}
We should emphasis that not every $0$-algebra is originated from BV algebra.
\end{remark}

\begin{definition}
A differential BV algebra  over $\Bbbk$ is a $4$-tuple     
$\big(\sC,\Delta, Q, \,\cdot\,\big)$
where the triple $\big(\sC,\Delta,\,\cdot\,\big)$ is a BV algebra over $\Bbbk$ and the triple
$\big(\sC,Q,\,\cdot\,\big)$ is a CDGA over $\Bbbk$ such that
$Q\Delta + \Delta Q=0$.
\end{definition}

\begin{corollary}
Let $\big(\sC, \Delta, Q, \,\cdot\,\big)$ be a differential BV algebra  with 
BV bracket  $(\bullet,\bullet)$ associated to the BV algebra $\big(\sC, \Delta,\,\cdot\, \big)$.
Then $Q$ is a derivation of the BV bracket $(\bullet,\bullet)$.
\end{corollary}

\begin{proof}
For any homogeneous elements $a,b\in\sC$, we have
$$(Q\Delta + \Delta Q)(a\cdot b)=0.$$ 
Now use the definition of the BV bracket and
the property that $Q$ is a derivation of the product $\cdot$. 
Then use the property $Q\Delta +\Delta Q=0$
again to deduce that
$$
Q(a,b)= (Q a, b) + (-1)^{|a|+1} (a, Q b).
$$
 \qed
\end{proof}

\begin{corollary}

Let $\left(\sC, \Delta, \cdot, Q\right)$ be a differential BV algebra  with 
BV bracket  $(\bullet,\bullet)$ associated to the BV algebra $\left(\sC, \Delta,\,\cdot\,\right)$.
Then the quadruple $\left(\sC, \cdot, (\bullet,\bullet), Q\right)$, after forgetting $\Delta$, 
is a differential $0$-algebra.

\end{corollary}

\begin{remark}
We should emphasis that not every differential $0$-algebra is originated from
differential BV algebra.
\end{remark}

Our standard example of differential BV algebra is 
from so-called semi-classical BV master equation.

\begin{example}
Let  $\left(\sC, \Delta, \,\cdot\,\right)$ is a BV algebra with associated BV bracket 
$(\bullet,\bullet)$.
For an $S \in \sC^0$ satisfying
$$
\eqalign{
\Delta S=0,\cr
(S,S)=0,
}
$$
define $Q:= (S,\bullet):\sC^k\longrightarrow \sC^{k+1}$. Then 
 $\left(\sC,\Delta,\, \cdot\,,  Q:=(S,\bullet)\right)$
is a differential BV algebra. To see this note that,
for any $a\in\sC$,
$$
\Delta(S, a) +(S, \Delta a)=(\Delta S, a) - (S, \Delta a) + (S,\Delta a)=0,
$$
using $\Delta$ being a derivation of the bracket and the condition $\Delta S=0$.
Thus $\Delta Q + Q\Delta =0$. It remains to show
that $(\sC, Q,\,\cdot\,)$ is a CDGA, which follows from (i)
the graded-Poisson law of the bracket,  $(S, a\cdot b)= (S,a)\cdot b +(-1)^{|a|} a\cdot (S,b)$,
implies that $Q$ is a derivation of the product, (ii) the condition $(S,S)=0$ and the super-Jacobi
law implies that $Q^2=0$.
\end{example}

\begin{remark}
Let  $\left(\sC, \Delta, \,\cdot\,\right)$ is a BV algebra with associated BV bracket 
$(\bullet,\bullet)$. Then the triple $\left(\sC, \Delta,(\bullet,\bullet)\right)$
is obviously a DGLA. But we shall never be interested in such a DGLA  
since it is irrelevant to our problems.
By the way the above DGLA is rather boring.
Also we shall  never be interested in the $\Delta$-cohomology.
\end{remark}

\subsection{BV QFT Algebra}

Let $(\sC,\,\cdot\,)$ be a ${Z}$-graded super-commutative and associative unital
$\Bbbk$-algebra with the multiplication $\cdot$.
The physical Planck constant $\hbar$ will be regarded as a formal parameter with zero ghost number.
Let  $\sC[[\hbar]]:= \sC\otimes_{\Bbbk}\Bbbk[[\hbar]]
=\sC \oplus \hbar \sC\oplus \hbar^{2}\sC\oplus \cdots$.
We shall denote an element of $\sC[[\hbar]]$ by an upright {\bf bold} letter, 
i.e., $\bos{a}\in \sC[[\hbar]]$, and an element of $\sC$ by an {\itshape italic} letter, i.e.,
$a \in \sC$. Formal power series expansion of an element $\bos{a}$ in $\sC[[\hbar]]$ 
shall be denoted as follows; 
$$
\bos{a} = a^{(0)} + \hbar a^{(1)} +\hbar^2 a^{(2)}+\cdots,  
$$
where $a^{(n)}\in \sC$ for all $n=0,1,2,\cdots$.
We shall often denote $a^{(0)}$ by $a$.
On $\sC[[\hbar]]$ there is a canonical $\Bbbk[[\hbar]]$-bilinear
product induced from $\sC$, which will be denoted
by the same symbol $\cdot$;
$$
\bos{a}\cdot \bos{b} :=\sum_{n=0}^{\infty}\hbar^{n}\sum_{j=0}^{n}a^{(j)}\cdot b^{(n-j)}, 
$$
such that
$\big(\sC[[\hbar]],\,\cdot\,\big)$ is a ${Z}$-graded super-commutative and associative 
unital $\Bbbk[[\hbar]]$-algebra. The unit in $\sC$,
which is also the unit in $\sC[[\hbar]]$, shall be denoted by $1$. 
In general a $\Bbbk$-multilinear map 
of $\sC$ into $\sC$  canonically induces a $\Bbbk[[\hbar]]$ multilinear map 
of $\sC[[\hbar]]$ into $\sC[[\hbar]]$,
and we shall not distinguish them.
We shall often deals with certain $\Bbbk[[\hbar]]$-linear map in the form
$\bos{L}=L^{(0)}+\hbar L^{(1)}+\hbar^{2}L^{(2)}+\cdots$ on $\sC[[\hbar]]$ into itself,
where $L^{(0)},L^{(1)}, L^{(2)},\cdots$ is an infinite sequence of $\Bbbk$-linear
maps on $\sC$ into itself and each $L^{(n)}$ increase the ghost number by $N$.
Then we shall often say that $\bos{L}=L^{(0)}+\hbar L^{(1)}+\hbar^{2}L^{(2)}+\cdots$
is a sequence of $\Bbbk$-linear maps of ghost number $N$ parametrized by $\hbar$ 
on $\sC$ into itself. Such the map $\bos{L}=L^{(0)}+\hbar L^{(1)}+\hbar^{2}L^{(2)}+\cdots$
is a $\Bbbk[[\hbar]]$-linear map of ghost number $N$ on $\sC[[\hbar]]$ into itself,
and its action on $\bos{a}\in \sC[[\hbar]]$ is
$$
\eqalign{
{\bf{L}}\left(\bos{a}\right) 
&= \sum_{n=0}^{\infty}\sum_{j=0}^{n}\hbar^{n}L^{(n-j)}\left(a^{(j)}\right)
\cr
&=
L^{(0)}a^{(0)} + \hbar\left( L^{(0)}a^{(1)} +  L^{(1)}a^{(0)}\right) +\hbar^2
\left(L^{(0)}a^{(2)} +L^{(2)}a^{(0)} +L^{(1)}a^{(1)}\right)+\cdots. 
}
$$
Let $\bos{L}_{1}$ and $\bos{L}_{2}$ be two sequences of 
of $\Bbbk$-linear maps of ghost number $N_{1}$ and $N_{2}$, respectively, 
parametrized by $\hbar$ on $\sC$ into itself. 
Then the composition $\bos{L}_{3}=\bos{L}_{1}\bos{L}_{2}$
as $\Bbbk[[\hbar]]$-linear maps is a sequence $\bos{L}_{3}=L_{3}^{(0)}+\hbar L_{3}^{(1)}
+\cdots$ of $\Bbbk$-linear maps of ghost number $N_{1}+N_{2}$ on $\sC$ into itself
such that
$$
\bos{L}_{3}=\sum_{n=0}^{\infty}\sum_{j=0}^{n}\hbar^{n}L_{1}^{(n-j)}L_{2}^{(j)}.
$$
Projection of any structure parametrized by $\hbar$ on $\sC[[\hbar]]$ 
to $\sC$ will be called taking classical limit.

\begin{definition}
\label{Def1}
Let  $\bos{K}=Q +\hbar K^{(1)} +\hbar^{2}K^{(2)}+\cdots$ 
be a sequence of $\, \Bbbk$-linear maps of ghost number $1$
parametrized by $\hbar$ on $\sC$ into $\sC$
satisfying $\bos{K}^2=0$ and $\bos{K}1=0$.
Then the triple
$$
\bigl(\sC[[\hbar]], \bos{K}, \hbox{ }\cdot\hbox{ }\bigr)
$$ 
is  a BV QFT algebra if the failure of $\bos{K}$ being  a derivation of the product
$\cdot$ is {\it divisible} by $\hbar$ and the binary operation measuring
the failure is a derivation of the product.
\end{definition}
The condition $\bos{K}^{2}=0$ says that
the infinite sequence $Q,K^{(1)}, K^{(2)},\cdots$ of $\Bbbk$-linear maps
on $\sC$ into itself with ghost number $1$ satisfy the following infinite
sequence of relations;
\eqn\toera{
\eqalign{
Q^{2}&=0,\cr
Q K^{(n)}+K^{(n)}Q+\sum_{\ell=1}^{n-1}K^{(n-\ell)} K^{(\ell)}&=0 
\hbox{ for all } n=1,2,\cdots .
}
}
In particular 
the classical limit $Q$ of $\bos{K}$ satisfies
$Q^{2}=0$ so that $(\sC, Q)$ is a cochain complex.
Since the failure of $\bos{K}$ being a derivation of the product is proportional to $\hbar$,
it follows that $Q$ is a derivation of the product.
Thus, the classical limit 
$$
\bigl(\sC, Q,\,\cdot\,\bigr)
$$
of the BV QFT algebra  is a $\bos{Z}$-graded super-commutative unital
differential graded algebra (CDGA) over $\Bbbk$. 

On $\sC[[\hbar]]$, as a $\Bbbk[[\hbar]]$-module freely generated by $\sC$, 
there is a natural automorphism group consists of arbitrary 
sequence $\bos{g}=1+ g^{(1)}\hbar +g^{(2)}\hbar^{2}+ \cdots$
of $\Bbbk$-linear maps of ghost number $0$ on $\sC$ into
itself parametrized by $\hbar$ satisfying $\bos{g}\bigl|_{\hbar=0}=1$.
Such an automorphism
will act on both the unary operation $\bos{K}$ and the binary operation $\cdot$
as 
$
\bos{K}\rightarrow \bos{K}^{\pr}$ such that $\bos{K}^{\pr}=\bos{g}\bos{K} \bos{g}^{-1}
$
and $\cdot \rightarrow \cdot^{\pr}$ such that 
$\cdot^{\pr}= \bos{g}\left(\bos{g}^{-1}\cdot \bos{g}^{-1}\right)$.
It is, then, trivial that $\left(\sC[[\hbar]],\bos{K}^{\pr},\,\cdot^{\pr}\,\right)$ is also a BV
QFT algebra. Note that such automorphisms  fix the classical limit, i.e.,
$Q=Q^{\pr}:=\bos{K}^{\pr}\bigr|_{\hbar=0}$ and $a\cdot^{\pr}b= a\cdot b$
for $a,b\in \sC$.
Such an automorphism should be regarded as ``gauge symmetry'' 
of quantization procedure, so that
the resulting two BV QFT algebras $\left(\sC[[\hbar]], \bos{K},\,\cdot\,\right)$ 
and $\left(\sC[[\hbar]], \bos{K}^{\pr},\,\cdot^{\pr}\,\right)$ should be regarded as equivalent.
Thus we are lead to study BV QFT algebra modulo the ``gauge symmetry'', while our 
algebraic path integral shall be ``gauge invariant''.

Let $\bos{\n}_{2}$ denotes the binary operation,
which measures the failure of $\bos{K}$ being a derivation of the product;
$$
(-1)^{|\bos{a}|}\bos{\n}_{2}\left(\bos{a},\bos{b}\right):= \bos{K}\left(\bos{a}\cdot \bos{b} \right)
- \bos{K a}\cdot \bos{b} -(-1)^{|\a|}\bos{a}\cdot \bos{Kb}.
$$
Then $\bos{\n}_{2}$ is a $\Bbbk[[\hbar]]$-bilinear map of ghost number $1$ on 
$\sC[[\hbar]]\otimes \sC[[\hbar]]$ into $\sC[[\hbar]]$.
By definition $\bos{\n}_{2}$ is divisible by $\hbar$, thus 
$-\Fr{1}{\hbar}\bos{\n}_{2}$ is  also $\Bbbk[[\hbar]]$-bilinear map of ghost number $1$ on 
$\sC[[\hbar]]\otimes \sC[[\hbar]]$ into $\sC[[\hbar]]$. Then, again by definition,
both $\bos{\n}_{2}$ and $-\Fr{1}{\hbar}\bos{\n}_{2}$ is a derivation of the product.
We note that
a BV QFT algebra $\big(\sC[[\hbar]],\bos{K},\,\cdot\,\big)$ is a BV algebra
over $\Bbbk[[\hbar]]$ with associated BV bracket $\bos{\n}_{2}$.
It follows that the triple $\left(\sC[[\hbar]], \bos{K},\bos{\n}_{2}\right)$ is
a DGLA over $\Bbbk[[\hbar]]$. Also
the triple 
 $\big(\sC[[\hbar]], \bos{K},-\Fr{1}{\hbar}\bos{\n}_{2}\big)$
is another DGLA over $\Bbbk[[\hbar]]$.
We emphasis that  {\it not} every BV algebra over $\Bbbk[[\hbar]]$ is a BV QFT algebra.
Here are simple examples:
\begin{example}
Let $(\sC,\Delta,\,\cdot\,)$ be a BV algebra over $\Bbbk$ with associated
BV bracket $\left(\hbox{ },\hbox{ }\right)$. Then
\begin{itemize}

\item
the triple $\left(\sC[[\hbar]], \Delta,\,\cdot\,\right)$ is a BV-algebra
over $\Bbbk[[\hbar]]$ but is not a BV QFT algebra, since the failure of
$\Delta$ being derivation of the product is not divisible by $\hbar$,

\item
the triple $\left(\sC[[\hbar]], -\hbar\Delta,\,\cdot\,\right)$ is a BV-algebra
over $\Bbbk[[\hbar]]$ as well as  a BV QFT algebra.

\end{itemize}
We are not interested in BV algebras over $\Bbbk[[\hbar]]$ but in BV QFT algebras.
Also we have no use of the 
DGLA $\left(\sC[[\hbar]], \bos{K},\bos{\n}_{2}\right)$, which is rather boring object.
\end{example}
\begin{definition}

Let the triple $\bigl(\sC[[\hbar]], \bos{K}, \hbox{ }\cdot\hbox{ }\bigr)$ be
a BV QFT algebra. Define a $\Bbbk[[\hbar]]$-bilinear map
$
\left(\hbox{ },\hbox{ }\right):\sC[[\hbar]]^{k_{1}}\otimes \sC[[\hbar]]^{k_{2}}
\longrightarrow \sC[[\hbar]]^{k_{1}+k_{2}+1}
$
by the formula
$$
\left(\bos{a},\bos{b}\right):= -(-1)^{|\bos{a}|}\Fr{1}{\hbar}
\left(\bos{K}\left(\bos{a}\cdot \bos{b} \right)
- \bos{K a}\cdot \bos{b} -(-1)^{|\bos{a}|}\bos{a}\cdot \bos{Kb}\right).
$$
Then we call the triple 
$$
\big(\sC[[\hbar]], \bos{K},(\bullet,\bullet)\big)
$$
the descendant algebra of the the BV QFT algebra  
$\bigl(\sC[[\hbar]], \bos{K}, \hbox{ }\cdot\hbox{ }\bigr)$.

\end{definition}

\begin{corollary}
The descendant algebra 
$\big(\sC[[\hbar]], \bos{K},(\bullet,\bullet)\big)$ 
is a DGLA over $\Bbbk[[\hbar]]$ and the bracket 
$(\bullet,\bullet)$ is a derivation of the product $\cdot$;

\begin{enumerate}
\item the operator $\bos{K}$ is a (graded) derivation of the bracket
$$
\bos{K}\left(\bos{a}, \bos{b}\right)=\left(\bos{K}\bos{a}, \bos{b}\right)
+(-1)^{|\bos{a}|+1}\left(\bos{a}, \bos{K}\bos{b}\right),
$$
\item the bracket is graded commutative
$$
(\bos{a}, \bos{b}) = -(-1)^{(|\bos{a}|+1)(|\bos{b}|+1)}(\bos{b}, \bos{a}),
$$
\item the bracket is a (graded) derivation of the bracket (the graded Jacobi-identity) 
$$
(\bos{a},(\bos{b},\bos{c})) = ((\bos{a},\bos{b}),\bos{c}) 
+(-1)^{(|\bos{a}|+1)(|\bos{b}|+1)}(\bos{b}, (\bos{a},\bos{c})),
$$
\item the bracket is a (graded) derivation of the product (graded Poisson-law)
$$
(\bos{a}, \bos{b}\cdot \bos{c})= (\bos{a},\bos{b})\cdot \bos{c} 
+ (-1)^{(|\bos{a}|+1)|\bos{b}|}\bos{b}\cdot (\bos{a},\bos{c}).
$$
\end{enumerate}
\end{corollary}

\begin{remark}
We emphasis that the bracket $(\hbox{ },\hbox{ })$ is a  secondary notion 
in the definition of  BV QFT algebra. So we call
the triple $\big(\sC[[\hbar]], \bos{K},(\bullet,\bullet)\big)$ 
the {\it descendant} DGLA to the BV QFT algebra
$(\sC[[\hbar]], \bos{K}, \hbox{ }\cdot\hbox{ })$. 
We also remark that the bracket $(\hbox{ },\hbox{ })$ in general may depend on $\hbar$;
$$
(\hbox{ },\hbox{ })=
(\hbox{ },\hbox{ })^{(0)}+ \hbar (\hbox{ },\hbox{ })^{(1)}+\hbar^2 (\hbox{ },\hbox{ })^{(2)}
+\cdots,
$$
where $(\hbox{ },\hbox{ })^{(\ell)}$ are $\Bbbk$-bilinear maps on $\sC\otimes \sC$ into $\sC$.
Abusing notation, we shall denote the classical limit $(\hbox{ },\hbox{ })^{(0)}$ of  the bracket 
$(\hbox{ },\hbox{ })$ by the same notation $(\hbox{ },\hbox{ })$.

\end{remark}

The classical limit 
 $\bigl(\sC, Q,(\bullet,\bullet)\bigr)$
of the  descendant DGLA 
is a DGLA over $\Bbbk$.
We  emphasis that the  DGLA  $\bigl(\sC, Q,(\bullet,\bullet)\bigr)$ 
has a quantum origin, however the secondary notion as it is. 
Under the natural automorphism group of BV QFT algebra 
the bracket in its descendant DGLA changes as 
$\left(\bullet,\bullet\right)^{\pr}
= \bos{g}\left(\bos{g}^{-1}\bullet, \bos{g}^{-1}\bullet\right)$,  
while its classical limit remains fixed.

\begin{remark}
Not every DGLA  over $\Bbbk$ is a classical limit of the descendant DGLA of
a BV QFT algebra.  
\end{remark}

\begin{corollary}
Let $\big(\sC,Q,\,\cdot\,\big)$ be the classical limit of a BV QFT algebra.
Let $\big(\sC,Q,(\bullet,\bullet)\big)$ be the classical limit of the
descendant DGLA. Then the quadruple 
$\big(\sC, Q,\,\cdot\,, (\bullet,\bullet)\big)$
is a differential $0$-algebra over $\Bbbk$.
\end{corollary}
\begin{remark} Not every a differential $0$-algebra over $\Bbbk$ is a classical limit
of a BV QFT algebra and its descendant combined.
\end{remark}

Now consider some examples of BV QFT algebra and its descendant.

\begin{example}
Let $\big(\sC,\Delta,\,\cdot\,\big)$ be a BV algebra over $\Bbbk$ with associated
BV bracket $(\bullet,\bullet)$. Then
the triple 
$$
\big(\sC[[\hbar]], -\hbar\Delta,\,\cdot\,\big)
$$
is a BV QFT algebra
with the descendant DGLA 
$$
\big(\sC[[\hbar]], -\hbar\Delta,(\bullet,\bullet)\big).
$$
The Maurer-Cartan equation of the descendant DGLA is
\eqn\muac{
-\hbar \Delta \bos{S} +\Fr{1}{2}\left(\bos{S},\bos{S}\right)=0,
}
where $\bos{S}=S+\hbar S^{(1)}+\hbar^{2}S^{(2)}+\cdots\in \sC[[\hbar]]^{0}$, 
is precisely the BV quantum master equation
and which solution  is a BV quantum master action.
Let $\bos{S}$ be such a solution. Then \muac\ implies that
the operator 
$$
-\hbar \Delta +\left(\bos{S}, \hbox{ }\right):\sC[[\hbar]]^{\bullet}\longrightarrow 
\sC[[\hbar]]^{\bullet+1}
$$ 
squares to zero.
Then the triple
$$
\big(\sC[[\hbar]], -\hbar\Delta+\left(\bos{S}, \hbox{ }\right),\,\cdot\,\big)
$$
is also a BV QFT algebra, which descendant DGLA is
$$
\big(\sC[[\hbar]], 
-\hbar\Delta +\left(\bos{S},\hbox{ }\right),(\bullet,\bullet)
\big),
$$
with the same bracket of the previous one, since the bracket is a derivation of the bracket
- note that the BV bracket does not depend on $\hbar$, and $(\bos{S},\hbox{ })$ is
a derivation of the product. 
The Maurer-Cartan equation of the above DGLA 
$$
-\hbar \Delta \bos{\G} +\left(\bos{S},\bos{\G}\right)+\Fr{1}{2}\left(\bos{\G},\bos{\G}\right)=0,
\quad \bos{\G}\in \sC[[\hbar]]^{0}
$$
appears to be controlling deformation of QFT, as we mentioned in the introduction.
But the DGLA is a descendant notion and so is its Maurer-Cartan equation is.
\end{example}

\begin{example}
Let $\big(\sC,\Delta,Q,\,\cdot\,\big)$ be a differential BV algebra over $\Bbbk$ with associated
BV bracket $(\bullet,\bullet)$. Then
the triple 
$$
\big(\sC[[\hbar]], -\hbar\Delta + Q,\,\cdot\,\big)
$$
is a BV QFT algebra
with the descendant DGLA 
$$
\big(\sC[[\hbar]], -\hbar\Delta+Q,(\bullet,\bullet)\big).
$$
The Maurer-Cartan equation of the descendant DGLA is
\eqn\muad{
-\hbar \Delta \bos{S} +Q\bos{S}+\Fr{1}{2}\left(\bos{S},\bos{S}\right)=0,
}
where $\bos{S}=S+\hbar S^{(1)}+\hbar^{2}S^{(2)}+\cdots\in \sC[[\hbar]]^{0}$.
Assume that there is a solution $\bos{S}$ such that $\bos{S}=S \in \sC^{0}$, i.e.,
\eqn\muae{
\eqalign{
\Delta S &=0
,\cr
Q S +\Fr{1}{2}\left(S,S\right)&=0.
}
}
The above equation is called semi-classical BV master equation.
\end{example}

\begin{example}
Let $\big(\sC[[\hbar]],\bos{K}, \,\cdot\,\big)$ is a BV QFT algebra 
such that $\bos{K}= Q + \hbar K^{(1)}$. Let
$\big(\sC[[\hbar]],\bos{K},(\bullet,\bullet)\big)$ be the
descendant DGLA.
Then the quadruple
$\big(\sC,-K^{(1)}, Q, \,\cdot\,\big)$ is a differential BV algebra with associated
BV bracket $(\bullet,\bullet)$, since
\begin{enumerate}
\item
the condition $\bos{K}^{2}=0$ implies that $Q^{2}=Q K^{(1)}+K^{(1)}Q
=\left(K^{(1)}\right)^{2}=0$, 
\item
the condition
$\bos{K}\left(\bos{a}\cdot\bos{b}\right) -\bos{K}\bos{a}\cdot \bos{b}
-(-1)^{|\bos{a}|}\bos{a}\cdot \bos{K}\bos{b}=-\hbar (-1)^{|\bos{a}|}\left(\bos{a},\bos{b}\right)$
implies that $Q$ is a derivation of the product and
$$
-(-1)^{|a|}\left(a,b\right)
=K^{(1)}\left(a\cdot b\right)- K^{(1)}a\cdot b +(-1)^{|a|}a\cdot K^{(1)}b,
$$
\item
the bracket $(\bullet,\bullet)$ is a derivation of the product by definition.
\end{enumerate}
\end{example}

\begin{remark}
A differential BV algebra $\big(\sC, \Delta, Q,\,\cdot\,\big)$ or BV algebra
is not a BV QFT algebra, though we can make one out of it in particular fashion and
might be able to extract the original. Once constructing a BV QFT algebra
$\big(\sC[[\hbar]], Q-\hbar\Delta,\,\cdot\,\big)$ its natural automorphism
group send $Q-\hbar\Delta$ to 
$$
Q-\hbar\left(\Delta - [Q,g^{(1)}] \right)
-\hbar^{2}\left([\Delta, g^{(1)}]-[Q, g^{(2)}] +g^{(1)}Q g^{(1)}- g^{(1)}g^{(1)}Q\right)
+\cdots
$$
so that we certainly do not have a differential BV algebra by extracting coefficients of
above.
\end{remark}

\begin{example}
Let $\big(\sC[[\hbar]],\bos{K},\,\cdot\,\big)$ be a BV QFT algebra with
the descendant DGLA 
$$
\big(\sC[[\hbar]],\bos{K},\left(\hbox{ },\hbox{ }\right)\big).
$$
Let $(V, \cdot^{\pr})$ be a $\Z$-graded associative algebra over $\Bbbk$.
Then $\big((\sC\otimes V)[[\hbar]],\bos{K},\,\cdot\,\big)$,
with abusing notations,
is also a BV QFT algebra, where $\bos{K}$ means that 
$\bos{K}\otimes 1$ and the product $\cdot$ means that
$(\bos{a}\otimes \bos{\a})\cdot (\bos{b}\otimes \bos{\b})=
(-1)^{|\bos{\a}||\bos{b}|}\bos{a}\cdot\bos{b}\otimes \bos{\a}\cdot^{\pr}\bos{\b}$,
with the descendant DGLA 
$$
\big((\sC\otimes V)[[\hbar]],\bos{K},\left(\hbox{ },\hbox{ }\right)\big),
$$
where the bracket $\left(\hbox{ },\hbox{ }\right)$ means that
$\left(\bos{a}\otimes \bos{\a}, \bos{b}\otimes \bos{\b}\right)=
(-1)^{|\bos{\a}|(|\bos{b}|+1)}\left(\bos{a},\bos{b}\right)\otimes \bos{\a}\cdot^{\pr}\bos{\b}$.
\end{example}

\newsec{Observables and Expectation Values/Quantum Complex and Quantum Homotopy Invariants}

Throughout this section we consider the cochain complex 
$\big(\sC[[\hbar]],\bos{K}\big)$ and its classical limit $\big(\sC, Q\big)$
 in a BV QFT algebra   $\big(\sC[[\hbar]], \bos{K}, \hbox{ }\cdot\hbox{ }\big)$. 
 We denote the cohomology
of the complex $\left(\sC, Q\right)$ by $H$ and the cohomology class of an element
$O \in \sC$ satisfying $Q O=0$ by $[O]$. The cochain complex 
$\left(\sC[[\hbar]],\bos{K}\right)$ is defined modulo natural automorphism 
$\bos{g}= 1 +\hbar g^{(1)}+\hbar^{2}g^{(2)}+\cdots :\sC[[\hbar]]\longrightarrow
\sC[[\hbar]]$, where $g^{(\ell)}$, for $\ell=1,2,3,\cdots$, is a $\Bbbk$-linear map
on $\sC$ into $\sC$ of ghost number zero, such that it is gauge equivalent to
the cochain complex $\left(\sC[[\hbar]],\bos{K}^{\pr}=\bos{g}\bos{K}\,\bos{g}^{-1}\right)$.
Such automorphism preserves the classical parts of both $\sC[[\hbar]]$ and $\bos{K}$,
i.e., $\bos{K}\bigl|_{\hbar=0}= \bos{K}^{\pr}\bigl|_{\hbar=0}=Q$. Thus both the cochain
complex $(\sC,Q)$ and its cohomology $H$ are fixed. 

The goal of this section is twofold;
we are going to formalize (i) the procedure of constructing observables and (ii) the procedure
of evaluating   expectation values of observables. 

\begin{itemize}
\item
For the first goal, we are going to
build a certain  cochain complex $(H[[\hbar]],\bos{\k})$ on 
$H[[\hbar]]:= H\otimes_{\Bbbk}\Bbbk[[\hbar]]$, 
where $\bos{\k}=\hbar\k^{(1)}+\hbar^{2}\k^{(2)}+\cdots$ is a sequence of
$\Bbbk$-linear maps on $H$ into itself parametrized by $\hbar$
such that $\bos{\k}\bigl|_{\hbar=0}=0$ and $\bos{\k}^{2}=0$. 
Together with $(H[[\hbar]],\bos{\k})$, we are going
to build certain  map $\bos{f}$ on $H[[\hbar]]$ into $\sC[[\hbar]]$ of ghost number zero
such that
\begin{enumerate}
\item the $\Bbbk[[\hbar]]$-linear map
$\bos{f}=f+\hbar f^{(1)}+\hbar^{2} f^{(2)}+\cdots$ is a sequence of
$\Bbbk$-linear maps on $H$ into $\sC$  parametrized by $\hbar$,
\item 
the classical part $f$ of $\bos{f}$ is a cochain map on $(H,0)$ into $(\sC,Q)$
which induces the identity map on the cohomology $H$,
\item the $\Bbbk[[\hbar]]$-linear map $\bos{f}$ is a cochain map
on $\left(H[[\hbar]],\bos{\k}\right)$ into $\left(\sC[[\hbar]], \bos{K}\right)$;
$$
\bos{K}\bos{f}=\bos{f}\bos{\k}.
$$
\end{enumerate}
Both $\bos{\k}$ and $\bos{f}$ are defined modulo the natural automorphism on $H[[\hbar]]$.
Also the map $\bos{f}$ shall be defined up to natural notion
of  quantum homotopy.
The map $f:H\longrightarrow \sC$ corresponds to assigning a classical observable
to each and every cohomology class in $H$ which exists always.
The map $\bos{f}:H[[\hbar]]\longrightarrow \sC[[\hbar]]$ 
corresponds to an attempt to assign a quantum observable
to each and every cohomology class in $H$, which is not away possible. 
One of our conclusion shall be that a classical observable $O$,
that is
$O\in \sC$ and $Q O=0$,  can be extended to a quantum observable
if and only if $\bos{\k}\left(\left[O\right]\right)=0$.

\item
For the second goal we are going to examine certain  
$\Bbbk[[\hbar]]$-linear map $\bos{c}$ on $\sC[[\hbar]]$ into $\Bbbk[[\hbar]]$,
where 
\begin{enumerate}
\item the $\Bbbk[[\hbar]]$-linear map
 $\bos{c}=c^{(0)}+\hbar c^{(1)}+\hbar^{2} c^{(2)}+\cdots$ is a sequence of
$\Bbbk$-linear maps on $H$ into $\Bbbk$  parametrized by $\hbar$, 
\item the $\Bbbk[[\hbar]]$-linear map
$\bos{c}$ is a cochain map from  $\left(\sC[[\hbar]],\bos{K}\right)$ 
into $\left(\Bbbk[[\hbar]],0\right)$;
$$
\bos{c}\,\bos{K} =0,
$$
which is defined up to a natural notion
of  quantum homotopy.
\end{enumerate}
The map $\bos{c}$ shall correspond to a Batalin-Vilkovisky-Feynman path integral
and variations of $\bos{c}$ within a quantum homotopy class
correspond to continuous variations of gauge fixing condition. Then we shall associate
quantum expectation value to a classical observable $O$ by
$\bos{c}\circ\bos{f}\left(\left[O\right]\right) \in \Bbbk[[\hbar]]$,
which depends only on the cohomology class of classical observable and
is a quantum homotopy invariant if and only if $\bos{\k}\left(\left[O\right]\right)=0$.
In other word the expectation value does not depends on continuous variations 
of gauge fixing condition if the classical observable is extendable to a quantum observable.
Also the 
$\Bbbk[[\hbar]]$-linear map $\bos{c}\circ\bos{f}: H[[\hbar]]\longrightarrow \Bbbk[[\hbar]]$
is invariant under the automorphism of $\sC[[\hbar]]$.
\end{itemize}
Our formalizations are not against to the lore of QFT but clarifications of them.

\subsection{Classical to Quantum Observables}

We call  a representative $O \in \sC$ of a cohomology class of $(\sC, Q)$
a {\it classical observable}, and two classical
observables (classical physically)-equivalent if they belong to the same cohomology class.
Thus the set of equivalence classes of classical observables
is just the cohomology $H$ of the classical complex $(\sC, Q)$.
By a {\it quantum observable} $\bos{O}$ 
we mean a representative of a  nontrivial cohomology class of the quantum complex  
$(\sC[[\hbar]], \bos{K})$. We say two quantum observables are physically equivalent
if they belong to the same cohomology class of the complex  
$(\sC[[\hbar]], \bos{K})$. In this paper we are not interested in general quantum observables
but in quantum observables which are extended from classical observables.

We say  a classical observable $O\in \sC^{|O|}$
is  extendable to a quantum observable if there is an 
$\bos{O}\in \sC[[\hbar]]^{|O|}$ such that
 $\bos{O}\bigr|_{\hbar=0} = O$ and $\bos{K}\bos{O}=0$. 
 We call such an element
$\bos{O}$  an extension
of the classical observable $O$ to a quantum observable. Note that such an $\bos{O}$
does not belong to $\hbox{Im} \bos{K}$  unless the classical observable $O$ is trivial
- assume that $\bos{O} = \bos{K}\bos{\La}$ and $O$ is nontrivial,
then $O = Q \La$ for $\La =\bos{\La}\bigr|_{\hbar=0}$, which is a contradiction. 
Let $O\in \sC^k$ is  a representative of a cohomology
class of $(\sC, Q)$ which admits an extension to quantum
observable $\bos{O}$. Then any other representative $O^\pr \in \sC^k$ of 
the cohomology class $[O]$ of $O$ also has an extension
to a quantum observable $\bos{O^\pr}$; let 
$O^\pr = O + Q \l$ for some $\l \in \sC^{k-1}$, then $\bos{O^\pr} = \bos{O} + \bos{K}\l$
is an extension of $O^\pr$ to a quantum observable - $\bos{K}\bos{O^\pr} =0$ and 
$\bos{O^\pr}\bigl|_{\hbar=0} = \bos{O}\bigl|_{\hbar=0} + (\bos{K}\l)\bigl|_{\hbar=0}=
O + Q\l =O^\pr$. 
Let $O\in \sC^k$ is  a representative of a cohomology
class of  $(\sC, Q)$ and assume that $O$  does not admit an extension to quantum
observable. Then any other representative $O^\pr \in \sC^k$ of 
the  cohomology class $[O]$ of $O$ also does not admit an extension
to quantum observable; assume that a classical observable $O$
does not admit an extension to a quantum observable while
$O^\pr = O + Q \l$ for some $\l \in \sC^{k-1}$ has an extension  
to a quantum observable $\bos{O}^\pr$. Then $\bos{O} = \bos{O^\pr} - \bos{K}\l$
is an extension of $O$ to a quantum observable - $\bos{K}\bos{O} =0$ and 
$\bos{O}\bigl|_{\hbar=0} = \bos{O^\pr}\bigl|_{\hbar=0} - (\bos{K}\l)\bigl|_{\hbar=0}=
O + Q\l -Q\l =O$, which is a contradiction. 
Thus the existence of extension of a classical observable depends on its
classical cohomology class.  

So  an extension of a classical observable $O$ to quantum observable 
is an association $a \in H^{|a|}$ to $\bos{O}\in \sC[[\hbar]]^{|a|}$ such
that $\bos{K}\bos{O}=0$ and  the classical cohomology class $[O]$ of 
the projection $O=\bos{O}\bigr|_{\hbar=0}$ of $\bos{O}$ to $\sC$ is $a$.  
We may say two such extensions $\bos{O}$ and $\bos{O}^{\pr}$ of a classical observable is
equivalent if $\bos{O}^{\pr}-\bos{O} = \bos{K}\bos{\La}$ for some 
$\bos{\La}\in \sC[[\hbar]]^{|O|-1}$. We, however, note that extensions of a classical observable
can be much more arbitrary. Here is a simple demonstration: Let $\bos{X}$ and
$\bos{Y}$ be extensions of two classical observables $X$ and $Y$ such that $[X]\neq [Y]$.
Then both $\bos{X} +\hbar \bos{X}$ and $\bos{X} +\hbar\bos{Y}$  are extensions of $X$,
and all three extensions of $X$ are not equivalent. Clearly, there are infinitely many
variations of above examples with the similar feature. But all of such infinite possible
examples are silly things to take seriously.

In the following subsection we are going to develop obstruction theory for extending
classical observable to quantum observables together with dealing every ambiguity
in details. The basic strategy is to work with every equivalence class of classical observables
simultaneously.

\subsection{Obstructions and Ambiguities}

We begin with recalling some elementary terminology from homological algebra.
Let $(V, d_{V})$ and $(W, d_{W})$ be two cochain complexes over $\Bbbk$
and let $H_{V}$ and $H_{W}$ denote their cohomology. 
A cochain map $f$ on $(V,d_{V})$ into $(W, d_{W})$ is a degree preserving
$\Bbbk$-linear map $f: V^{\bullet}\longrightarrow W^{\bullet}$, which commutes
with the differentials, $ f d_{V}= d_{W}f$. It is understood that the map $f$ denotes
collectively for every map defined for each degree, say $f_{j}:V^{j}\longrightarrow W^{j}$, 
and it is zero map whenever  its source or target is trivial. A cochain map $f$ induces
a well-defined map on $H_{V}$ into $H_{W}$ since it sends $\Ker d_{V}$ to $\Ker d_{W}$
and $\hbox{Im }d_{V}$ to $\hbox{Im } d_{W}$. A cochain map is called quasi-isomorphism
if it induces an isomorphism between the cohomologies. There is an obvious way of
constructing a cochain map $f$ from arbitrary  $\Bbbk$-linear map 
$s: V^{\bullet}\longrightarrow W^{\bullet-1}$ of degree $-1$ by the formula 
$f= s d_{V}+ d_{W} s$. Such a cochain
map is called homotopic to zero and denoted by $f\,\sim\, 0$. It is clear that a cochain
map $f \,\sim\, 0$ vanishes on cohomology. 
We say two cochain maps $f$ and $f^{\pr}$ are homotopic 
and denote $f\,\sim\,f^{\pr}$ if $f^{\pr}-f\,\sim\,0$. Cochain homotopy is an equivalence
relation and the equivalence class of a cochain map is called homotopy type of the cochain map.
Whenever we consider a cochain map it is understood to be defined up to homotopy.

Let $a$ be an element of the cohomology $H$ of the complex $(\sC,Q)$. 
We say an element $O \in \sC^{|a|}$ a representative of $a$ if $Q O=0$ and
the cohomology class $[O]$ of $O$ is $a$.  It follows that such a choice of representative 
is defined modulo $Q$-exact term.
Choosing a representative for each and every element in the cohomology $H$ of 
$(\sC, Q)$ such that $\Bbbk$-linearity is preserved
defines a cochain map $f:H\longrightarrow \sC$ from the cochain complex
$(H,0)$ with zero differential to $(\sC, Q)$, i.e., $Q f=0$, which induces an isomorphism
of the cohomology since $H$ is the own cohomology of $(H,0)$ as well as the cohomology of
$(\sC,Q)$.  Thus $f$ is a quasi-isomorphism.
Furthermore the induced isomorphism on $H$ must be the identity map
since $[f(a)]=a$ for every $a\in H$ by definition.
The ambiguity in choosing representatives  corresponds to homotopy of the cochain map
$f$: Let $f^{\pr}=f + Q s$, where $s$ is a $\Bbbk$-linear map $s:H^{\bullet}\longrightarrow
\sC^{\bullet -1}$ of ghost number $-1$. Then $Q f^{\pr}=0$ and $f^{\pr}$ also induces
the identity map on the cohomology $H$, since $Q$ vanishes on $H$.  
Thus the map $f$ 
is unique up to homotopy. 

\begin{remark}
Let $\g \in \left(\Ker Q\cap \sC^{|\g|}\right)$, i.e., $\g$ has the ghost number $|\g|$
and satisfies $Q\g=0$. By taking cohomology class of $\g$ we have $[\g]\in H^{|\g|}$.
Then apply the map $f$ to $\left[\g\right]$ to obtain an element $f\left(\left[\g\right]\right)$
in $\sC^{|\g|}$ satisfying $Q f\left(\left[\g\right]\right)=0$.
Since $\left[f\left(\left[\g\right]\right)\right] =\left[\g\right]$, it follows that
$\g=f\left(\left[\g\right]\right)+ Q \beta$ for some $\beta \in \sC^{|\g|-1}$.
Now consider any $\Bbbk$-linear map $g: H^{\bullet}\longrightarrow \sC^{\bullet+|g|}$
which image belongs to $\Ker Q$. Such a map can be identified with composition of
a linear map $\xi:H^{\bullet}\longrightarrow H^{\bullet +|g|}$ with
the map $f:H^{\bullet} \longrightarrow \sC^{\bullet}$ up to homotopy;
$$
g - f\xi = Q \eta,
$$
for some $\Bbbk$-linear map $\eta: H^{\bullet}\longrightarrow \sC^{\bullet +|g|-1}$.
The linear map $\xi:H^{\bullet}\longrightarrow H^{\bullet +|g|}$ can be constructed
by taking the homotopy type of the map $g$. To be more explicit consider any $a \in H$
and its image $g(a)$ of the map $g: H^{|a|}\longrightarrow \sC^{|a|+|g|}$ so that $Q g(a)=0$.
By taking the cohomology class  of $g(a)$ we obtain $[g(a)] \in H^{|a|+|g|}$. Doing
this for each and every elements in $H$ defines a linear map $\xi:H^{\bullet}\longrightarrow
H^{\bullet +|g|}$ such that $\xi(a):=\left[g(a)\right]$. We may simply say that
$\xi$ is the cohomology class $[g]$ of $g$.
Now we compose $\xi$ with
$f$ to obtain $f\xi: H^{\bullet}\longrightarrow \sC^{\bullet +|g|}$. Since $f$ is a map
choosing a representative of each and every cohomology class, it follow that
$f\xi(a)$ and $g(a)$ belongs to the same cohomology class.
\end{remark}

The differential $0$
in $(H,0)$ is induced from the differential $Q$ in $(\sC, Q)$ and is zero since
$H$ is the $Q$-cohomology.
The cohomology  $H$ of the  complex $(\sC,Q)$ has more structure
induced from the cochain complex $(\sC[[\hbar]],\bos{K})$,
which is, modulo its natural automorphism,
nothing but
an infinite sequence $Q,K^{(1)}, K^{(2)},\cdots$ of $\Bbbk$-linear maps
on $\sC$ into itself with ghost number $1$ satisfying the following infinite
sequence of relations;
\eqn\toera{
\eqalign{
Q^{2}&=0,\cr
Q K^{(n)}+K^{(n)}Q+\sum_{\ell=1}^{n-1}K^{(n-\ell)} K^{(\ell)}&=0 \hbox{ for all } n=1,2,\cdots .
}
}
On $H$ the differential $Q$ vanishes, while
the $\Bbbk$-linear maps
$K^{(1)},K^{(2)},\cdots$ on $\sC$ shall induce 
certain infinite sequence $\k^{(1)},\k^{(2)}, \cdots$ of $\Bbbk$-linear
maps on $H$ into itself with ghost number $1$. It is natural to expect that 
such the sequence $\k^{(1)},\k^{(2)}, \cdots$  of $\Bbbk$-linear maps may satisfy
the following infinite
sequence of relations;
\eqn\toerb{
\eqalign{
\sum_{\ell=1}^{n-1}\k^{(n-\ell)} \k^{(\ell)}&=0 \hbox{ for all } n=1,2,\cdots .
}
}
It is also natural to consider the expression
$\bos{\k}:= \hbar \k^{(1)}+\hbar^{2}\k^{(2)}+\cdots$,
the sequence $\k^{(1)},\k^{(2)}, \cdots$ of $\Bbbk$-linear maps on $H$ into $H$ 
parametrized by $\hbar$, and $H[[\hbar]]:=H\otimes_{k}\Bbbk[[\hbar]]$ so that
$\bos{\k}$ is a $\Bbbk[[\hbar]]$-linear map on $H[[\hbar]]$ into $H[[\hbar]]$
with ghost number $1$. Then the expected relations in \toerb\ is summarized
by $\bos{\k}^{2}=0$ and $\bos{\k}\bigr|_{\hbar=0}=0$, 
so that $\left(H[[\hbar]],\bos{\k}\right)$ is a cochain
complex over $\Bbbk[[\hbar]]$ with the classical limit $(H,0)$.
There is natural automorphism on $H[[\hbar]]$
generated by an arbitrary infinite sequence
 $\bos{\xi}=1+\hbar\xi^{(1)}+\hbar\xi^{(2)}
+\cdots$ of $\Bbbk$-linear maps, parametrized by $\hbar$, on $H$ into $H$ 
with ghost number $0$ satisfying $\bos{\xi}\bigr|_{\hbar=0}=1$.
 So it is also natural to expect that
$\bos{\k}:= \hbar \k^{(1)}+\hbar^{2}\k^{(2)}+\cdots$ is defined up the gauge symmetry
$\bos{\k}^{\pr}= \bos{\xi}^{-1}\bos{\k}\bos{\xi}$. 
Then,  $\k^{(1)}$ must be invariant, $\k^{\pr(2)}$ is sent to 
$\k^{\pr(2)}=\k^{(2)} + \k^{(1)}\xi^{(1)}-\xi^{(1)}\k^{(1)}$ etc.

%


Now we are going to construct $\bos{\k}=\hbar\k^{(1)}+\hbar^{2}\k^{(2)}+\cdots$ 
and morphism
$\bos{f}=f +\hbar f^{(1)}+\hbar^{2} f^{(2)}+\cdots :H[[\hbar]]\longrightarrow
\sC[[\hbar]]$ satisfying $\bos{K}\bos{f} = \bos{f}\bos{\k}$ with taking care of all
ambiguities. Our construction and proof is inductive using the identification
$$
\Bbbk[[\hbar]]=\lim_{\leftarrow}\left( \Bbbk[\hbar]\bigr/ \hbar^{n}\Bbbk[\hbar]\right) 
\hbox{ as } n\longrightarrow \infty.
$$

To see the leading part of it, we 
write down first two leading terms for the condition $\bos{K}^{2}=0\mod \hbar^{2}$;
\eqn\thea{
\eqalign{
Q^{2}&=0,\cr
Q K^{(1)} + K^{(1)}Q&=0.\cr
}
}
The second relation in \thea\ implies that $K^{(1)}$ induces unique $\Bbbk$-linear map 
of ghost number $1$ on
$H$ into itself;
$$
\k^{(1)}:H^\bullet \longrightarrow H^{\bullet +1},
$$ 
since $K^{(1)}$ sends $\Ker Q$ to $\Ker Q$
and $\hbox{Im } Q$ to $\hbox{Im }Q$; (i) let $\eta \in \Ker Q$, $Q\eta=0$,
then $K^{(1)}(\eta) \in \Ker Q$ since $Q K^{(1)}\eta = - K^{(1)}Q\eta=0$, (ii)
let $\eta \in \hbox{Im }Q$, that is $\eta = Q\l$, then
$K^{(1)}(\eta) \in Im Q$ since $K^{(1)}\eta = K^{(1)}Q\l= - Q (K^{(1)}\l)$.

Now we consider the role of $f$. It is easy to show that
$K^{(1)} f -f\k^{(1)} \in \Ker Q$, since $Q(K^{(1)} f)= -K^{(1)} (Q f)=0$
and $Q f=0$. It can be also shown that
$(K^{(1)} f -f\k^{(1)} ) \subset Im Q$ using a contradiction. 
Assume that for an $a \in H^i$ there exists some $b \in H^{i+1}$ with
$b\neq 0$
such that
$$
(K^{(1)} f -f\k^{(1)} )(a) = f(b) - Q \l,
$$
with some $\l \in \sC^{i}$. We take the cohomology class in the both hand sides
of the above to get $\left[(K^{(1)} f -f\k^{(1)} )(a)\right] = \left[f(b)\right]$,
which implies that   $\k^{(1)} (a)-\k^{(1)} (a)=b$ since
$K^{(1)}=\k^{(1)}$ and $f$ is the identity map on $H$. Then $b$ must be zero,
which is a contradiction.
Now may declare a solution $\l$ of
$(K^{(1)} f -f\k^{(1)} )(a) = - Q \l$ for each and every $a \in H$ as the image  
$f^{(1)}(a)$ of another map
$f^{(1)}:H^{\bullet}\longrightarrow \sC^{\bullet}$, so that we have
\eqn\modha{
K^{(1)}f -f\k^{(1)}=- Q f^{(1)}.
}
Define 
$\bos{f}:=f +\hbar f^{(1)}\mod \hbar^{2}$
and $\bos{\k} :=\hbar\k^{(1)}\mod \hbar^{2}$.
Then we have $\bos{K}\,\bos{f}= \bos{f}\,  \bos{\k}\mod \hbar^{2}$ 
and $\bos{\k}^{2}=0\mod \hbar^{2}$, 
which summarize
\eqn\modhb{
\eqalign{
Q f&=0,\cr 
K^{(1)}f+ Q f^{(1)}&=f \k^{(1)}.
}
}
We should emphasis that $\bos{\k}\bigr|_{\hbar=0}=0$, and, thus
the condition $\bos{\k}^{2}=0\mod \hbar^{2}$ is vacuous.

Now we should examine possible ambiguity in the above procedure.
First of all the map $f:H^{\bullet}\longrightarrow \sC^{\bullet}$ is
defined up to homotopy, i.e., modulo $\hbox{Im } Q$. Let $f^{\pr}$ be
a map defined by another choice of representative for each and every
element in $H$. Then  
$$
f^{\pr}=f+ Q s^{(0)},
$$ 
for some arbitrary $\Bbbk$-linear map 
$s^{(0)}:H^{\bullet}\longrightarrow \sC^{\bullet-1}$.
Secondly $f^{(1)}$ in \modha\ is defined
modulo $\Ker Q$. To deal with such ambiguity, let's first repeat
the same procedure as above using the map $f^{\pr}=f + Q s^{(0)}$ instead of $f$.
We shall end up 
\eqn\modhab{
\eqalign{
Q f^{\pr}&=0,\cr
K^{(1)}f^{\pr} + Q f^{\pr(1)}&=f^{\pr}\k^{\pr(1)},
}
}
where $f^{\pr(1)}$ is a $\Bbbk$-linear map on $H$ to $\sC$ defined
modulo $\Ker Q$ and $\k^{\pr(1)}=\k^{(1)}$.
Rewriting the $2$nd equation above, by substituting $f^{\pr}=f+ Q s^{(0)}$,
 as follows;
$$
f\k^{(1)}= - Q s^{(0)}\k^{(1)} -Q K^{(1)}f +K^{(1)}f+ Q f^{\pr(1)}, 
$$
where we have used $K^{(1)}Q=-Q K^{(1)}$, we can compare with \modhb\
to conclude that 
$$
Q w^{(1)}=0,
$$ 
where
$$
w^{(1)}:=f^{(1)\pr} - f^{(1)} - K^{(1)}f - s^{(0)}\k^{(1)}.
$$
So we may have some controls of ambiguity; $w^{(1)}$ could be an arbitrary
$\Bbbk$-linear map on $H$ into $\sC$ of ghost number zero, but its
image is contained in $\Ker Q\cap \sC$. 
By taking cohomology class of $w^{(1)}(a)$ for each and every
$a \in H$, we obtain a $\Bbbk$-linear map on $H$ into itself;
$$
\xi^{(1)}:H\longrightarrow H.
$$
where $\xi^{(1)}(a)=\left[w^{(1)}(a)\right]$ for each and every $a\in H$.
Note that $\xi^{(1)}$ is also arbitrary. Now we compose $\xi^{(1)}$ with
$f$ to get a $\Bbbk$-linear map $f\xi^{(1)}:H\longrightarrow \sC$.
Since $f$ is a map choosing a representative for each and every $a \in H$,
we know that $\left[f\xi^{(1)}(a)\right]=\left[w^{(1)}(a)\right]$. Thus
$$
w^{(1)}= f\xi^{(1)} + Q s^{(1)}
$$
for some $\Bbbk$-linear map $s^{(1)}$ on $H$ into $\sC$ of ghost number zero.
Combining with the definition of $w^{(1)}$
we have $f^{\pr(1)}=f^{(1)} + f\xi^{(1)}
+Q s^{(1)}+ K^{(1)} s^{(0)}+ s^{(0)}\k^{(1)}$.

Now we collect everything together
to have the following general forms;
\eqn\modhaa{
\k^{\pr (1)}=\k^{(1)},
\qquad
\left\{\eqalign{ 
f^{\pr}&=f+  Q s^{(0)}
,\cr
f^{\pr(1)}&=f^{(1)}+ f\xi^{(1)}
+Q s^{(1)}+ K^{(1)} s^{(0)}+ s^{(0)}\k^{(1)}.
}\right.
}
The whole things can be written in more suggestive way.
Once we  define $\bos{\xi}:= 1+\hbar \xi^{(1)} \mod \hbar^{2}$, 
the relations in \modhaa\
is
 \eqn\modhc{
\eqalign{
\bos{\k}^{\pr}&=\bos{\xi}^{-1}\bos{\k}\bos{\xi}\mod \hbar^{2}
,\cr
\bos{f}^{\pr}&= \bos{f}\,\bos{\xi}
+ \bos{K}\,\bos{s} + \bos{s}\,\bos{\k}^{\pr}\mod \hbar^{2}.
}
}
It is obvious that $\bos{\k}^{\pr2}=0\mod \hbar^{2}$
and $\bos{K}\bos{f}^{\pr}=\bos{f}^{\pr}\bos{\k}^{\pr}\mod \hbar^{2}$,
which summarize \modhab. Thus every  arbitrariness represented by
$\bos{\xi}\mod\hbar^{2}$ is from the natural automorphism on 
$H[[\hbar]]$ modulo $\hbar^{2}$.

Now we consider the problem of  a classical observable
$O$ to a quantum observable.
Note that  $f^{\pr}([O])= f([O])+ Q s^{(0)}([O])$
with arbitrary $s^{(0)}$ belongs to the same $Q$-cohomology class
of $O$ and gives every possible representative of the class by variations
of $s^{(0)}$. Also  $K^{(1)}O$ and $K^{(1)}f^{\pr}([O])$
belongs to the same $Q$-cohomology class, which is $\k^{(1)}([O])\in H$;
$$
\left[K^{(1)}O\right] =\left[K^{(1)}f^{\pr}([O])\right]
=\k^{(1)}([O]).
$$
But the classical observable $O$ is extendable to quantum observable modulo $\hbar^{2}$
if and only if $K^{(1)}O$ is $Q$-exact, i.e., $K^{(1)}O=-Q O^{(1)}$ for some $O^{(1)}$
so that $\bos{O}:=O +\hbar O^{(1)}$ satisfies $\bos{K}\bos{O}=0\mod\hbar^{2}$.
Thus the necessary and sufficient condition for that is $\k^{(1)}([O])=0$. Assuming
so, $\bos{f}^{\pr}([O]) =\bos{f}\left(\bos{\xi}([O])\right) + \bos{K}\bos{s}([O])\mod \hbar^{2}$ 
gives every possible extension by variations of
$\bos{s}\mod \hbar^{2}$ and $\bos{\xi}\mod \hbar^{2}$.

For the next order, we 
consider the following map;
$$
g^{(2)}:= K^{(2)}f + K^{(1)}f^{(1)} -f^{(1)}\k^{(1)}
$$
on $H$ into $\sC$ of ghost number $1$.
We claim that 
$$
Q g^{(2)}=- f\k^{(1)}\k^{(1)},
$$
which is a consequence of $\bos{K}\bos{f}=\bos{f}\bos{\k}\mod\hbar^{2}$ and 
$\bos{K}^{2}=0\mod \hbar^{3}$. By taking 
the $Q$-cohomology class to the above relation
we have  $\left[f\k^{(1)}\k^{(1)}\right]=0$, which implies that
$$
\k^{(1)}\k^{(1)}=0,
$$ 
since $f$ induces the identity map on $H$.
It, then, follows that $Q g^{(2)}=0$. Thus the image of the map 
$g^{(2)}:H^{\bullet}\longrightarrow \sC^{\bullet+1}$ is contained in $\Ker Q$.
By taking the $Q$-cohomology class of $g^{(2)}$ we obtain a $\Bbbk$-linear map
$\k^{(2)}:H^{\bullet}\longrightarrow H^{\bullet+1}$. Composing
it with $f:H^{\bullet}\longrightarrow C^{\bullet}$, we have the map $f\k^{(2)}:H^{\bullet}
\longrightarrow \sC^{\bullet+1}$, which belongs to the same cohomology class of
$g^{(2)}$. Thus
\eqn\modhd{
g^{(2)}=f\k^{(2)} - Q f^{(2)}
}
for some $\Bbbk$-linear map $f^{(2)}$ on 
$H$ into $\sC$ of ghost number $0$, which is defined modulo $\Ker Q$.
\modhd. Combined with the definition of $g^{(2)}$, \modhd\ gives
$$
K^{(2)}f + K^{(1)}f^{(1)}+Q f^{(2)}=f\k^{(2)} +f^{(1)}\k^{(1)}
$$
Now we define 
$\bos{f}:=f +\hbar f^{(1)}+\hbar^{2}f^{(2)}\mod \hbar^{3}$
and $\bos{\k} :=\hbar\k^{(1)}+\hbar^{2}\k^{(2)}\mod \hbar^{3}$.
Then we have $\bos{K}\,\bos{f}= \bos{f}\,  \bos{\k}\mod \hbar^{3}$ 
and $\bos{\k}^{2}=0\mod \hbar^{3}$, 
which summarize
\eqn\modhbi{
\k^{(1)}\k^{(1)}=0,\qquad
\eqalign{
Q f&=0,\cr 
K^{(1)}f+ Q f^{(1)}&=f \k^{(1)},\cr
K^{(2)}f + K^{(1)}f^{(1)}+Q f^{(2)}&=f\k^{(2)} +f^{(1)}\k^{(1)}.
}
}

We also consider the ambiguities from the previous steps \modhaa.
Let
$$
g^{\pr(2)}:= K^{(2)}f^{\pr} + K^{(1)}f^{\pr(1)} -f^{\pr(1)}\k^{(1)}
$$
After a direct computation we obtain 
\eqn\mora{
\eqalign{
g^{\pr(2)}=
& g^{(2)} -Q\left(
 f^{(1)}\xi^{(1)}
+K^{(1)} s^{(1)}
+K^{(2)}s^{(0)}
+s^{(1)}\k^{(1)}\right)
\cr
&
+f\left(\k^{(1)}\xi^{(1)}- \xi^{(1)}\k^{(1)}\right).
}
}
It follows that $Q g^{\pr(2)}=0$ since $Q g^{(2)}=0$ and $Q f=0$.
Thus we have
$$
\k^{\pr(2)}:=\left[g^{\pr (2)}\right]= \k^{(2)} +\k^{(1)}\xi^{(1)}- \xi^{(1)}\k^{(1)}
$$
and 
\eqn\morb{
g^{\pr(2)}
=f^{\pr}\k^{\pr(2)}- Q f^{\pr (2)}
}
for some $f^{\pr(2)}$ defined modulo $\Ker Q$.

Now we want to compare $f^{\pr(2)}$ with $f^{(2)}$.
We begin with rewriting \morb\ as follows;
$$
\eqalign{
g^{\pr(2)}
&=f^{\pr}\k^{\pr(2)}- Q f^{\pr (2)}
\cr
&=f \k^{(2)} 
+f\left(\k^{(1)}\xi^{(1)}- \xi^{(1)}\k^{(1)}\right)
-Q\left( f^{\pr(2)}-s^{(0)}\k^{\pr 2}\right).
}
$$
while we recall that
$$
g^{(2)}= f\k^{(2)} - Q f^{(2)}
$$
Since both equations above contain the same term $f\k^{(2)}$,
we have
$$
g^{\pr(2)}
-f\left(\k^{(1)}\xi^{(1)}- \xi^{(1)}\k^{(1)}\right)
+Q\left( f^{\pr(2)}-s^{(0)}\k^{\pr 2}\right)
=g^{(2)}+ Q f^{(2)}.
$$
Now we use the relation \mora\ to conclude
that  $Q w^{(2)}=0$, where
$$
w^{(2)}:=f^{\pr(2)}- f^{(2)} -\left(
f^{(1)}\xi^{(1)}
+K^{(1)} s^{(1)}
+K^{(2)}s^{(0)}
+s^{(1)}\k^{(1)}
+s^{(0)}\k^{\pr(2)}
\right)
$$
Then by taking cohomology  we have a $\Bbbk$-linear
map $\left[w^{(2)}\right]:H^{\bullet}\longrightarrow H^{\bullet}$, 
which is an arbitrary $\Bbbk$-linear map. 
So we introduce a new ``ghost variable'' 
$\xi^{(2)}:H^{\bullet} \longrightarrow H^{\bullet}$ for the arbitrariness.
Then we have $w^{(2)}=f\xi^{(2)}+ Q s^{(2)}$ for some $\Bbbk$-linear
map $s^{(2)}$ on $H$ into $\sC$ of ghost number $0$.
Finally we have
$$
f^{\pr(2)}= f^{(2)}+f\xi^{(2)} + 
f^{(1)}\xi^{(1)}
+Q s^{(2)}
+K^{(1)} s^{(1)}
+K^{(2)}s^{(0)}
+s^{(1)}\k^{(1)}
+s^{(0)}\k^{\pr(2)}
$$
Now we denote $\bos{s}=s^{(0)}+\hbar s^{(1)}+\hbar^{2}s^{(2)}\mod\hbar^{3}$
and $\bos{\xi}=1+\hbar\xi^{(1)}+\hbar^{2}\xi^{(2)}\mod\hbar^{3}$. Then
every ambiguity is summarized by
$$
\eqalign{
\bos{\k}^{\pr}&=\bos{\xi}^{-1}\bos{\k}\bos{\xi}\mod \hbar^{3},\cr
\bos{f}^{\pr}&= \bos{f}\bos{\xi} + \bos{K}\bos{s} + \bos{s}\bos{\k}^{\pr}
\mod \hbar^{3},
}
$$
and we have $\bos{\k}^{\pr 2}=0\mod\hbar^{3}$ and 
$\bos{K}\bos{f}^{\pr}=\bos{\k}^{\pr}\bos{f}^{\pr}\mod \hbar^{3}$.
Thus every  arbitrariness represented by
$\bos{\xi}\mod\hbar^{3}$ is from the natural automorphism on 
$H[[\hbar]]$ modulo $\hbar^{3}$. Also a classical observable
$O$ is extendable to a quantum observable modulo $\hbar^{3}$
if and only if $\bos{\k}([O])=0\mod \hbar^{3}$ and
$\bos{f}^{\pr}([O]) =\bos{f}\left(\bos{\xi}([O])\right) + \bos{K}\bos{s}([O])\mod \hbar^{3}$ 
gives every possible extension by variations of
$\bos{s}\mod \hbar^{3}$ and $\bos{\xi}\mod \hbar^{3}$.

It is clear what to expect  in general.

\begin{theorem}
Let $f$ be a cochain map from $(H, 0)$ to $(\sC, Q)$ which induces
the identity map on the cohomology $H$.
On $H[[\hbar]]:= H\otimes_{\Bbbk}\Bbbk[[\hbar]]$,
{ modulo its natural automorphism},
\begin{enumerate}

\item
there is a unique $\Bbbk[[\hbar]]$-linear map 
$\bos{\k}=\hbar \k^{(1)}+\hbar^{2}\k^{(2)}+\cdots$
into itself of ghost number zero,
which is induced from an infinite sequence $0,\k^{(1)},\k^{(2)},\cdots$ of
$\;\Bbbk$-linear maps on $H$ into $H$ parametrized by $\hbar$,
satisfying $\bos{\k}^{2}=0$ and $\bos{\k}\bigr|_{\hbar=0}=0$,

\item
there is a 
$\Bbbk[[\hbar]]$-linear map 
$\bos{f}=f+\hbar f^{(1)}+\hbar^{2}f^{(2)}+\cdots$
into $\sC[[\hbar]]$ itself of ghost number zero,
which is induced from an infinite sequence $f,f^{(1)},f^{(2)},\cdots$ of
$\;\Bbbk$-linear maps on $H$ into $\sC$  parametrized by $\hbar$,
satisfying $\bos{f}\bigr|_{\hbar=0}=f$, $\bos{K}\bos{f}=\bos{f} \bos{\k}$ 
and being defined
up to ``quantum homotopy'';
$$
\bos{f}\,\bos{\sim}\,\bos{f}^{\pr}=\bos{f} + \bos{K}\,\bos{s} +\bos{s}\,\bos{\k},
$$
where $\bos{s}=s^{(0)}+\hbar {s}^{(1)}+\hbar^{2}s^{(2)}+\cdots$
is an arbitrary sequence of
$\;\Bbbk$-linear maps of ghost number $-1$, parametrized by $\hbar$, on $H$ into $\sC$ . 
\end{enumerate}

\end{theorem}

\begin{remark}
The ``quantum homotopy'' relation
$\bos{f}\,\bos{\sim}\,\bos{f}^{\pr}=\bos{f} + \bos{K}\,\bos{s} +\bos{s}\,\bos{\k}$
modulo $\hbar$ is $f\,\sim\,f^{\pr}=f + Q \,s^{(0)}$ since $\bos{\k}=0\mod \hbar$.
Thus it reduce to homotopy equivalence of cochain maps from $(H,0)$ to $(\sC, Q)$.
\end{remark}

\begin{remark}
The natural automorphism on $H[[\hbar]]$ is an arbitrary sequence 
$\bos{\xi}=1+\hbar \xi^{(1)}+\hbar^{2}\xi^{(2)}+\cdots$ of $\;\Bbbk$-linear
maps, parametrized by $\hbar$, on $H$ into itself of ghost number $0$ satisfying
$\bos{\xi}\bigr|_{\hbar=0}=1$. Such an automorphism fixes $H$ and sends 
$\bos{\k}$ to $\bos{\xi}^{-1}\,\bos{\k}\,\bos{\xi}$
and $\bos{f}$ to $\bos{f}\,\bos{\xi}$, since  $\bos{\k}:H[[\hbar]]\longrightarrow H[[\hbar]]$
and $\bos{f}:H[[\hbar]]\longrightarrow \sC[[\hbar]]$. Note that every automorphism
fixes $f$ as well as $\k^{(1)}$, since $\bos{\k}\bigl|_{\hbar=0}=0$. An automorphism
$\bos{\xi}=1+\hbar\xi^{(1)}+\hbar\xi^{(2)}+\cdots $ sends, for examples, 
$$
\eqalign{
\k^{(1)}&\longrightarrow \k^{(1)}
,\cr
\k^{(2)}&\longrightarrow \k^{(2)} + \k^{(1)}\xi^{(1)}-\xi^{(1)}\k^{(1)}
,\cr
\k^{(3)}&\longrightarrow \k^{(3)} + \k^{(2)}\xi^{(1)}-\xi^{(1)}\k^{(2)}
+ \k^{(1)}\xi^{(2)}-\xi^{(2)}\k^{(1)}
-\xi^{(1)}\k^{(1)}\xi^{(1)} +\xi^{(1)}\xi^{(1)}\k^{(1)},
}
$$
since $\bos{\xi}^{-1}= 1 -\hbar \xi^{(1)}+\hbar^{2}\left(\xi^{(1)}\xi^{(1)}-\xi^{(2)}\right)
+\cdots$.
\end{remark}

Our proof of the above theorem is based an induction and relies on the
following mouthy lemma:

\begin{lemma}\label{speea}
Let $\bos{\xi}:= 1 +\hbar\xi^{(1)}+\hbar^{2}\x^{(2)}+\cdots 
$ 
be an arbitrary infinite sequence
of $\;\Bbbk$-linear maps, parametrized by $\hbar$, on $H$ into $H$ of ghost number 
$0$ such that $\bos{\xi}\bigl|_{\hbar=0}=1$.
Let $f$ be a cochain map from $(H, 0)$ to $(\sC, Q)$ which induces
the identity map on the cohomology $H$.

Fix $n>1$. Assume that 
there is a sequence 
$$
\bos{\k}:=\hbar \k^{(1)}+\hbar^{2}\k^{(2)}+\cdots+\hbar\k^{(n)}\mod \hbar^{n+1},
$$  
of $\;\Bbbk$-linear maps, parametrized by $\hbar$ such that
$\bos{\k}\bigl|_{\hbar=0}=0$, on $H$ into $H$ of ghost number $1$
and a sequence 
$$
\bos{f} := f +\hbar f^{(1)}+\cdots + \hbar^{n}f^{(n)}\mod \hbar^{n+1},
$$
of  $\;\Bbbk$-linear maps, parametrized by $\hbar$, on $H$ into $\sC$ 
of ghost number $0$ such that
\begin{enumerate}

\item 
$\bos{\k} \mod\hbar^{n+1}$ 
satisfies $\bos{\k}^{2}=0 \mod \hbar^{n+1}$ and is defined uniquely 
modulo an action of $\;\bos{\xi}$ such that
$$
\bos{\k}^{\pr}= \bos{\xi}^{-1}\bos{\k}\,\bos{\xi}\mod \hbar^{n+1}
$$

\item 
$\bos{f} \mod \hbar^{n+1}$ satisfies $\bos{K}\bos{f}=\bos{f}\bos{\k}\mod \hbar^{(n+1)}$
and is defined up to ``quantum homotopy'' modulo an action of $\;\bos{\xi}$
such that
$$
\bos{f}^{\pr}=\bos{f}\,\bos{\xi}
+ \bos{K}\,\bos{s}+\bos{s}\,\bos{\k}^{\pr} \mod \hbar^{n+1}
$$
where $\bos{s}=s^{(0)}+\hbar s^{(1)}+\cdots+\hbar^{n} s^{(n)}\mod \hbar^{n+1}$
is an arbitrary sequence of $\;\Bbbk$-linear maps,
parametrized by $\hbar$, on $H$ into $\sC$ of ghost number $-1$

\end{enumerate}
Then there is a sequence 
$$
\tilde\bos{\k}:=\hbar \k^{(1)}+\hbar^{2}\k^{(2)}+\cdots+\hbar^{n}\k^{(n)}
+\hbar^{n+1}\k^{(n+1)}
\mod \hbar^{n+2},
$$  
of $\;\Bbbk$-linear maps, parametrized by $\hbar$, on $H$ into $H$ of ghost number $1$
and a sequence 
$$
\tilde\bos{f} := f +\hbar f^{(1)}+\cdots + \hbar^{n}f^{(n)}
+ \hbar^{n+1}f^{(n+1)}\mod \hbar^{n+2},
$$
of  $\;\Bbbk$-linear maps, parametrized by $\hbar$, on $H$ into $\sC$ 
of ghost number $0$ such that
\begin{enumerate}

\item 
$\tilde\bos{\k}=\bos{\k}\mod\hbar^{n+1}$, and
$\tilde\bos{\k}$ 
satisfies $\tilde\bos{\k}^{2}=0 \mod \hbar^{n+2}$ and is defined uniquely 
modulo an action of $\;\bos{\xi}$ such that
$$
\tilde\bos{\k}^{\pr}= \bos{\xi}^{-1}\tilde\bos{\k}\,\bos{\xi}\mod \hbar^{n+2}
$$

\item 
$\tilde\bos{f}=\bos{f}\mod\hbar^{n+1}$,
and $\tilde\bos{f} \mod \hbar^{n+1}$ satisfies 
$\bos{K}\tilde\bos{f}=\tilde\bos{f}\tilde\bos{\k}\mod \hbar^{(n+2)}$
and is defined up to ``quantum homotopy'' modulo an action of $\bos{\xi}$
such that
$$
\tilde\bos{f}^{\pr}=\tilde\bos{f}\,\bos{\xi}
+ \bos{K}\,\tilde\bos{s}+\tilde\bos{s}\,\tilde\bos{\k}^{\pr} \mod \hbar^{n+2}
$$
where 
$\tilde\bos{s}=s^{(0)}+\hbar s^{(1)}+\cdots+\hbar^{n} s^{(n)}
+\hbar^{(n+1)}s^{(n+1)}\mod \hbar^{n+2}$
is an arbitrary sequence of $\;\Bbbk$-linear maps,
parametrized by $\hbar$, on $H$ into $\sC$ of ghost number $-1$.
\end{enumerate}
\end{lemma}
\begin{remark}
It is clear that $\bos{\k}^{\pr 2}=0\mod \hbar^{n+1}$ is implied
by $\bos{\k}^{2}=0\mod \hbar^{(n+1)}$. Also 
$\bos{K}\bos{f}=\bos{f}\bos{\k}\mod \hbar^{(n+1)}$ implies that
$\bos{K}\bos{f}^{\pr}=\bos{f}^{\pr}\bos{\k}^{\pr}\mod \hbar^{(n+1)}$;
by a direct computation we have
$$
\bos{K}\,\bos{f}^{\pr}
-\bos{f}^{\pr}\,\bos{\k}^{\pr}
=\bos{K}\bos{f}\bos{\xi} - \bos{f}\,\bos{\xi}\,\bos{\k}^{\pr} 
=\bos{K}\bos{f}\bos{\xi} - \bos{f}\,\bos{\k}\bos{\xi}
=0.
$$
\end{remark}
A proof the above lemma is equivalent to a proof of our theorem
since we have already shown that the assumption is true for $n=1$,
(we also did it for $n=2$ as a quick demonstration). 
Then we take the limit $n\longrightarrow \infty$.

So it remains to prove the lemma.   
Our proof relies on the following two propositions, which shall be proved later:
\begin{proposition}\label{claimone}
Let $g^{(n+1)}$ be the $\;\Bbbk$-linear map on $H$ into $\sC$ of ghost number $1$
defined by the formula
$$
g^{(n+1)}:=K^{(n+1)}f
+\sum_{\ell=1}^{n}K^{(n+1-\ell)}f^{(\ell)}
-\sum_{\ell=1}^{n}f^{(\ell)}\k^{(n+1-\ell)}.
$$
Then
$$
Q  g^{(n+1)}
=-\sum_{\ell=1}^{n} f\k^{(n+1-\ell)}\k^{(\ell)}.
$$
\end{proposition}

\begin{proposition}\label{claimtwo}
Let $g^{\pr(n+1)}$ be the $\;\Bbbk$-linear map on $H$ into $\sC$ of ghost number $1$
defined by the formula
$$
g^{\pr(n+1)}:=K^{(n+1)}f^{\pr}
+\sum_{\ell=1}^{n}K^{(n+1-\ell)}f^{\pr(\ell)}
-\sum_{\ell=1}^{n}f^{\pr(\ell)}\k^{\pr(n+1-\ell)}.
$$
Then
$$
\eqalign{
g^{\pr(n+1)}-g^{(n+1)}
=&
- Q
\left(
\sum_{\ell=1}^{n}f^{(n+1-\ell)}\xi^{(\ell)}
+\sum_{\ell=0}^{n}K^{(n+1-\ell)}s^{(\ell)} 
+\sum_{\ell=1}^{n}s^{(n+1-\ell)}\k^{\pr(\ell)}
\right)
\cr
&+ f\left(\bos{\xi}^{-1}\bos{\k}\bos{\xi}\right)^{(n+1)}
}
$$
\end{proposition}

\begin{proof}
From proposition \ref{claimone} we have
$$
Q  g^{(n+1)}(a)
=-\sum_{\ell=1}^{n} f\left(\k^{(n+1-\ell)}\left(\k^{(\ell)}(a)\right)\right),
$$
for any $a \in H$.
By taking the $Q$-cohomology class for the both hand sides of the above
we have
$$
0=-\sum_{\ell=1}^{n} \left[f\left(\k^{(n+1-\ell)}\left(\k^{(\ell)}(a)\right)\right)\right].
$$
It follows that 
$$
\sum_{\ell=1}^{n}\k^{(n+1-\ell)}\left(\k^{(\ell)}(a)\right)=0,
$$
since $f$ induces the identity map on $H$, $\left[f(b)\right]=b$ for any $b\in H$.
Since the above identity is true for each and every element in $H$, it implies
that
\eqn\clotb{
\sum_{\ell=1}^{n}\k^{(n+1-\ell)}\k^{(\ell)}=0.
}
It also follows that
$$
Q g^{(n+1)}=0.
$$
Thus the image of $g^{(n+1)}:H^{\bullet}\longrightarrow \sC^{\bullet+1}$ is contained
in $\Ker Q \cap \sC$. By taking the cohomology class of $g^{(n+1)}(a)$ for each
$a \in H$
we obtain a $\;\Bbbk$-linear map
$$
\k^{(n+1)}:H^{\bullet}\longrightarrow H^{\bullet +1},
$$
which  is defined by 
$$
\k^{(n+1)}(a):=\left[g^{(n+1)}(a)\right],
$$
for each and every $a \in H$.
 By composing
$\k^{(n+1)}$ with $f$, we have a $\;\Bbbk$-linear map
$f\k^{(n+1)}: H^{\bullet}\longrightarrow \sC^{\bullet+1}$, such that
$\left[f\left(\k^{(n+1)}(a)\right)\right]= \left[g^{(n+1)}(a)\right]$ for every $a\in H$. 
Thus there is some
$\;\Bbbk$-linear map $f^{(n+1)}:H^{\bullet}\longrightarrow \sC^{\bullet}$ 
of ghost number $0$
such that
\eqn\clotcx{
g^{(n+1)}=f\k^{(n+1)} - Q f^{(n+1)},
}
where $f^{(n+1)}$ is defined modulo $\Ker Q$.
Then, after
using the definition of $g^{(n+1)}$ in proposition \ref{claimone}, we have
\eqn\clotc{
K^{(n+1)}f
+\sum_{\ell=1}^{n}K^{(n+1-\ell)}f^{(\ell)}
+Q f^{(n+1)}
=\sum_{\ell=1}^{n}f^{(\ell)}\k^{(n+1-\ell)}
+f\k^{(n+1)}. 
}

Using $\k^{(n+1)}$ and $f^{(n+1)}$
we can extend both $\bos{\k}\mod \hbar^{n+1}$ and $\bos{f}\mod \hbar^{n+1}$
to the next level $\tilde\bos{\k}\mod \hbar^{n+2}$ and $\tilde\bos{f}\mod \hbar^{n+2}$
as follows
$$
\eqalign{
\tilde\bos{\k} &:=\sum_{\ell=1}^{n}\hbar^{n}\k^{(n)}
+ \hbar^{(n+1)}\k^{(n+1)}\mod \hbar^{n+2}
,\cr
\tilde\bos{f} &:=f+\sum_{\ell=1}^{n}\hbar^{n}f^{(n)}
+ \hbar^{(n+1)}f^{(n+1)}\mod \hbar^{n+2}.
}
$$
It is obvious that $\tilde\bos{\k}=\bos{\k}\mod \hbar^{n+1}$ and
$\tilde\bos{f}=\bos{f}\mod \hbar^{n+1}$.
Then the assumption that $\bos{\k}^{2}=0 \mod \hbar^{n+1}$ 
together with the identity \clotb\
implies that
$$
\tilde\bos{\k}^{2}=0 \mod \hbar^{n+2}.
$$
Also the assumption that $\bos{K}\,\bos{f}= \bos{f}\,\bos{\k}\mod \hbar^{(n+1)}$ 
together with
the relation \clotc\ implies that
$$
\bos{K}\,\tilde\bos{f}= \tilde\bos{f}\,\tilde\bos{\k}\mod \hbar^{(n+2)}.
$$

Now we deal with every ambiguity in the above extension.
For this, we dully repeat the similar procedure with $\bos{f}^{\pr}\mod \hbar^{n+1}$ and
$\bos{\k}^{\pr}\mod \hbar^{n+1}$ with a twist. 
Let's first recall the identity in proposition \ref{claimtwo};
$$
\eqalign{
g^{\pr(n+1)}-g^{(n+1)}
=&
- Q\left(
\sum_{\ell=1}^{n}f^{(n+1-\ell)}\xi^{(\ell)}
+\sum_{\ell=0}^{n}K^{(n+1-\ell)}s^{(\ell)} 
+\sum_{\ell=1}^{n}s^{(n+1-\ell)}\k^{\pr(\ell)}
\right)
\cr
&+ f\left(\bos{\xi}^{-1}\bos{\k}\bos{\xi}\right)^{(n+1)},
}
$$
which implies that $Q g^{\pr(n+1)}=0$ since 
$Q g^{(n+1)}=0$ as seen previously and $Q^{2}=Q f=0$. 
Thus the image of $g^{\pr(n+1)}:H^{\bullet}\longrightarrow \sC^{\bullet+1}$ is contained
in $\Ker Q \cap \sC$. By taking the cohomology class of $g^{\pr(n+1)}(a)$ for each
$a \in H$
we obtain a $\;\Bbbk$-linear map
$$
\k^{\pr(n+1)}:H^{\bullet}\longrightarrow H^{\bullet +1}
$$
defined by $\k^{\pr(n+1)}(a):=\left[g^{\pr(n+1)}(a)\right]$ for each and every $a\in H$. 
Then Claim (2) 
implies that 
$$
\left[g^{\pr(n+1)}(a)\right]=\left[g^{(n+1)}(a)\right]
+\left[f\left(\left(\bos{\xi}^{-1}\bos{\k}\bos{\xi}\right)^{(n+1)}(a)\right)\right],
$$
for each and every $a\in H$.
Thus we obtain that
\eqn\mhiu{
\k^{\pr(n+1)}=\k^{(n+1)} + \left(\bos{\xi}^{-1}\bos{\k}\bos{\xi}\right)^{(n+1)}.
}

By composing
$\k^{\pr(n+1)}$ with $f^{\pr}$, we have a $\;\Bbbk$-linear map
$f^{\pr}\k^{\pr(n+1)}: H^{\bullet}\longrightarrow \sC^{\bullet+1}$, such that
$\left[f^{\pr}\left(\k^{\pr(n+1)}(a)\right)\right]= \left[g^{\pr(n+1)}(a)\right]$ 
for every $a\in H$. Thus there is some
$\;\Bbbk$-linear map $f^{\pr(n+1)}:H^{\bullet}\longrightarrow \sC^{\bullet}$ 
of ghost number $0$
such that
\eqn\clotd{
g^{\pr(n+1)}=f^{\pr}\k^{\pr(n+1)} - Q f^{\pr(n+1)}.
}
where $f^{\pr(n+1)}$ is defined modulo $\Ker Q$.
Now we want to compare \clotd\ with \clotcx.
We begin with rewriting \clotd\ as follows;
$$
\eqalign{
g^{\pr(n+1)}
&=f^{\pr}\k^{\pr(n+1)}- Q f^{\pr (n+1)}
\cr
&=f \k^{(n+1)} 
+f\left(\bos{\xi}^{-1}\bos{\k}\bos{\xi}\right)^{(n+1)}
-Q\left( f^{\pr(2)}-s^{(0)}\k^{\pr (n+1)}\right),
}
$$
where we have used $f^{\pr}= f + Q\,s^{(0)}$ and
the relation \mhiu\ between $\k^{\pr(n+1)}$ and $\k^{(n+1)}$.
Then we write the above equation as follows
$$
f \k^{(n+1)} 
=g^{\pr(n+1)}
-f\left(\bos{\xi}^{-1}\bos{\k}\bos{\xi}\right)^{(n+1)}
-Q\left( f^{\pr(2)}-s^{(0)}\k^{\pr (n+1)}\right),
$$
while, from \clotcx, we also have
$$
 f\k^{(n+1)}= g^{(n+1)} -Q f^{(n+1)}.
$$
Thus  we obtain the following equality;
$$
g^{\pr(n+1)}
-f\left(\bos{\xi}^{-1}\bos{\k}\bos{\xi}\right)^{(n+1)}
-Q\left( f^{\pr(n+1)}-s^{(0)}\k^{\pr (n+1)}\right)
= g^{(n+1)} -Q f^{(n+1)}.
$$
We, then, use Claim (2) to conclude
that
$$
Q w^{(n+1)}=0,
$$
where
$$
\eqalign{
w^{(n+1)}:= 
&
f^{\pr(n+1)}
-f^{(n+1)}
- s^{(0)}\k^{\pr(n+1)}
\cr
&
-\sum_{\ell=0}^{n}
\left(f^{(n+1-\ell)}\xi^{(\ell)}
+K^{(n+1-\ell)}s^{(\ell)} 
+s^{(n+1-\ell)}\k^{\pr(\ell)}
\right).
}
$$
Then by taking cohomology  we have a $\;\Bbbk$-linear
map $\left[w^{(n+1)}\right]:H^{\bullet}\longrightarrow H^{\bullet}$, 
which is an arbitrary $\;\Bbbk$-linear map. 
So we introduce a new ``ghost variable'' 
$\xi^{(n+1)}:H^{\bullet} \longrightarrow H^{\bullet}$ for the arbitrariness.
Then we have 
$$
w^{(n+1)}=f\xi^{(n+1)}+ Q s^{(n+1)}
$$ 
for some $\;\Bbbk$-linear
map $s^{(n+1)}$ on $H$ into $\sC$ of ghost number $0$.
Finally we use  the definition of $w^{(n+1)}$ to 
conclude that
$$
\eqalign{
f^{\pr(n+1)}=
&f^{(n+1)}+ s^{(0)}\k^{\pr(n+1)}
+f\xi^{(n+1)}+ Q s^{(n+1)}
\cr
&
+\left(
\sum_{\ell=1}^{n}f^{(n+1-\ell)}\xi^{(\ell)}
+\sum_{\ell=0}^{n}K^{(n+1-\ell)}s^{(\ell)} 
+\sum_{\ell=1}^{n}s^{(n+1-\ell)}\k^{\pr(\ell)}
\right)
}
$$
In more tidier form, we have
\eqn\mhix{
\eqalign{
f^{\pr(n+1)}
=
&f^{(n+1)} + \sum_{\ell=1}^{n}f^{(n+1-\ell)}\xi^{(\ell)}+ f \xi^{(n+1)}
\cr
&
+\sum_{\ell=0}^{n}K^{(n+1-\ell)}s^{(\ell)} + Q s^{(n+1)}
+\sum_{\ell=1}^{n+1} s^{(n+1-\ell)}\k^{\pr(\ell)}.
}
}

Using $\k^{\pr(n+1)}$ and $f^{\pr(n+1)}$
we extend both $\bos{\k}^{\pr}\mod \hbar^{n+1}$ and $\bos{f}^{\pr}\mod \hbar^{n+1}$
to the next level $\tilde\bos{\k}^{\pr}\mod \hbar^{n+2}$ 
and $\tilde\bos{f}^{\pr}\mod \hbar^{n+2}$
as follows
$$
\eqalign{
\tilde\bos{\k}^{\pr} &:=\sum_{\ell=1}^{n}\hbar^{n}\k^{\pr(n)}
+ \hbar^{(n+1)}\k^{\pr(n+1)}\mod \hbar^{n+2}
,\cr
\tilde\bos{f}^{\pr} &:=f^{\pr}+\sum_{\ell=1}^{n}\hbar^{n}f^{\pr(n)}
+ \hbar^{(n+1)}f^{\pr(n+1)}\mod \hbar^{n+2}.
}
$$
Then the assumption that
$\bos{\k}^{\pr}=\bos{\xi}^{-1}\bos{\k}\,\bos{\xi}\,\mod \hbar^{n+1}$
together with the relation \mhiu\ implies that
$$
\tilde\bos{\k}^{\pr}=\bos{\xi}^{-1}\tilde\bos{\k}\,\bos{\xi}\,\mod \hbar^{n+2}.
$$
Also the assumption that
$\bos{f}^{\pr}= \bos{f}\,\bos{\xi} + \bos{K}\,\bos{s} + \bos{s}\,\bos{\k}^{\pr}
\,\mod \hbar^{n+1}$ together with the relation \mhix\ implies
that
$$
\tilde\bos{f}^{\pr}= \tilde\bos{f}\,\bos{\xi} + \bos{K}\,\tilde\bos{s} 
+ \tilde\bos{s}\,\tilde\bos{\k}^{\pr}\,
\mod \hbar^{n+2},
$$
where  $\tilde\bos{s}=s^{(0)}+\hbar s^{(1)}+\cdots+\hbar^{n}s^{(n)}
+\hbar^{(n+1)} s^{(n+1)}\mod\hbar^{n+2}$.

Thus our proof shall be complete once we prove the two claims we have made.
\qed.
\end{proof}

\begin{corollary}
Let $O$ be a classical observable. Then $O$ can be extended to a quantum observable
$\bos{O}$ if and only if $\;\bos{\k}\left(\left[O\right]\right)=0$,
i.e., $\k^{(\ell)}\left(\left[O\right]\right)=0$ for all $\ell=1,2,\cdots$.
Let $\;\bos{\k}\left(\left[O\right]\right)=0$. Then $\bos{f}\left(\left[O\right]\right)$
is the extension of the classical observable modulo ``quantum homotopy'',
i.e., modulo  $\bos{K}$-exact term.
\end{corollary}

%
%

\subsubsection{Proofs of Claims (1) and (2)}
\begin{claim}[1]
Let $g^{(n+1)}$ be the $\;\Bbbk$-linear map on $H$ into $\sC$ of ghost number $1$
defined by the formula
$$
g^{(n+1)}:=K^{(n+1)}f
+\sum_{\ell=1}^{n}K^{(n+1-\ell)}f^{(\ell)}
-\sum_{\ell=1}^{n}f^{(\ell)}\k^{(n+1-\ell)}.
$$
Then
$$
Q  g^{(n+1)}
=-\sum_{\ell=1}^{n} f\k^{(n+1-\ell)}\k^{(\ell)}.
$$
\end{claim}

\begin{proof}

It is convenient to denote $f=f^{(0)}$ so that 
$$
g^{(n+1)}=\sum_{\ell=0}^{n}K^{(n+1-\ell)}f^{(\ell)}
-\sum_{\ell=1}^{n}f^{(\ell)}\k^{(n+1-\ell)}.
$$
Applying $Q$ to above, we have
\eqn\couya{
\eqalign{
Q g^{(n+1)}
=&\sum_{\ell=0}^{n}Q K^{(n+1-\ell)}f^{(\ell)} 
-\sum_{\ell=1}^{n}Q f^{(\ell)}\k^{(n+1-\ell)}
\cr
=&\sum_{\ell=1}^{n+1}Q K^{(\ell)}f^{(n+1-\ell)}
-\sum_{\ell=1}^{n}
\sum_{j=1}^{\ell}\left(-K^{(j)}f^{(\ell-j)} +f^{(\ell-j)}\k^{(j)}
\right)\k^{(n+1-\ell)}
\cr 
=&-\sum_{\ell=1}^{n+1}\left(
 K^{(\ell)}Q
+\sum_{j=1}^{\ell-1}K^{(\ell-j)}K^{(j)}\right)f^{(n+1-\ell)}
+\sum_{\ell=1}^{n}\sum_{j=1}^{\ell}
K^{(j)}f^{(\ell-j)}k^{(n+1-\ell)}
\cr
&
-\sum_{\ell=1}^{n}\sum_{j=1}^{\ell}f^{(\ell-j)}\k^{(j)}\k^{(n+1-\ell)}
,
}
}
where 
(i) for the $2$-nd equality we used a re-summation and the assumption
$\bos{K}\bos{f}= \bos{f}\bos{\k}\mod\hbar^{n+1}$, which  implies 
that 
\eqn\couyb{
\left\{
\eqalign{
Q f^{(0)}&=0,\cr
Q f^{(\ell)}&=\sum_{j=1}^{\ell}\left(-K^{(j)}f^{(\ell-j)} +f^{(\ell-j)}\k^{(j)}\right),\qquad \ell=1,\cdots,n.
}
\right.
}
(ii) for the $3$rd equality we used the condition 
$\bos{K}^{2}=0$ modulo $\hbar^{n+2}$, which implies that
$$ 
Q K^{(\ell)}=- K^{(\ell)}Q
- \sum_{j=1}^{\ell-1}K^{(\ell-j)}K^{(j)},\qquad \ell=1,\cdots,n+1.
$$
Now consider the first two terms after the last equality in \couya. After a re-summation 
we have
$$
\eqalign{
&-\sum_{\ell=1}^{n+1}\sum_{j=1}^{\ell}
\left( K^{(\ell)}Q+ K^{(\ell-j)}K^{(j)}\right)f^{(n+1-\ell)}
+\sum_{\ell=1}^{n}\sum_{j=1}^{\ell}
K^{(j)}f^{(\ell-j)}k^{(n+1-\ell)}
\cr
&\qquad\qquad=
- K^{(n+1)}Q f^{(0)}
-\sum_{j=1}^{n}K^{(n+1-j)}\left(Q f^{(j)}+\sum_{\ell=1}^{j}K^{(\ell)}f^{(j-\ell)}
-\sum_{\ell=1}^{j}f^{(j-\ell)}\k^{(\ell)}
\right)
\cr
&\qquad\qquad
=0,
}
$$
where the last equality is due to \couyb.
Thus we obtain that
$$
\eqalign{
Q g^{(n+1)}
=-\sum_{\ell=1}^{n}\sum_{j=1}^{\ell}f^{(\ell-j)}\k^{(j)}\k^{(n+1-\ell)}.
}
$$
After a re-summation we have
$$
\eqalign{
Q g^{(n+1)}
=&-\sum_{j=1}^{n}f^{(n+1-j)}\left( \sum_{\ell=1}^{j-1}\k^{(j-\ell)}\k^{(\ell)}\right)
-f^{(0)}\sum_{\ell=1}^{n}\k^{(n+1-\ell)}\k^{(\ell)}.
}
$$
Finally we use the assumption that $\bos{\k}\bos{\k}=0 \mod\hbar^{n+1}$,
which implies that
$$
\sum_{\ell=1}^{j-1}\k^{(j-\ell)}\k^{(\ell)}=0, j=1,2,\cdots,n,
$$
to prove the claim that
$$
\eqalign{
Q  g^{(n+1)}
&=-f^{(0)}\sum_{\ell=1}^{n}\k^{(n+1-\ell)}\k^{(\ell)}
\cr
&\equiv
-\sum_{\ell=1}^{n}f\k^{(n+1-\ell)}\k^{(\ell)}.
}
$$
\qed
\end{proof}

\begin{claim}[2]
Let $g^{\pr(n+1)}$ be the $\;\Bbbk$-linear map on $H$ into $\sC$ of ghost number $1$
defined by the formula
$$
g^{\pr(n+1)}:=K^{(n+1)}f^{\pr}
+\sum_{\ell=1}^{n}K^{(n+1-\ell)}f^{\pr(\ell)}
-\sum_{\ell=1}^{n}f^{\pr(\ell)}\k^{\pr(n+1-\ell)}.
$$
Then
$$
\eqalign{
g^{\pr(n+1)}-g^{(n+1)}
=&
- Q
\left(
\sum_{\ell=1}^{n}f^{(n+1-\ell)}\xi^{(\ell)}
+\sum_{\ell=0}^{n}K^{(n+1-\ell)}s^{(\ell)} 
+\sum_{\ell=1}^{n}s^{(n+1-\ell)}\k^{\pr(\ell)}
\right)
\cr
&+ f\left(\bos{\xi}^{-1}\bos{\k}\bos{\xi}\right)^{(n+1)}
}
$$
\end{claim}

\begin{proof}
Recall that, from the assumptions of Lemma,
$$
\bos{\k}^{\pr}=\bos{\xi}^{-1}\bos{\k}\bos{\xi} \mod \hbar^{n+1}
$$
and
$$
\bos{f}^{\pr}=\bos{f}\bos{\xi} + \bos{K}\bos{s} +\bos{s}\bos{\k}^{\pr}\mod \hbar^{n+1}
$$
where $\bos{\xi}=1+\hbar\xi^{(1)}+\hbar^{2}\xi^{(2)}+\cdots$.
The expansion of $\bos{\xi}^{-1}$
shall be denoted by
$$
\bos{\xi}^{-1}=1+\hbar\bar\xi^{(1)}+\hbar^{2}\bar\xi^{(2)}+\cdots.
$$
Then, the identity $\bos{\xi}^{-1}\bos{\xi}=1$ implies that
\eqn\headx{
\eqalign{
\bar\xi^{(1)}&=-\xi^{(1)}
,\cr
\bar\xi^{(\ell)}&=-\sum_{j=1}^{\ell-1} \xi^{(j)}\bar \xi^{(\ell-j)} -\xi^{(\ell)}\hbox{ for } \ell>1.
}
}

Now we consider the expression 
$\left(\bos{\xi}^{-1}\bos{\k}\bos{\xi} \right)^{(n+1)}$,
which can be dully expanded as follows;
$$
\left(\bos{\xi}^{-1}\bos{\k}\bos{\xi} \right)^{(n+1)}
= \sum_{\ell=1}^{n}\k^{(n+1-\ell)}\xi^{(\ell)}
+\sum_{\ell=1}^{n}\bar\xi^{(\ell)}\k^{(n+1-\ell)}
+\sum_{\ell=1}^{n-1}\bar\xi^{(\ell)}\sum_{j=1}^{n-\ell}\k^{(n+1-\ell-j)}\xi^{(j)}.
$$
Also, for $1\leq \ell\leq n-2$, we have
$$
\eqalign{
\k^{\pr(n+1-\ell)}
=& \k^{(n+1-\ell)}
+\sum_{j=1}^{n-\ell}\k^{(n+1-\ell-j)}\xi^{(j)}
\cr
&
+\sum_{j=1}^{n-\ell}\bar\xi^{(j)}\k^{(n+1-\ell-j)}
+\sum_{j=1}^{n-\ell-1}\bar\xi^{(j)}\sum_{i=1}^{n-\ell-j}\k^{(n+1-\ell-j-i)}\xi^{(i)},
}
$$
while 
$$
\eqalign{
\k^{\pr(2)}
&=\k^{(2)}+\k^{(1)}\xi^{(1)}-\xi^{(1)}\k^{(1)}
,\cr
\k^{\pr(1)}
&=\k^{(1)}.
}
$$
Using \headx,
we have
$$
\eqalign{
\left(\bos{\xi}^{-1}\bos{\k}\bos{\xi} \right)^{(n+1)}
=& \sum_{\ell=1}^{n}\k^{(n+1-\ell)}\xi^{(\ell)}
-\sum_{\ell=1}^{n}\xi^{(\ell)}\k^{(n+1-\ell)}
-\sum_{\ell=1}^{n-1}\xi^{(\ell)}\sum_{j=1}^{n-\ell}\k^{(n+1-\ell-j)}\xi^{(j)}
\cr
&
-\sum_{\ell=2}^{n}\sum_{j=1}^{\ell-1}\xi^{(j)}\bar\xi^{(\ell-j)}\k^{(n+1-\ell)}
-\sum_{\ell=2}^{n-1}\sum_{i=1}^{\ell-1}\xi^{(i)}\bar\xi^{(\ell-i)}\sum_{j=1}^{n-\ell}\k^{(n+1-\ell-j)}\xi^{(j)}
}
$$
After re-summations of the last two terms in the right hand side of the above,
we have
$$
\eqalign{
\left(\bos{\xi}^{-1}\bos{\k}\bos{\xi} \right)^{(n+1)}
=& \sum_{\ell=1}^{n}\k^{(n+1-\ell)}\xi^{(\ell)}
-\sum_{\ell=1}^{n}\xi^{(\ell)}\k^{(n+1-\ell)}
-\sum_{\ell=1}^{n-1}\xi^{(\ell)}\sum_{j=1}^{n-\ell}\k^{(n+1-\ell-j)}\xi^{(j)}
\cr
&
-\sum_{\ell=1}^{n-1}\xi^{(\ell)}\sum_{j=1}^{n-\ell}\bar\xi^{(j)}\k^{(n+1-\ell-j)}
-\sum_{\ell=1}^{n-2}\xi^{(\ell)}\sum_{j=1}^{n-\ell-1}\bar\xi^{(j)}
\sum_{i=1}^{n-\ell-j}\k^{n-\ell +1-j-i}\xi^{(i)}.
}
$$
Comparing above with the definitions of $\k^{\pr(n+1-\ell }$ for $1\leq\ell\leq n$,
we obtain the following identity; 
\eqn\heady{
\left(\bos{\xi}^{-1}\bos{\k}\bos{\xi} \right)^{(n+1)}
= 
\sum_{\ell=1}^{n}\k^{(n+1-\ell)}\xi^{(\ell)}
-\sum_{\ell=1}^{n}\xi^{(\ell)}\k^{\pr(n+1-\ell)}.
}
After the similar manipulations we also obtain 
\eqn\headya{
\eqalign{
\k^{\pr(n+1-\ell)}
&=\k^{(n+1-\ell)} 
+\sum_{j=1}^{n-\ell}\k^{(n+1-\ell-j)}\xi^{(j)}
-\sum_{j=1}^{n-\ell}\xi^{(j)}\k^{\pr(n+1-\ell-j)}
\quad\hbox{for }1\leq\ell\leq n-1,
\cr
\k^{\pr(1)}&=\k^{(1)}.
}
}
We also note that
\eqn\headz{
\eqalign{
f^{\pr}&=f+ Q s^{(0)}
,\cr
f^{\pr(\ell)}&=f^{(\ell)} 
f\xi^{(\ell)}+\sum_{j=1}^{\ell}f^{(j)}\xi^{(\ell-j)}
+
Q s^{(\ell)} 
+ \sum_{j=1}^{\ell}K^{(j)}s^{(\ell-j)} 
+ \sum_{j=1}^{\ell}s^{(\ell-j)}\k^{\pr(j)},
}
}
for $1\leq \ell \leq n$. 

From the definitions of $g^{\pr(n+1)}$ and $g^{(n+1)}$, we
have
$$
\eqalign{
g^{\pr(n+1)}-g^{(n+1)}
=&
K^{(n+1)}\left(f^{\pr}-f\right)
+\sum_{\ell=1}^{n}K^{(n+1-\ell)}\left(f^{\pr(\ell)}-f^{(\ell)}\right)
\cr
&
-\sum_{\ell=1}^{n} f^{\pr(\ell)}\k^{\pr(n+1-\ell)}
+\sum_{\ell=1}^{n} f^{\pr(n)}\k^{(1)},
}
$$
which leads, after substituting $f^{\pr}$ and $f^{\pr(\ell)}$ using \headz, to 
the following complicated formula;
\eqn\headache{
\eqalign{
g^{\pr(n+1)}-g^{(n+1)}
=&
K^{(n+1)}Q s^{(0)}
+\sum_{\ell=1}^{n}K^{(n+1-\ell)}Q s^{(\ell)}
+\sum_{\ell=1}^{n}\sum_{j=1}^{\ell}K^{(n+1-\ell)}K^{(j)}s^{(\ell-j)}
\cr
&
+ \sum_{\ell=1}^{n}K^{(n+1-\ell)}f\xi^{(\ell)}
+ \sum_{\ell=1}^{n}\sum_{j=1}^{\ell}K^{(n+1-\ell)}f^{(j)}\xi^{(\ell-j)}
\cr
&
-\sum_{\ell=1}^{n} f^{(\ell)} \k^{\pr(n+1-\ell)}
+\sum_{\ell=1}^{n}f^{(\ell)}\k^{(n+1-\ell)}
-\sum_{\ell=1}^{n}\sum_{j=1}^{\ell}f^{(j)}\xi^{(\ell-j)}\k^{\pr(n+1-\ell)}
\cr
&
+\sum_{\ell=1}^{n}\sum_{j=1}^{\ell}K^{(n+1-\ell)}s^{(\ell-j)}\k^{\pr(j)}
- \sum_{\ell=1}^{n}\sum_{j=1}^{\ell}K^{(j)}s^{(\ell-j)} \k^{\pr(n+1-\ell)}
\cr
&
- \sum_{\ell=1}^{n}\sum_{j=1}^{\ell}s^{(\ell-j)}\k^{\pr(j)}\k^{\pr(n+1-\ell)}
\cr
&
-Q \sum_{\ell=1}^{n}s^{(\ell)} \k^{\pr(n+1-\ell)}
-f\sum_{\ell=1}^{n}\xi^{(\ell)}\k^{\pr(n+1-\ell)}.
}
}
%
%
%
We examine the right hand side of the above equality line by line:

\begin{enumerate}
\item The $1$st line: After a re-summation of the last term we have
$$
K^{(n+1)}Q s^{(0)}
+\sum_{\ell=0}^{n-1}
\left( K^{(n+1-\ell)}Q+\sum_{j=1}^{n-\ell}K^{(n+1-\ell-j)}K^{(j)}\right)s^{(\ell)}.
$$
Using the identity $\bos{K}^{2}=0$, which implies that
$$
\eqalign{
K^{(n+1-\ell)}Q+\sum_{j=1}^{n-\ell}K^{(n+1-\ell-j)}K^{(j)}&=-
Q K^{n+1-\ell} \quad\hbox{for }1\leq\ell\leq n-1
,\cr 
K^{(n+1)}Q&= -Q K^{(n+1)},
}
$$
we have
$$
L1=- Q\sum_{\ell=0}^{n} K^{n+1-\ell} s^{(\ell)}.
$$

\item The $2$-nd line:
 After a re-summation of the last term we have
$$
K^{(1)}f\xi^{(n)}
+\sum_{\ell=1}^{n-1}
\left( K^{(n+1-\ell)}f+\sum_{j=1}^{n-\ell}K^{(n+1-\ell-j)}f^{(j)}\right)\xi^{(\ell)}
$$
Using the assumption  $\bos{K}\bos{f}=\bos{f}\bos{\k} \mod \hbar^{n+1}$,
which implies that
$$
K^{(1)}f=- Q f^{(1)}+ f\k^{(1)},
$$
and, for $1\leq \ell\leq n-1$,
$$
K^{(n+1-\ell)}f+\sum_{j=1}^{n-\ell}K^{(n+1-\ell-j)}f^{(j)}
+Q f^{(n+1-\ell)}=f\k^{(n+1-\ell)}+\sum_{j=1}^{n-\ell}f^{(j)}\k^{(n+1-\ell-j)},
$$
we have
$$
\eqalign{
L2=&-Q \sum_{\ell=1}^{n}f^{(n+1-\ell)}\xi^{(\ell)}
+f\sum_{\ell=1}^{n}\k^{(n+1-\ell)}\xi^{\ell}
+\sum_{\ell=1}^{n-1}\sum_{j=1}^{n-\ell}f^{(j)}\k^{(n+1-\ell-j)}\xi^{(\ell)}.
}
$$
\item 
 The $3$rd line:
After a re-summation of the last term we have
$$
-\sum_{\ell=1}^{n} f^{(\ell)} 
\left(\k^{\pr(n+1-\ell)}
-\k^{(n+1-\ell)}\right)
-\sum_{\ell=1}^{n-1}\sum_{j=1}^{n-\ell}f^{(\ell)}\xi^{(j)}\k^{\pr(n+1-\ell-j)},
$$
which can be re-grouped as follows
$$
-f^{(n)}\left(\k^{\pr(1)}-\k^{(1)}\right)
-\sum_{\ell=1}^{n-1} f^{(\ell)} 
\left(
\k^{\pr(n+1-\ell)}
-\k^{(n+1-\ell)}
+\sum_{j=1}^{n-\ell}\xi^{(j)}\k^{\pr(n+1-\ell-j)}
\right).
$$
Now we use  the identities in \headya\ to have
$$
L3= -\sum_{\ell=1}^{n-1}\sum_{j=1}^{n-\ell} f^{(\ell)} \k^{(n+1-\ell-j)}\xi^{(j)}.
$$
Note that $L3$ cancels the last term  of $L2$;
$$
L2+L3=-Q \sum_{\ell=1}^{n}f^{(n+1-\ell)}\xi^{(\ell)}
+f\sum_{\ell=1}^{n}\k^{(n+1-\ell)}\xi^{\ell}.
$$

\item The $4$-th line:
The two terms  cancel each others;
$$
L4=0.
$$

\item 
 The $5$-th line:
After a re-summation we have
$$
-\sum_{\ell=0}^{n-1}s^{(\ell)}\left(\sum_{j=1}^{n-\ell}\k^{n+1-\ell-j}\k^{(j)}\right).
$$
Using the assumption that $\bos{\k}^{2}=0\mod \hbar^{(n+1)}$, which implies
that
$$
\sum_{j=1}^{n-\ell}\k^{n+1-\ell-j}\k^{(j)}=0\quad \hbox{for}\quad 0\leq\ell\leq n-1,
$$
we have
$$
L5=0.
$$
\item
The $6$-th (the last) line:
We do nothing;
$$
L6=
-Q \sum_{\ell=1}^{n}s^{(\ell)} \k^{\pr(n+1-\ell)}
-f\sum_{\ell=1}^{n}\xi^{(\ell)}\k^{\pr(n+1-\ell)}.
$$
\end{enumerate}
Adding everything together, $g^{\pr(n+1)}-g^{(n+1)}=L1+L2+L3+L4+L5+L6$,
we have
$$
\eqalign{
g^{\pr(n+1)}-g^{(n+1)}
=&
- Q
\left(
\sum_{\ell=1}^{n}f^{(n+1-\ell)}\xi^{(\ell)}
+\sum_{\ell=0}^{n}K^{(n+1-\ell)}s^{(\ell)} 
+\sum_{\ell=1}^{n}s^{(n+1-\ell)}\k^{\pr(\ell)}
\right)
\cr
&+ f\left(
\sum_{\ell=1}^{n}\k^{(n+1-\ell)}\xi^{(\ell)}
-\sum_{\ell=1}^{n}\xi^{(\ell)}\k^{\pr(n+1-\ell)}
\right).
}
$$
Finally, after using the identity in \heady, we are done.
\qed
\end{proof}

\subsection{BV QFT}
\def\expb#1{\left<#1\right>}
A BV QFT for us is a BV QFT algebra $\left(\sC[[\hbar]],\bos{K},\,\cdot\,\right)$ with
an additional algebraic notion corresponding to Batalin-Vilkovisky-Feynman path integral. 

Before we jump into making a definition, let's examine
an arbitrary $\;\Bbbk[[\hbar]]$-linear map $\left<\hbox{ }\right>$ of certain ghost number $N$
on $\sC[[\hbar]]$ into $\;\Bbbk[[\hbar]]$, which is a sequence
${\left<\phantom{O}\right>}=
{\left<\phantom{O}\right>}^{(0)}+ \hbar {\left<\phantom{O}\right>}^{(1)}
+\hbar^{2}{\left<\phantom{O}\right>}^{(2)}+\cdots$
of $\;\Bbbk$-linear maps on $\sC$ into $\Bbbk$ and satisfies
 $\left<\bos{K}\bos{\l}\right>=0$
for any $\bos{\l}=\l^{(0)}+\hbar \l^{(1)}+\hbar^{2}\l^{(2)}
+\cdots\in \sC[[\hbar]]$.
Then
the following formal sum vanishes;
$$
\sum_{n=0}^{\infty}\hbar^{n}\sum_{k=0}^{n}
\expb{\left(\bos{K}\bos{\l}\right)^{(n-k)}}^{(k)}
=0,
$$
where $\left(\bos{K}\bos{\l}\right)^{(j)}= Q \l^{(j)}+\sum_{i=1}^{j}K^{(i)}\l^{(j-i)}$.
Thus  the condition $\expb{\bos{K}\bos{\l}}=0$ for any $\bos{\l}\in\sC[[\hbar]]$
is equivalent to the following infinite
sequence of conditions; 
\eqn\qftia{
\sum_{k=0}^{n}
\expb{\left(\bos{K}\bos{\l}\right)^{(n-k)}}^{(k)}=0
\quad\hbox{ for all } n=0,1,2,\cdots, \&\hbox{ for any $\bos{\l}\in\sC[[\hbar]]$}.
}
The first few leading relations, for a demonstration, are
$$
\eqalign{
\expb{Q \l^{(0)}}^{(0)}=0
,\cr
\expb{Q \l^{(1)}+K^{(1)}\l^{(0)}}^{(0)}
+\expb{Q \l^{(0)}}^{(1)}
=0
,\cr
\expb{Q \l^{(2)}+K^{(1)}\l^{(1)}+K^{(2)}\l^{(0)}}^{(0)}
+\expb{Q \l^{(1)}+K^{(1)}\l^{(0)}}^{(1)}
+\expb{Q \l^{(0)}}^{(2)}
=0.
}
$$
The first condition in the above implies  that $\expb{Q x}^{(0)}=0$ for any $x\in \sC$.
It follows that $\expb{Q \l^{(1)}}^{(0)}=\expb{Q \l^{(2)}}^{(0)}=0$ 
since $\l^{(1)},\l^{(2)}\in \sC$.
This property can be used to simplify the remaining  relations as follows;
$$
\eqalign{
\expb{K^{(1)}\l^{(0)}}^{(0)}
+\expb{Q \l^{(0)}}^{(1)}
=0
,\cr
\expb{K^{(1)}\l^{(1)}+K^{(2)}\l^{(0)}}^{(0)}
+\expb{Q \l^{(1)}+K^{(1)}\l^{(0)}}^{(1)}
+\expb{Q \l^{(0)}}^{(2)}
=0.\cr
}
$$
etc.
Now  the  first condition in the above implies that
$\expb{K^{(1)}x}^{(0)}
+\expb{Q x}^{(1)}
=0$ for any $x\in \sC$.
Thus we have a further simplification
$$
\eqalign{
\expb{K^{(2)}\l^{(0)}}^{(0)}
+\expb{K^{(1)}\l^{(0)}}^{(1)}
+\expb{Q \l^{(0)}}^{(2)}
=0,\cr
}
$$
implying that $\expb{K^{(2)}x}^{(0)}
+\expb{K^{(1)}x}^{(1)}
+\expb{Q x}^{(2)}
=0$ for any $x\in \sC$.
This demonstration suggests that the condition that $\expb{\bos{K}\bos{\l}}=0$
for all $\bos{\La}\in \sC[[\hbar]]$ is equivalent to the following
infinite sequence of conditions;
\eqn\qftib{
\eqalign{
\expb{Q x}^{(0)}&=0 \hbox{ for any $x\in\sC$}
,\cr
\expb{Q x}^{(n)} +\sum_{\ell=1}^{n}\expb{K^{(\ell)}x}^{(n-\ell)}&=0
\hbox{ for any $x\in\sC$} \& \hbox{ for all } n=1,2,\cdots.
}
}
Proof is omitted.

Now let's adopt more usual notation such that 
$c^{(\ell)}(y):= \expb{y}^{(\ell)}$, $y\in \sC$, the  condition \qftib\ 
can be written in more suggestive way as follows;
$$
c^{(n)} Q +\sum_{\ell=1}^{n}c^{(n-\ell)} K^{(\ell)}=0
\quad\hbox{ for all } n=0,1,2,\cdots,
$$
which are conditions  for the sequence $c^{(0)},c^{(1)},c^{(2)},\cdots$ 
of $\Bbbk$-linear maps
on $\sC$ into $\Bbbk$. It is obvious that the above conditions can be written
as $\bos{c}\circ \bos{K}=0$ where $\bos{c}= c^{(0)}+\hbar c^{(1)}+\hbar^{2}c^{(2)}+\cdots$.
Then we have a natural notion of homotopy, in the sense that $\bos{c}^{\pr}:= \bos{c} 
+\bos{r} \bos{K}$ for any $\bos{r}=r^{(0)}+\hbar r^{(1)}+\hbar r^{(2)}+\cdots$,
where $r^{(0)},r^{(1)},r^{(2)},\cdots$ is a sequence of $\;\Bbbk$-linear maps of ghost number
$N-1$ on $\sC$ into $\Bbbk$, we automatically have $\bos{c}^{\pr} \bos{K}=0$ 
since $\bos{K}^{2}=0$. Explicitly 
$$
\eqalign{
c^{\pr(0)}-c^{(0)}&= r^{(0)}Q,\cr
c^{\pr(n)}-c^{(n)}&= r^{(n)}Q +\sum_{\ell=1}^{n}r^{(n-\ell)} K^{(\ell)}\hbox{ for all $n=1,2,\cdots$}.
}
$$
We, then, say $\bos{c}$ and $\bos{c}^{\pr}$ are ''quantum homotopic'' and denote
$\bos{c}\,\bos{\sim}\,\bos{c}^{\pr}$. We shall denote the ``quantum homotopy type (class)''
of $\bos{c}$ by $\{\bos{c}\}$. 
\begin{remark}
Variations of $\bos{c}$ within the same quantum homotopy type is a realization of continuous 
deformations or homologous deformations of Lagrangian subspace $\mL$ (gauge choice)
in the BV quantization scheme.
\end{remark}

\begin{definition}
An unital  BV QFT is a 
BV QFT algebra $(\sC[[\hbar]], \bos{K}, \hbox{ }\cdot\hbox{ })$ 
with a sequence $c^{(0)},c^{(1)}, c^{(2)},\cdots$ of $\;\Bbbk$-linear maps
of ghost number zero
on $\sC$ into $\;\Bbbk$ such that $\bos{c}:= c^{(0)}+\hbar c^{(1)}+\hbar^{2}c^{(2)}+
\cdots$ satisfies $\bos{c}\,\bos{K}=0$, $\bos{c}(1)=1$ and defined up to ``quantum homotopy'';
$$
\bos{c}\,\bos{\sim}\,\bos{c}^{\pr}=\bos{c} + \bos{r}\circ\bos{K},
$$
for some  $\bos{r}= r^{(0)}+\hbar r^{(1)}+\hbar^{2}r^{(2)}+ \cdots$,
where $r^{(0)},r^{(1)}, r^{(2)},\cdots$ is a sequence of $\;\Bbbk$-linear maps
$r^{(\ell)}:\sC\longrightarrow \Bbbk$ with ghost number $-1$.
\end{definition}
Note that the ghost number of $\Bbbk$ (and $\Bbbk[[\hbar]]$) is concentrated to
zero. So the sequence $c^{(0)},c^{(1)}, c^{(2)},\cdots$ of $\Bbbk$-linear maps
should be zero maps on $\sC^{n}$ for $n\neq 0$.

A BV QFT does not need to be restricted to be unital.  Let's
assume the $\bos{c}(1)\neq 1$. Consider the case that $\bos{c}(1)\neq 0$. Then we can
simply divide everything by $\bos{c}(1)$ provided that $c^{(0)}(1)\neq 0$
to get an unital theory. The case $\bos{c}(1)=0$
is simply uninteresting. Alternatively we can consider the case that the ghost number
of $\bos{c}$ is non-zero. 
\begin{definition}
A BV QFT  with ghost number anomaly $N\in \bos{Z}$, $N\neq 0$, is a 
BV QFT algebra $(\sC[[\hbar]], \bos{K}, \hbox{ }\cdot\hbox{ })$ 
with a sequence $c^{(0)},c^{(1)}, c^{(2)},\cdots$ of $\Bbbk$-linear maps
of ghost number $-N$
on $\sC$ into $\Bbbk$ such that $\bos{c}:= c^{(0)}+\hbar c^{(1)}+\hbar^{2}c^{(2)}+
\cdots$ satisfies $\bos{c}\,\bos{K}=0$ and is defined up to ``quantum homotopy'';
$$
\bos{c}\,\bos{\sim}\,\bos{c}^{\pr}=\bos{c} + \bos{r}\bos{K},
$$
for some  $\bos{r}= r^{(0)}+\hbar r^{(1)}+\hbar^{2}r^{(2)}+ \cdots$,
where $r^{(0)},r^{(1)}, r^{(2)},\cdots$ is a sequence of $\,\Bbbk$-linear maps
$r^{(\ell)}:\sC\longrightarrow \Bbbk$ with ghost number $-N-1$.
\end{definition}
Note again that the ghost number of $\Bbbk$ (and $\Bbbk[[\hbar]]$) is concentrated to
zero. So the sequence $c^{(0)},c^{(1)}, c^{(2)},\cdots$ of $\Bbbk$-linear maps
should be zero maps on $\sC^{n}$ for $n\neq N$.

\begin{remark}
The terminology  'ghost number anomaly' has the following origin. Consider
a path integral $``\int_{\mL}d\!\m '' e^{-\bos{S}/\hbar}$
in the BV quantization scheme. The BV quantum master action $\bos{S}$ has ghost number
zero, while the path integral measure $d\!\m$ may have non-zero ghost number. The later
is usually due to the possible  zero-modes of anti-commuting classical fields, which modes 
 do not contribute to $\bos{S}\bigl|_{\mL}$ but contribute to 
 the path integral measure $d\!\m$.
The net violation of ghost number in $d\!\m$ due to those zero-modes is called 
the ghost number anomaly, which is closely related with the index theory.
Assume that the theory has ghost number anomaly $N$. Then, by the properties of
Berezin integral of anti-commuting field, $\int d\!\th 1=0$ and $\int d\!\th \th=1$,
the path integrals always vanish unless one insert
suitable observable with ghost number $N$. Ghost number anomaly is an important 
feature of Witten's topological field theory, see \cite{W}. 
The ghost number anomaly should not depend on continuous or homologous deformations
of $\mL$, but may depend on ``homology'' type of $\mL$.
\end{remark}

We may also accommodate the various possible cases with different ghost number anomalies
into a single definition by replacing $\Bbbk$ with a $\bos{Z}$-graded free 
$\Bbbk$-module $V=\sum_{j\in \bos{Z}}V^{j}$, where $V^{j}\simeq \Bbbk$ but with
ghost number $j$.

\begin{definition}
A BV QHT is a 
BV QFT algebra $(\sC[[\hbar]], \bos{K}, \hbox{ }\cdot\hbox{ })$ 
with a sequence $\bos{c}:= c^{(0)}+\hbar c^{(1)}+\hbar^{2}c^{(2)}+
\cdots$ of $\Bbbk$-linear maps, parametrized by $\hbar$,
of ghost number $0$
on $\sC$ into a $\bos{Z}$-graded free $\Bbbk$-module $V=\sum_{j\in \bos{Z}}V^{j}$
such that $\bos{c}\,\bos{K}=0$ and $\bos{c}$ is defined up to ``quantum homotopy'';
$$
\bos{c}\,\bos{\sim}\,\bos{c}^{\pr}=\bos{c} + \bos{r}\bos{K},
$$
for some  $\bos{r}= r^{(0)}+\hbar r^{(1)}+\hbar^{2}r^{(2)}+ \cdots$,
where $r^{(0)},r^{(1)}, r^{(2)},\cdots$ is a sequence of $\,\Bbbk$-linear maps
$r^{(\ell)}:\sC\longrightarrow V$ with ghost number $-1$.
\end{definition}
In the above definition it is understood that there is a
sequence $c_{j}^{(0)},c_{j}^{(1)}, c_{j}^{(2)},\cdots$ of $\;\Bbbk$-linear maps on
$\sC^{j}$ into $V^{j}\simeq \Bbbk$ for each $j$.

Now we are ready to define expectation values of  observables. 
We recall that our first theorem
give a canonical sequence $f^{(0)}, f^{(1)}, f^{(2)},\ldots $ of $\Bbbk$-linear maps of ghost
number $0$
on $H$, the space of equivalence classes of classical observables, to $\sC$ such
that $\bos{f}:=f^{(0)} +\hbar f^{(1)}+\hbar^{2} f^{(2)}+\cdots$ satisfies
$\bos{K}\,\bos{f}= \bos{f}\,\bos{\k}$ and is defined up to ``quantum homotopy'';
$$
\bos{f}\,\bos{\sim}\,\bos{f}^{\pr}= \bos{f}+\bos{K}\,\bos{s} + \bos{s}\,\bos{\k},
$$
for any sequence   $\bos{s}= s^{(0)}+\hbar s^{(1)}+\hbar^{2}s^{(2)}+ \cdots$,
$\Bbbk$-linear maps of ghost number $-1$, parametrized by $\hbar$, on $H$ to $\sC$.
We can compose the map $\bos{f}$, regarded as a $\Bbbk[[\hbar]]$-linear map
on $H[[\hbar]]=H\otimes_{\Bbbk}\Bbbk[[\hbar]]$ into $\sC[[\hbar]]$, with
the map $\bos{c}:= c^{(0)}+\hbar c^{(1)}+\hbar^{2}c^{(2)}+
\cdots$, regarded as a $\Bbbk[[\hbar]]$-linear map
on $\sC[[\hbar]]$ into $\Bbbk[[\hbar]]$ (or into $V[[\hbar]]$, to obtain a
$\Bbbk[[\hbar]]$-linear map 
$\bos{\iota}:=\bos{c}\bos{f}=\iota^{(0)}+\hbar \iota^{(1)}+\hbar^{2}\iota^{(2)}+\cdots$ 
on $H[[\hbar]]$ into $\Bbbk[[\hbar]]$ (or into $V[[\hbar]]$, 
such that 
$$
\iota^{(n)}=\sum_{\ell=0}^{n}c^{(n-\ell)} f^{(\ell)}, \quad n=0,1,2,\cdots.
$$

Note that the ambiguity of $\bos{\iota}$ due to the ambiguities of 
$\bos{f}\,\bos{\sim}\,\bos{f}^{\pr}=\bos{f} +\bos{K}\,\bos{s} + \bos{s}\,\bos{\k}$ 
and $\bos{c}\,\bos{\sim}\,\bos{c}^{\pr}=\bos{c} +\bos{r}\,\bos{K}$
is
$$
\eqalign{
\bos{\iota}^{\pr}-\bos{\iota}
&\equiv \bos{c}^{\pr}\bos{f}^{\pr} -\bos{c}\,\bos{f}
\cr
&=\bos{c}\left(\bos{K}\,\bos{s} + \bos{s}\,\bos{\k}\right)
+\bos{r}\,\bos{K}\left(\bos{f} +\bos{K}\,\bos{s} + \bos{s}\,\bos{\k}\right)
\cr
&=\bos{c}\,\bos{s}\,\bos{\k} + \bos{r}\,\bos{K}\,\bos{f}+ \bos{r}\,\bos{K}\,\bos{s}\,\bos{\k}
\cr
&=
 \left(\bos{c}\,\bos{s}+\bos{r}\,\bos{f} +\bos{r}\,\bos{K}\,\bos{s}\right)\bos{\k},
}
$$
where we used $\bos{c}\,\bos{K}=0$ and $\bos{K}^{2}=0$ for the second equality
and $\bos{K}\,\bos{f}=\bos{f}\,\bos{\k}$ for the third equality.
An automorphism $\bos{g}$ on $\sC[[\hbar]]$ sends $\bos{f}$ to
$\bos{g}\,\bos{f}$ and $\bos{c}$ to $\bos{c}\,\bos{g}^{-1}$,
since $\bos{f}$ and $\bos{c}$ are $\Bbbk[[\hbar]]$-linear maps
to $\sC[[\hbar]]$ and from $\sC[[\hbar]]$, respectively.
It follows that the composition $\bos{\iota}=\bos{c}\bos{f}$ is invariant
under the automorphism of BV QFT algebra.
We also recall that a classical observable $O$ is extendable to 
a quantum observable $\bos{O}$
if and only if $\bos{\k}\left(\left[O\right]\right)=0$. 
If follows that
$\bos{\iota}\left(\left[O\right]\right)=\bos{\iota}^{\pr}\left(\left[O\right]\right)$.

By the way, it is the cohomology class of classical observable 
that is observable to a classical observer.  Also there is no genuine classical observable so
that every classical observation must be classical approximation of quantum observation.
So we shall omit the decorations ``classical'' and ``quantum'' and define observables
and their expectation values;

%
%
\begin{definition}[Theorem]
An observable $o$ is an element of the cohomology $H$ of the complex $(\sC,Q)$
satisfying $\k^{(n)}\left(o\right)=0$ for all $n=1,2,3,\cdots$, i.e., $\bos{\k}(o)=0$. 
The quantum expectation value of
an observable $o$ is 
$$
\bos{\iota}
\left(o\right)
=\bos{c}\left(\bos{f}(o)\right)
=\sum_{n=0}^{\infty}\hbar^{n}\sum_{\ell=0}^{n}c^{(n-\ell)}\left(f^{(\ell)}\left(o\right)\right),
$$
which  is a ``quantum'' homotopy invariant as well as invariant under the automorphism of BV
QFT algebra.
\end{definition}
Recall the $\bos{f}$ maps the identity $1\in H$ to the identity $1\in \sC\subset\sC[[\hbar]]$.
Thus, for an unital BV QFT, we have $\bos{\iota}(1)=1$ since $\bos{c}(1)=1$.
\begin{definition}
Let $o$ be an observable. Then we call $\bos{f}(o)\in \sC[[\hbar]]$ a
quantum representative of the observable $o$, or the quantum representative of
observable $o$ with respect to the quantum extension map $\bos{f}$.
Similarly we call $f(o)\in \sC$ a classical representative of the observable $o$.
\end{definition}

\begin{remark}
We say an element in $H$ which is not annihilated by $\bos{\k}$ an invisible.
An important question is that why invisibles exist and what is the
meaning of their existence? We will not discuss this issue here, but the answer shall
be that the invisibles are responsible to  the fundamental  quantum symmetry and any non-Abelian
classical gauge symmetry is its avatar. 
\end{remark}

%
%

\begin{remark}
From now on we shall use the time-honored symbol $\expb{\phantom{O}}$ instead of $\bos{c}$
for a BV QFT, where it is understood that a ``quantum'' homotopy type of $\bos{c}$ is fixed,
such that $\iota(a) =\left<\bos{f}(a)\right>$ for $a\in H$.
\end{remark}


\newsec{Quantum master equation, quantum coordinates on moduli space and
an exact solution of BV QFT}

This is the beginning of the second part of this paper on an exact solution of generating
functional of quantum correlations functions of a BV QFT.
We shall assume that $\bos{\k}=0$
identically on $H$ such that every element of $H$ is observable. We shall also
assume that $H$ is finite dimensional for each ghost number.

From the assumption that $\bos{\k}=0$ and theorem \ref{Propos1}, we have a sequence 
$\bos{f}=f +\hbar f^{(1)}+\cdots$ of 
$\Bbbk$-linear maps on $H$ into $\sC$
of ghost number zero
such that $\bos{K} \bos{f}=0$, which classical limit  
$f=\bos{f}\bigr|_{\hbar=0}$ is 
a quasi-isomorphism of complexes $f:(H, 0)\longrightarrow (\sC, Q)$,
which induces the identity map on $H$. 
From the condition $\bos{K}1=0$, thus $Q1=0$, in the definition of BV QFT algebra,
there is a distinguished element $e \in H^0$ corresponding to the cohomology class
$[1]$ of the unit $1$ in $(\sC,\;\cdot\;)$.  On $H$ there is also 
an unique binary product $m_2:H\otimes H \longrightarrow H$ of ghost number $0$
induced from the product in the CDGA $(\sC, Q,\;\cdot\;)$;
let $a,b \in H$ then  $m_2(a,b):=\left[f(a)\cdot f(b)\right]$ 
which is an homotopy invariant since $Q$ is a derivation of the product $\cdot$,
and $m_2(e, b)=m_2(b,e)=b$, such that $(H,0, m_2)$ is a CDGA with unit $e$ with zero differential.
The product $m_2$ is super-commutative $m_2(a,b)= (-1)^{|a||b|}m_2(b,a)$ since
the product $\cdot$ is super-commutative.
It is natural to fix $f$ and $\bos{f}$ such that $f(e)=1$ and $\bos{f}(e)=1$.

It is convenient to choose a homogeneous basis $\{e_\a\}$ of $H$ such that one of its component,
say $e_0$, is the distinguished element $e$. Let $t_H=\{t^\a\}$ be the dual basis (basis of $H^*$) 
such that $|t^\a| + |e_\a|=0$, which is 
a coordinates system on $H$ with a distinguished coordinate $t^0$.
We denote $n$-th  symmetric product of the graded vector space $H^*$
by $S^n(H^*)$;
$$
S^n(H^*)=(H^*)^{\otimes n}\Big/a\otimes b -(-1)^{|a||b|}b\otimes a,
$$
and consider the following natural increasing filtration
\eqn\filt{
S^{(0)}(H^*)\subset S^{(1)} (H^*)\subset \cdots\subset S^{(k-1)} (H^*) \subset S^{(k)} (H^*)\subset 
}
where
$$
S^{(k)}(H^*) =\bigoplus_{j=0}^k S^j(H^*).
$$
Let
$$
S(H^*)= \lim_{n\rightarrow \infty} S^{(k)}(H^*),
$$
which is a super-commutative and associative filtered algebra over $\Bbbk$ isomorphic
to $\Bbbk[[t_H]]$.

The product $m_2$ on $H$ is a bilinear map on $S^2(H)$ into $H$ of ghost number $0$,
since it is super-commutative.
The product $m_2$ is specified
by structure constants $\{m_{\a\b}{}^\g\}$ such that 
$$
\eqalign{
m_2(e_\a, e_\b)= m_{\a\b}{}^\g e_\g,
\qquad
m_2(e_0, e_\b)=m_2(e_\b, e_0)= e_\b,
}
$$
where we are using the Einstein summation conventions that a repeated upper 
and lower index is summed over. Note that $m_{\b 0}{}^\g= m_{0\b}{}^\g =\d_{\b}{}^\g$ 
(the Kronecker delta).
The binary map $m_2: S^2(H)\rightarrow H$ is dualize to $m^*_2 :H^*\rightarrow S^2(H^*)$
and is extended uniquely to a $\Bbbk$-linear map $m_2^\sharp :S(H^*)\rightarrow S(H^*)$
as a derivation of ghost number zero. Explicitly
$m_2^* t^\g = \Fr{1}{2}t^\b t^\a m_{\a\b}{}^\g$ and
$$
m_2^\sharp= \Fr{1}{2}t^\b t^\a m_{\a\b}{}^\g \rd_\g
$$
where the derivative symbol
$\rd_\g=\Fr{\rd}{\rd t^\g}$ means the extension of $m_2^*$ as a derivation. 
Then we have
$$
\left[\rd_0, m_2^\sharp\right]= t^\a \rd_\a,
$$
since $m_{\b 0}{}^\g=m_{0\b}{}^\g = \d_{\b}{}^\g$.
Any $\Bbbk$-multilinear map $m_n: S^n(H)\longrightarrow H$ of
ghost number zero is similarly dualized to a $\Bbbk$-linear map $m^*_n: H^* \longrightarrow
S^n(H^*)$ of ghost number zero, which can be extended uniquely  to
a $\Bbbk$-linear map $m^\sharp_n: S(H^*) \longrightarrow
S(H^*)$  as a derivation of ghost number zero - let
$m_n(e_{\a_1}, \cdots, e_{\a_n})= m_{\a_1\cdots \a_n}{}^\g e_{\g}$,
then 
$m^*_n t^\g =\Fr{1}{n!}t^{\a_n}\cdots t^{\a_1}m_{\a_1\cdots \a_n}{}^\g$ and
$m^\sharp_n =\Fr{1}{n!}t^{\a_n}\cdots t^{\a_1}m_{\a_1\cdots \a_n}{}^\g\rd_\g$.
We shall often use the single notation $m_n$ for  $m_n,m^*_n,m_n^\sharp$.

Now the triple $\big(\Bbbk[[t_H]]\otimes\sC[[\hbar]], \bos{K}, \;\cdot\;\big)$ is a BV QFT
algebra, where $\bos{K}$ and $\cdot$ are the shorthand notations 
for $1\otimes \bos{K}$ and 
$
(a\otimes \bos{x})\cdot (b\otimes \bos{y})= (-1)^{|\bos{x}||{b}|}ab\otimes \bos{x}\cdot\bos{y}
$
for $a,b \in \Bbbk[[t_H]]$ and $\bos{x},\bos{y}\in \sC[[\hbar]]$, respectively.
We denote its descendant algebra by 
$\big(\Bbbk[[t_H]]\otimes\sC[[\hbar]], \bos{K}, (\hbox{ },\hbox{ })\big)$,
where 
$
\big( a\otimes \bos{x}, b\otimes \bos{y}\big)
= (-1)^{(|\bos{x}|+1)|{b}|}ab\otimes \big(\bos{x},\bos{y}\big)$. 
The operator $m_n^\sharp$ acts on $\Bbbk[[t_H]]\otimes\sC[[\hbar]]$
as a derivation, $m_n^\sharp\otimes 1$, increasing the word length of $t_H$ by $n-1$.
We shall omit the tensor product
symbol whenever possible. 

Let $O_\a= f(e_\a)$ and $\bos{O}_\a = \bos{f}(e_\a)$.
Then $\{O_\a\}$ is  a set of representative of the basis $\{e_a\}$ of
$H$ such that $Q O_\a=0$ and $[O_\a]=e_\a$.
The set $\{\bos{O}_\a\}$ is 
then a fixed  quantization of the generating set $\{O_\a\}$ of classical observables
such that $\bos{K}\bos{O}_\a=0$ and $\bos{O}_\a\big|_{\hbar=0}=O_\a$.
Finally we let $\bos{\Theta}_1= t^\a\bos{O}_\a$. It follows that
$$
\bos{K}\bos{\Theta}_1=0,\qquad \rd_0\bos{\Theta}_1=1.
$$
Now the following theorem contains
the complete information of quantum correlation functions;
\begin{theorem}
On $H$
there is a sequence $m_2, m_3, m_4, \cdots$ of multilinear products 
$m_n:S^n H\rightarrow H$ of ghost number $0$ such that
$m_2(e_0, e_\a)= e_\a$ and $m_n(e_0,e_{\a_2},\cdots, e_{\a_n})=0$ for all $n=3,4,5,\cdots$.
And, 
there is a family of BV QFTs specified by
$$
\bos{\Theta}=\bos{\Theta}_1 + \bos{\Theta}_2+\bos{\Theta}_3+\cdots
\in \big(\Bbbk[[t_H]]\otimes \sC[[\hbar]]\big)^0,
$$
where $\;\bos{\Theta}_n =\Fr{1}{n!}t^{\a_n}\cdots t^{\a_1} \bos{O}_{\a_1\cdots\a_n}
 \in \big(S^n(H^*)\otimes \sC[[\hbar]]\big)^0$,
satisfying

\begin{enumerate}
\item quantum master equation:
$$
\eqalign{
0=&\bos{K}\bos{\Theta}_1
,\cr
\hbar \bos{\Theta}_2  =&\Fr{1}{2}\bos{\Theta}_1\cdot \bos{\Theta}_1- m^\sharp_2 \bos{\Theta}_1 - \bos{K} \La_2
,\cr
\hbar \bos{\Theta}_3  =&
\Fr{2}{3}\bos{\Theta}_1\cdot \bos{\Theta}_2
-\Fr{1}{3}m_2^\sharp\bos{\Theta}_2
-\Fr{1}{3}\big(\bos{\Theta}_1,\La_2\big) 
-m_3^\sharp\bos{\Theta}_1
-\bos{K}\La_3
,\cr
\vdots\;&
\cr
\hbar\bos{\Theta}_{n} =& 
\sum_{k=1}^{n-1}\Fr{ k(n-k)}{n(n-1)} \bos{\Theta}_{k}\cdot  \bos{\Theta}_{n-k}
-\sum_{k=2}^{n-1}\Fr{k(k-1)}{n(n-1)}\left(  m^\sharp_{k} \bos{\Theta}_{n-k+1}
+ \big(\bos{\Theta}_{n-k},\La_{k}\big)\right)
\cr
&
- m^\sharp_n \bos{\Theta}_1-\bos{K}\La_n
,\cr
\vdots\;&
}
$$
for some $\La_n \in \big(\Bbbk[[t_H]]\otimes
\sC\big)^{-1}$ defined  modulo $\Ker \bos{K}$

\medskip
\item quantum identity: $\rd_0 \bos{\Theta}=1$.
\item quantum descendant equation
$$
\bos{K}\bos{\Theta} +\Fr{1}{2}\big(\bos{\Theta},\bos{\Theta}\big)=0,
$$
as  a consequence of quantum master equation.
\end{enumerate}
\end{theorem}

\begin{remark}
The conditions 
$m_2(e_0, e_\a)= e_\a$ and $m_n(e_0,e_{\a_2},\cdots, e_{\a_n})=0$ for $n\geq 3$
are equivalent to 
$$
\left[\rd_0, m_2^\sharp\right]= t^\a\rd_\a, \qquad
\left[\rd_0, m_n^\sharp\right]=0 \hbox{ for } n\geq 3.
$$
\end{remark}

In section $4.1.$ an idea of  proof will be presented for a pedagogical reason
before an actual proof in section $4.2$. In section $4.3$
we shall derived the algebra of quantum correlation functions.
 In section $4.4$ we shall discuss some
corollaries of our theorem comparing our notion of quantum coordinates with
the flat coordinates on moduli spaces topological strings.

\subsection{Idea of Proof}

The quantum master equation to be consistent its classical limit should make sense
as well. The classical limit of quantum master equation modulo $t_H^{n}$, $n\geq 3$, is
$Q \Theta_1=0$,
and 
\eqn\cuta{
\CM_k= m_k^\sharp \Theta_1 + \bos{K}\La_k
}
for $k=2,3,\cdots, n-1$,
where $\CM_k \in \big(S^k(H^*)\otimes \sC\big)^0$ is given by
$$
\eqalign{
\CM_{k}
:=&\Fr{1}{k(k-1) }\sum_{\ell=1}^{k-1} \ell(k-\ell) {\Theta}_{\ell}\cdot  {\Theta}_{k-\ell}
-\Fr{1}{k(k-1) }\sum_{\ell=2}^{k-1}\ell(\ell-1)
\left( {\Theta}_{k-\ell},\La_{\ell}\right)
\cr
&
-\Fr{1}{k(k-1) }\sum_{\ell=2}^{k-1}(k-\ell+1)(k-\ell) m^\sharp_{k-\ell+1} {\Theta}_{\ell}.
}
$$
Then $\CM_k$ should belong to $\Ker Q$ to make sense of the equation \cuta.
We may decompose $\CM_{k}$ as
$$
\CM_k =\Fr{1}{k!}t^{\a_k}\cdots t^{\a_1} M_{\a_1\cdots \a_k}
$$
such that $M_{\a_1\cdots \a_k} \in \sC^{|e_{\a_1}|+\cdots +|e_{\a_k}|}$.
Thus $M_{\a_1\cdots\a_k}$ must satisfy
$Q M_{\a_1\cdots\a_k}=0$. Then the expression $M_{\a_1\cdots\a_k}$
can be written as 
\eqn\cutb{
M_{\a_1\cdots\a_k}= m_{\a_1 \cdots \a_k}{}^\g O_\g +Q \l_{\a_1\cdots \a_k}
}
for uniquely defined set of constants $\{m_{\a_1\cdots\a_k}{}^\g\}$ and for some $\l_{\a_1\cdots\a_k}
\in \sC^{|e_{\a_1}|+\cdots +|e_{\a_k}|-1}$ defined modulo $\Ker Q$.
Once we make the following identifications
$$
\eqalign{
m_k^\sharp &= \Fr{1}{k!}t^{\a_k}\cdots t^{\a_1} m_{\a_1\cdots \a_k}{}^\g \rd_\g
,\cr
\La_k& =\Fr{1}{k!}t^{\bar\a_k}\cdots t^{\bar\a_1} \l_{\a_1\cdots \a_k},
}
$$
where
$t^{\bar \a} = (-1)^{|e_\a|} t^\a$, the equations \cuta\ and \cutb\ are equivalent.  

Set $n=3$, to begin with, we have $\CM_2=\Fr{1}{2}\Theta_1\cdot \Theta_1 \in \Ker Q$.
 Thus $\CM_2=m_2^\sharp \Theta_1 + Q \La_2$. We define
$\bos{\Theta}_2\in\big(S^2(H^*)\otimes \sC[[\hbar]]\big)^0$ by the formula 
$$
\bos{\Theta}_2:=\Fr{1}{\hbar}\left( \Fr{1}{2}\bos{\Theta}_1\cdot \bos{\Theta}_1 
- m_2^\sharp \bos{\Theta}_1 -\bos{K} \La_2\right)
$$
and show that $\rd_0 \bos{\Theta}_2=0$.
 Then we take the classical limit
$\Theta_2$ of $\bos{\Theta}_2$ and show that 
$\CM_3 =\Fr{2}{3}\Theta_1\cdot \Theta_2 -\Fr{1}{3}m_2^\sharp \Theta_2 -\Fr{1}{3}\big(\Theta_1,\La_2\big)$ 
satisfies $Q\CM_3=0$, so that
$\CM_3$ can be expressed as $\CM_3=m_3^\sharp \Theta_1 + Q \La_3$. 
Thus we are defining
$m_3^\sharp$ and $\La_3$ to proceed one step further, after showing that 
$\left[\rd_0,m_3^\sharp\right]=\rd_0 \La_3=0$, to 
define
$\bos{\Theta}_3\in\big(S^3(H^*)\otimes \sC[[\hbar]]\big)^0$ by the formula 
$$
\bos{\Theta}_3:=\Fr{1}{\hbar}\left(\Fr{2}{3}\bos{\Theta}_1\cdot \bos{\Theta}_2
 -\Fr{1}{3}m_2^\sharp \bos{\Theta}_2 -\Fr{1}{3}\big(\bos{\Theta}_1,\La_2\big)\right)
$$
and prove that $\rd_0 \bos{\Theta}_3=0$, et cetera, ad infinitum.

We are going to build an inductive system
$$
\bos{P}(1)\subset\bos{P}(2) \subset \bos{P}(3)\subset \cdots\subset \bos{P}(n-1)\subset \bos{P}(n),
$$
with respect to the filtration \filt\ such that
$\bos{P}(k)$ for each $2\leq k\leq n-1$
is a triplet;
$$
\bos{P}(k+1)=\left\{
\eqalign{
\bos{\Theta}^{(k+1)}&=\bos{\Theta}^{(k)} +\bos{\Theta}_{k+1} \in \big(S^{(k+1)}(H^*)\otimes \sC[[\hbar]]\big)^0
,\cr
m^{(k+1)}&=m^{(k)} + m_{k+1} :H^*\rightarrow S^{(k+1)}(H^*)
,\cr
\La^{(k+1)}&=\La^{(k)}+\La_{k+1} \in \big(S^{(k+1)}(H^*)\otimes \sC[[\hbar]]\big)^{-1}
}
\right.
$$
satisfying quantum master equation on $S^{(k+1)}(H^*)\otimes \sC[[\hbar]]$, i.e.\ modulo $t_H^{k+2}$,
with the initial conditions that $\bos{\Theta}^{(1)}=\bos{\Theta}_1$ and $m^{(1)}=\La^{(1)}=0$.
Then we send $n\rightarrow \infty$.

\subsubsection{$\bos{P}(1)\subset\bos{P}(2)$.}
Set 
$\bos{P}(1)=\left\{\bos{\Theta}^{(1)}=\bos{\Theta}_1, 0, 0\right\}$ so
that
$\bos{K}\bos{\Theta}_1=0$ and $\rd_0 \bos{\Theta}_1=1$.
 Let 
$\Theta_1=t^\a O_\a\in \big(H^*\otimes \sC\big)^0$ denote the classical limit of 
$\bos{\Theta}_1$. Then $Q\Theta_1=0$ and $\rd_0 {\Theta}_1=1$. 
Consider the expression
$\bos{\Theta}_1\cdot \bos{\Theta}_1 \in \big(S^2 (H^*)\otimes \sC\big)[[\hbar]]^0$
built from $\bos{P}(1)$,
which satisfies
\eqn\devtwo{
\bos{K}\left(\bos{\Theta}_1\cdot \bos{\Theta}_1\right) = -\hbar \big(\bos{\Theta}_1,\bos{\Theta}_1\big).
}
since $\bos{K}\bos{\Theta}_1=0$. Thus the classical
limit $\Theta_1\cdot \Theta_1$ of $\bos{\Theta}_1\cdot \bos{\Theta}_1$
belongs to $\Ker Q\cap  \big(S^2(H^*)\otimes \sC\big)^0$. It follows that
\eqn\mevtwo{
\Fr{1}{2}{\Theta}_1\cdot {\Theta}_1= m_2 \Theta_1 + Q \La_2,
}
for uniquely defined map $m_2: H^*\longrightarrow S^2(H^*)$ of ghost number $0$
and some $\La_2 \in \big(S^2(H^*)\otimes \sC\big)^{-1}$ defined modulo $\Ker Q$.
We note that the ghost number $0$ map $m_2: H^*\longrightarrow S^2(H^*)$ has
an unique extension to $m_2^\sharp :S(H^*)\rightarrow S(H^*)$ as a derivation.

From $1\cdot  {\Theta}_1=\Theta_1$, we deduce that  $\left[\rd_0, m_2^\sharp\right]= t^\a\rd_\a$ and
$\rd_0 \La_2=0$. Fix such a $\La_2$.
It follows that the expression
$$
\Fr{1}{2}\bos{\Theta}_1\cdot \bos{\Theta}_1- m_2^\sharp \bos{\Theta}_1 - \bos{K} \La_2
$$
is divisible by $\hbar$ and does not depend on $t^0$, i.e.,
$$
\rd_0\left(
\Fr{1}{2}\bos{\Theta}_1\cdot \bos{\Theta}_1- m_2 \bos{\Theta}_1 - \bos{K} \La_2\right)
=1\cdot \bos{\Theta}_1- t^\a\rd_\a \bos{\Theta}_1 =0,
$$
where we have used that $t^\a\rd_\a \bos{\Theta}_1=\bos{\Theta}_1$ since $\bos{\Theta}_1$
is a degree $1$ homogeneous polynomial of  $t_H$. 
Thus we can 
define $\bos{\Theta}_2 \in \big(S^2 (H^*)\otimes \sC\big)[[\hbar]]^0$
by the formula
\eqn\levtwo{
\hbar \bos{\Theta}_2 =\Fr{1}{2}\bos{\Theta}_1\cdot \bos{\Theta}_1- m_2^\sharp \bos{\Theta}_1 - \bos{K} \La_2,
}
which does not depend on $t^0$, i.e., $\rd_0\bos{\Theta}_2=0$.
Applying $\bos{K}$ to the above we obtain that
$\hbar\bos{K} \bos{\Theta}_2 =-\Fr{\hbar}{2}\big(\bos{\Theta}_1, \bos{\Theta}_1\big)$,
using the properties \devtwo, $\bos{K}\bos{\Theta}_1=0$ and $\bos{K}^2=0$. 
Thus we conclude that
\eqn\dstwo{
\bos{K}\bos{\Theta}_2 +\Fr{1}{2}\big(\bos{\Theta}_1, \bos{\Theta}_1\big)=0.
}

We record all the previous data
by the triplet $\bos{P}(2)$;
$$
\bos{P}(2)=
\left\{
\eqalign{
\bos{\Theta}^{(2)}&:=\bos{\Theta}_{1}+\bos{\Theta}_{2}
\in \left(S^{(2)}(H^*)\otimes \sC[[\hbar]]\right)^0
,\cr
m^{(2)\sharp}&:= m_2^\sharp: S(H^*)\longrightarrow S(H^*)
,\cr
\La^{(2)}&:=\La_2 
\in \left(S^{(2)}(H^*)\otimes \sC\right)^{-1},
}\right.
$$
satisfying 
$\rd_0 \bos{\Theta}^{(2)}=\rd_0\bos{\Theta}_1=1$, $[\rd_0, m^{(2)}]=[\rd_0, m_2]=t^\a\rd_\a$,
$\rd_0 \La_2=0$,
and the quantum master equation modulo $t_H^3$;
$$
\eqalign{
0&=\bos{K}\bos{\Theta}_1
,\cr
\hbar \bos{\Theta}_2  &=\Fr{1}{2}\bos{\Theta}_1\cdot \bos{\Theta}_1
- m_2^\sharp \bos{\Theta}_1 - \bos{K} \La_2
,\cr
}
$$
which implies the quantum descendant equation modulo $t_H^3$
$$
\eqalign{
\bos{K}\bos{\Theta}_1=0
,\cr
\bos{K}\bos{\Theta}_2 +\Fr{1}{2}\big(\bos{\Theta}_1, \bos{\Theta}_1\big)=0.
}
$$

\begin{corollary}
The bracket $(\hbox{ },\hbox{ })$ vanishes on the $Q$-cohomology $H$.
\end{corollary}
\begin{proof}
Taking the classical limit of \dstwo, we obtain that
$$
\left(\Theta_1,\Theta_1\right)= -2Q \Theta_2
$$
where $\Theta_1 = t^a O_\a=t^\a f(e_\a)$. 
This imply the corollary since $\{e_\a=[O_\a]\}$ form of a basis of $H$.
\qed
\end{proof}

\subsubsection{$\bos{P}(2)\subset\bos{P}(3)$.}
The following explicit description may be redundant  but is presented here to demonstrate
the method of proof.

Let $\mathbb{M}_3 \in \big(S^3(H^*)\otimes\sC\big)[[\hbar]]^0$ be the following
expression
$$
\mathbb{M}_3
=\Fr{1}{3}\bos{\Theta}_1\cdot \bos{\Theta}_2+\Fr{1}{3}\bos{\Theta}_2\cdot \bos{\Theta}_1
-\Fr{1}{3}m_2^\sharp\bos{\Theta}_2
-\Fr{1}{3}\big(\bos{\Theta}_1,\La_2\big) 
,
$$
which is defined from data in the system $\bos{P}(2)$. 

\begin{proposition}\label{devthree}
We have
$\bos{K}\mathbb{M}_3
=-\hbar\left(\bos{\Theta}_1, \bos{\Theta}_2\right)$
and $\rd_0 \mathbb{M}_3=0$.
\end{proposition}
 
\begin{proof}
Note that $\bos{\Theta}_1\cdot \bos{\Theta}_2
=\bos{\Theta}_2\cdot \bos{\Theta}_1$.
We have, by a direct computation,
$$
\bos{K}\mathbb{M}_3
=-\Fr{2\hbar}{3}\left(\bos{\Theta}_1, \bos{\Theta}_2\right)
+\Fr{2}{3}\bos{\Theta}_1\cdot \bos{K}\bos{\Theta}_2
-\Fr{1}{3}m_2^\sharp\bos{K}\bos{\Theta}_2
+\Fr{1}{3}\big(\bos{\Theta}_1,\bos{K}\La_2\big).
$$
Using $\bos{K}\bos{\Theta}_2=-\Fr{1}{2}\big(\bos{\Theta}_1,\bos{\Theta}_1\big)$, we obtain that
$$
\bos{K}\mathbb{M}_3
=-\hbar\left(\bos{\Theta}_1, \bos{\Theta}_2\right)
+\Fr{1}{6}\left(\bos{\Theta}_1,\bos{\Theta}_1\cdot \bos{\Theta}_1\right)
-\Fr{1}{3}\bos{\Theta}_1\cdot\left(\bos{\Theta}_1, \bos{\Theta}_1\right)
+\Fr{1}{6}m_2^\sharp\left(\bos{\Theta}_1, \bos{\Theta}_1\right)
-\Fr{1}{3}\left(\bos{\Theta}_1, m_2^\sharp\bos{\Theta}_1\right).
$$
Thus, after the Leibniz law of the bracket $\big(\hbox{ },\hbox{ }\big)$
and the fact that $m_2^\sharp$ is a derivation of the bracket, we have the
first claim. For the second claim, we obtain that
$$
\rd_0\mathbb{M}_3= \Fr{2}{3}\bos{\Theta}_2 - \Fr{1}{3}t^\a\rd_\a \bos{\Theta}_2=0,
$$
where we have used the properties that
$\rd_0 \bos{\Theta}_1 =1$, $\rd_0 \bos{\Theta}_2 =0$, $\rd_0 \La_2=0$, $\big(1, \La_2\big)=0$, 
$\left[\rd_0, m_2^\sharp\right]=t^\a\rd_\a$
and $t^\a\rd_\a \bos{\Theta}_2=2\bos{\Theta}_2$, since $\bos{\Theta}_2$ is a degree $2$
homogeneous polynomial of $t_H$.
 \qed
\end{proof}

Thus the classical
limit $\CM_3$ of $\bos{M}_3$
belongs to $\Ker Q\cap  \big(S^3(H^*)\otimes \sC\big)^0$ and $\rd_0 \CM_3=0$,
 where
$$
{\CM}_3
=\Fr{1}{3}{\Theta}_1\cdot {\Theta}_2+\Fr{1}{3}{\Theta}_2\cdot {\Theta}_1
-\Fr{1}{3}m_2^\sharp{\Theta}_2
-\Fr{1}{3}\big({\Theta}_1,\La_2\big). 
$$
%
 It follows that
\eqn\mevthree{
\CM_3= m_3^\sharp \Theta_1 + Q \La_3,
}
for uniquely defined map $m_3^\sharp: S(H^*)\rightarrow S(H^*)$ of ghost number $0$
and some $\La_3 \in \big(S^3(H^*)\otimes \sC\big)^{-1}$ defined modulo $\Ker Q$
such that $[\rd_0, m_3^\sharp]=0$ and $\rd_0 \La_3=0$. Fix a $\La_3$.
Then,  the expression
$
\mathbb{M}_3- m_3^\sharp \bos{\Theta}_1 - \bos{K} \La_2
$
must be divisible by $\hbar$ and independent to $t^0$. Thus we can 
define $\bos{\Theta}_3 \in \big(S^3 (H^*)\otimes \sC\big)[[\hbar]]^0$
by the formula
\eqn\levthree{
\hbar \bos{\Theta}_3 =\mathbb{M}_3- m_3^\sharp \bos{\Theta}_1 - \bos{K} \La_3,
}
which does not depend on $t^0$, i.e., $\rd_0 \bos{\Theta}_3=0$.
Applying $\bos{K}$ to the above we have
$\hbar\bos{K} \bos{\Theta}_2 =\bos{K}\mathbb{N}_3$.
Then, from proposition \ref{devthree},
we conclude that
\eqn\dsthree{
\bos{K}\bos{\Theta}_3 +\big(\bos{\Theta}_1, \bos{\Theta}_2\big)=0.
}
Thus
we have
the system $\bos{P}(3)$;
$$
\bos{P}(3)=
\left\{
\eqalign{
\bos{\Theta}^{(3)}&:=\bos{\Theta}^{(2)}+\bos{\Theta}_{3}
=
\bos{\Theta}_{1}+\bos{\Theta}_{2}+\bos{\Theta}_{3}
\in \left(S^{(3)}(H^*)\otimes \sC\right)[[\hbar]]^0
,\cr
m^{(3)\sharp}&:= m^{(2)\sharp}+m_3^\sharp=m_2^\sharp+m_3^\sharp: S(H^*)\longrightarrow S(H^*)
,\cr
\La^{(3)}&:=\La^{(2)}+\La_3=\La_2 +\La_3
\in \left(S^{(3)}(H^*)\otimes \sC\right)^{-1}\mod \Ker Q,
}\right.
$$
satisfying $\rd_0 \bos{\Theta}^{(3)}=\rd_0\bos{\Theta}_1=1$, 
$[\rd_0, m^{(3)\sharp}]=[\rd_0, m_2^\sharp]=t^\a\rd_\a$, $\rd_0 \La^{(3)}=0$,
and
the quantum master equation modulo $t_H^4$;
$$
\eqalign{
0&=\bos{K}\bos{\Theta}_1
,\cr
\hbar \bos{\Theta}_2  &
=\Fr{1}{2}\bos{\Theta}_1\cdot \bos{\Theta}_1- m_2^\sharp \bos{\Theta}_1 - \bos{K} \La_2
,\cr
\hbar \bos{\Theta}_3  &=
\Fr{2}{3}\bos{\Theta}_1\cdot \bos{\Theta}_2
-m_3^\sharp\bos{\Theta}_1
-\Fr{1}{3}m_2^\sharp\bos{\Theta}_2
-\bos{K}\La_3
 -\Fr{1}{3}\big(\bos{\Theta}_1,\La_2\big) 
}
$$
which implies quantum descendant equation modulo $t_H^4$
$$
\eqalign{
\bos{K}\bos{\Theta}_1=0
,\cr
\bos{K}\bos{\Theta}_2 +\Fr{1}{2}\big(\bos{\Theta}_1, \bos{\Theta}_1\big)=0
,\cr
\bos{K}\bos{\Theta}_3 +\big(\bos{\Theta}_1, \bos{\Theta}_2\big)=0
.
}
$$

\subsection{Proof}

\subsubsection{$\bos{P}(n-1)$.}

Fix $n > 3$,  assume that we have the following  inductive system
\eqn\sysi{
\bos{P}(1)\subset \bos{P}(2) \subset \bos{P}(3)\subset \cdots\subset \bos{P}(n-1),
}
where, for each $1\leq j \leq n-1$,  
$
\bos{P}(j)=\left\{
\bos{\Theta}^{(j)},
m^{(j)},
\La^{(j)}
\right\}
$
is a system with
$$
\left\{
\eqalign{
\bos{\Theta}^{(j)}&:=\bos{\Theta}_{1}+\bos{\Theta}_{2}+\cdots+\bos{\Theta}_{j}
\in \left(S^{(j)}(H^*)\otimes \sC\right)[[\hbar]]^0
,\cr
m^{(j)\sharp}&:= m_2^\sharp+m_3^\sharp+\cdots+m_j^\sharp: H^*\longrightarrow S^{(j)}(H^*)
,\cr
\La^{(j)}&:=\La_2 +\La_3+\cdots +\La_j
\in \left(S^{(j)}(H^*)\otimes \sC\right)^{-1},
}\right.
$$
satisfying $\rd_0 \bos{\Theta}^{(j)}=1$, $[\rd_0, m^{(j)}]=t^\a\rd_\a$, $\rd_0\La^{(j)}=0$,
and the
quantum master and descendant equations modulo $t_H^{j+1}$;
\eqn\eea{
\left\{
\eqalign{
0&=\bos{K}\bos{\Theta}_1
,\cr
\hbar \bos{\Theta}_{2}&=\mathbb{M}_2 - m_2^\sharp \bos{\Theta}_1 -\bos{K}\La_2
,\cr
&\;\vdots
\cr
\hbar \bos{\Theta}_{j}&=\mathbb{M}_j - m_j^\sharp \bos{\Theta}_1 -\bos{K}\La_j,
}
\right.
\qquad
\left\{
\eqalign{
\bos{K}\bos{\Theta}_1&=0
,\cr
\bos{K}\bos{\Theta}_2 +\Fr{1}{2}\left( \bos{\Theta}_{1}, \bos{\Theta}_{1}\right)&=0
,\cr
&\;\vdots
\cr
\bos{K}  \bos{\Theta}_{j}
 +\Fr{1}{2}  \sum_{\ell=1}^{j-1}\left( \bos{\Theta}_{\ell}, \bos{\Theta}_{j-\ell}\right)&=0,
}\right.
}
where, for $k=2,\cdots, j$,
$$
\eqalign{
\mathbb{M}_{k}
:=&\Fr{1}{k(k-1) }\sum_{\ell=1}^{k-1} \ell(k-\ell) \bos{\Theta}_{\ell}\cdot   \bos{\Theta}_{k-\ell}
-\Fr{1}{k(k-1) }\sum_{\ell=2}^{k-1}\ell(\ell-1)
\left( \bos{\Theta}_{k-\ell},\La_{\ell}\right)
\cr
&
-\Fr{1}{k(k-1) }\sum_{\ell=2}^{k-1}(k-\ell+1)(k-\ell) m_{k-\ell+1}^\sharp \bos{\Theta}_{\ell}.
}
$$
(Note that the expression
$\mathbb{M}_{k}\in \left(S^k(H^*)\otimes \sC[[\hbar]]\right)^0$ 
is built from  data of
$\bos{P}({k-1})$). 

\subsubsection{$\bos{P}(n-1)\subset \bos{P}(n)$.}

We are going to extend the system \sysi\ to
$$
\bos{P}(1)\subset \bos{P}(2) \subset \bos{P}(3)\subset \cdots\subset \bos{P}(n-1) \subset \bos{P}(n)
$$
such that $\bos{\Theta}_n =\bos{\Theta}^{(n)}-\bos{\Theta}^{(n-1)}$
and $m_n^\sharp =m^{(n)\sharp} -m^{(n-1)\sharp}$ are defined uniquely, while 
$\La_n =\La^{(n)}-\La^{(n-1)}$ is defined modulo $\Ker \bos{K}$.
Then we take $n\rightarrow \infty$ limit.

We shall need the following technical proposition:
\begin{proposition}\label{tosca}
The expression $\mathbb{M}_{n}\in \left(S^n(H^*)\otimes \sC[[\hbar]]\right)^0$;
$$
\eqalign{
\mathbb{M}_{n} := 
&
\Fr{1}{n(n-1)}\sum_{k=1}^{n-1} k(n-k) \bos{\Theta}_{k}\cdot  \bos{\Theta}_{n-k}
-\Fr{1}{n(n-1)}\sum_{k=2}^{n-1}(n-k+1)(n-k) m_{n-k+1}^\sharp \bos{\Theta}_{k}
\cr
&
-\Fr{1}{n(n-1)}\sum_{k=2}^{n-1}k(k-1)
\left( \bos{\Theta}_{n-k},\La_{k}\right),
}
$$
which is defined in terms of $\bos{P}(n-1)$,
satisfies
$$
\bos{K}\mathbb{M}_{n} = - \Fr{\hbar}{2} \sum_{k=1}^{n-1}\left( \bos{\Theta}_{k}, \bos{\Theta}_{n-k}\right).
$$
and
$\rd_0 \mathbb{M}_{n} =0$.
\end{proposition}
We shall postpone proving the above proposition and examine its consequences first.
An immediate consequence is that
the classical limit 
$\CM_{n}\in \left(S^n(H^*)\otimes \sC\right)^0$ of
$\mathbb{M}_{n}$ satisfies $Q \CM_{n}=0$ and is independent of $t^0$, i.e., $\rd_0 \CM_n=0$,
where
$$
\eqalign{
\CM_{n} := 
&
\Fr{1}{n(n-1)}\sum_{k=1}^{n-1} k(n-k) {\Theta}_{k}\cdot  {\Theta}_{n-k}
-\Fr{1}{n(n-1)}\sum_{k=2}^{n-1}(n-k+1)(n-k) m_{n-k+1}^\sharp {\Theta}_{k}
\cr
&
-\Fr{1}{n(n-1)}\sum_{k=2}^{n-1}k(k-1)
\left({\Theta}_{n-k},\La_{k}\right).
}
$$

It follows that
\eqn\mevn{
\CM_n= m_n^\sharp \Theta_1 + Q \La_n,
}
for uniquely defined map $m_n^\sharp: S(H^*)\rightarrow S(H^*)$ of ghost number $0$
and some $\La_n \in \big(S^n(H^*)\otimes \sC\big)^{-1}$ defined modulo $\Ker Q$
such that $[\rd_0, m_n^\sharp]=0$ and $\rd_0 \La_n=0$.
Then,  the expression
$
\mathbb{M}_n- m_n \bos{\Theta}_1 - \bos{K} \La_n
$
must be divisible by $\hbar$. Thus we can 
define $\bos{\Theta}_n \in \big(S^n (H^*)\otimes \sC\big)[[\hbar]]^0$
by the formula
\eqn\levn{
\hbar \bos{\Theta}_n =\mathbb{M}_n- m_n^\sharp \bos{\Theta}_1 - \bos{K} \La_n.
}
Applying $\bos{K}$ to the above we have
$\hbar\bos{K} \bos{\Theta}_n =\bos{K}\mathbb{M}_n$.
Then, from proposition \ref{tosca}, 
we conclude that
\eqn\dsnl{
\bos{K}\bos{\Theta}_n +\Fr{1}{2} \sum_{k=1}^{n-1}\left( \bos{\Theta}_{k}, \bos{\Theta}_{n-k}\right),
}
and $\rd_0 \bos{\Theta}_n=0$.

Thus we have defined $\bos{P}(n)$;
$$
\bos{P}(n)=
\left\{
\eqalign{
\bos{\Theta}^{(n)}&:=\bos{\Theta}^{(n-1)}+\bos{\Theta}_{n}
\in \left(S^{(n)}(H^*)\otimes \sC[[\hbar]]\right)^0
,\cr
m^{(n)\sharp}&:= m^{(n-1)\sharp}+m_n^\sharp: S(H^*)\longrightarrow S(H^*)
,\cr
\La^{(n)}&:=\La^{(n-1)} +\La_n
\in \left(S^{(n)}(H^*)\otimes \sC\right)^{-1},
}\right.
$$
satisfying $\rd_0 \bos{\Theta}^{(n)}=1$, $\left[\rd_0, m^{(n)\sharp}\right]=t^\a \rd_\a$, $\rd_0 \La^{(n)}=0$
and
\eqn\eea{
\left\{
\eqalign{
0&=\bos{K}\bos{\Theta}_1
,\cr
\hbar \bos{\Theta}_{2}&=\mathbb{M}_2 - m_2^\sharp \bos{\Theta}_1 -\bos{K}\La_2
,\cr
&\;\vdots
\cr
\hbar \bos{\Theta}_{n}&=\mathbb{M}_n - m_n^\sharp \bos{\Theta}_1 -\bos{K}\La_n,
}
\right.
\qquad
\left\{
\eqalign{
\bos{K}\bos{\Theta}_1&=0
,\cr
\bos{K}\bos{\Theta}_2 +\Fr{1}{2}\left( \bos{\Theta}_{1}, \bos{\Theta}_{1}\right)&=0
,\cr
&\;\vdots
\cr
\bos{K}  \bos{\Theta}_{n}
 +\Fr{1}{2}  \sum_{\ell=1}^{n-1}\left( \bos{\Theta}_{\ell}, \bos{\Theta}_{n-\ell}\right)&=0,
}\right.
}
where, for $k=2,\cdots, n$,
$$
\eqalign{
\mathbb{M}_{k}
:=&\Fr{1}{k(k-1) }\sum_{\ell=1}^{k-1} \ell(k-\ell) \bos{\Theta}_{\ell}\cdot   \bos{\Theta}_{k-\ell}
-\Fr{1}{k(k-1) }\sum_{\ell=2}^{k-1}\ell(\ell-1)
\left( \bos{\Theta}_{k-\ell},\La_{\ell}\right)
\cr
&
-\Fr{1}{k(k-1) }\sum_{\ell=2}^{k-1}(k-\ell+1)(k-\ell) m_{k-\ell+1}^\sharp \bos{\Theta}_{\ell}.
}
$$

Now we set
$$
\bos{\Theta} =\lim_{n\rightarrow \infty}\bos{\Theta}^{(n)}
$$ 
and the theorem follows once we prove proposition \ref{tosca}.
%
%
%

\subsubsection{Proof of proposition \ref{tosca}.}

Let $\tilde \mathbb{M}_n=n(n-1)\mathbb{M}$. Then, from a direct computation,
we have
$$
\eqalign{
\bos{K}\tilde\mathbb{M}_{n} = 
&
-\hbar\sum_{k=1}^{n-1} k(n-k)\left( \bos{\Theta}_{k},  \bos{\Theta}_{n-k}\right)
\cr
&
+\sum_{k=1}^{n-1} k(n-k)\bos{K} \bos{\Theta}_{k}\cdot  \bos{\Theta}_{n-k}
+\sum_{k=1}^{n-1} k(n-k) \bos{\Theta}_{k}\cdot \bos{K} \bos{\Theta}_{n-k}
\cr
&
-\sum_{k=2}^{n-1}k(k-1)
\left(\bos{K} \bos{\Theta}_{n-k},\La_{k}\right)
+\sum_{k=2}^{n-1}k(k-1)
\left( \bos{\Theta}_{n-k},\bos{K}\La_{k}\right)
\cr
&
-\sum_{k=2}^{n-1}(n-k+1)(n-k) m_{n-k+1}^\sharp\bos{K} \bos{\Theta}_{k},
}
$$
where we used the following identity
$$
\bos{K} \left( \bos{\Theta}_{k}\cdot  \bos{\Theta}_{n-k}\right) = -\hbar 
 \left( \bos{\Theta}_{k}, \bos{\Theta}_{n-k}\right) + \bos{K} \bos{\Theta}_{k}\cdot  \bos{\Theta}_{n-k} 
 +  \bos{\Theta}_{k}\cdot\bos{K}  \bos{\Theta}_{n-k},
 $$
as well as the properties that $\bos{K}$ is a derivation of the BV bracket and commutes
with $m_{n-k+1}^\sharp$. Further using 
 $\bos{K} \bos{\Theta}_{1}=0$ and the commutativity of the product
$ \bos{\Theta}_{k}\cdot \bos{K} \bos{\Theta}_{n-k}
= \bos{K} \bos{\Theta}_{n-k}\cdot \bos{\Theta}_{k}$,
we have
$$
\eqalign{
\bos{K}\tilde\mathbb{M}_{n} = 
&
-\hbar\sum_{k=1}^{n-1} k(n-k)\left( \bos{\Theta}_{k},  \bos{\Theta}_{n-k}\right)
\cr
&
+2\sum_{k=2}^{n-1} k(n-k)\bos{K} \bos{\Theta}_{k}\cdot  \bos{\Theta}_{n-k}
-\sum_{k=2}^{n-2}(n-k)(n-k-1)
\left(\bos{K} \bos{\Theta}_{k},\La_{n-k}\right)
\cr
&
+\sum_{k=2}^{n-1}k(k-1)
\left( \bos{\Theta}_{n-k},\bos{K}\La_{k}\right)
-\sum_{k=2}^{n-1}(n-k+1)(n-k) m_{n-k+1}^\sharp\bos{K} \bos{\Theta}_{k}.
}
$$ 
Then we use the assumptions from $\bos{P}(n-1)$ that
for all $k=2,\ldots, n-1$
$$
\eqalign{
k(k-1)\bos{K}\La_{k}
=&
-\hbar k(k-1) \bos{\Theta}_{k}
\cr
&
+\sum_{\ell=1}^{k-1} \ell(k-\ell) \bos{\Theta}_{\ell}\cdot   \bos{\Theta}_{k-\ell}
-\sum_{\ell=2}^{k-1}\ell(\ell-1)
\left( \bos{\Theta}_{k-\ell},\La_{\ell}\right)
\cr
&
-\sum_{\ell=1}^{k-1}(k-\ell+1)(k-\ell) m_{k-\ell+1}^\sharp \bos{\Theta}_{\ell}
\cr
\bos{K}  \bos{\Theta}_{k}
=& 
 - \Fr{1}{2}  \sum_{\ell=1}^{k-1}\left( \bos{\Theta}_{\ell}, \bos{\Theta}_{k-\ell}\right).
}
$$
to have
\eqn\zzx{
\eqalign{
\bos{K}\tilde\mathbb{M}_{n} = 
&
-\hbar\sum_{k=1}^{n-1} k(n-k)\left( \bos{\Theta}_{k},  \bos{\Theta}_{n-k}\right)
-\hbar\sum_{k=2}^{n-1}
 k(k-1)
\left( \bos{\Theta}_{n-k}, \bos{\Theta}_{k}\right)
\cr
&
-\sum_{k=2}^{n-1}\sum_{\ell=1}^{k-1} k(n-k)
\left( \bos{\Theta}_{\ell}, \bos{\Theta}_{k-\ell}\right)\cdot  \bos{\Theta}_{n-k}
\cr
&
+\sum_{k=2}^{n-1}
\sum_{\ell=1}^{k-1} \ell(k-\ell)
\left( \bos{\Theta}_{n-k}, \bos{\Theta}_{\ell}\cdot   \bos{\Theta}_{k-\ell}\right)
\cr
&
+\Fr{1}{2}\sum_{k=2}^{n-2}\sum_{\ell=1}^{k-1}(n-k)(n-k-1)
\left(\left( \bos{\Theta}_{\ell}, \bos{\Theta}_{k-\ell}\right),\La_{n-k}\right)
\cr
&
-\sum_{k=2}^{n-1}
\sum_{\ell=2}^{k-1}\ell(\ell-1)
\left( \bos{\Theta}_{n-k},
\left( \bos{\Theta}_{k-\ell},\La_{\ell}\right)\right)
\cr
&
+\Fr{1}{2}\sum_{k=2}^{n-1}\sum_{\ell=1}^{k-1}
(n-k+1)(n-k) m_{n-k+1}^\sharp\left( \bos{\Theta}_{\ell}, \bos{\Theta}_{k-\ell}\right)
\cr
&
-\sum_{k=2}^{n-1}
\sum_{\ell=1}^{k-1}(k-\ell+1)(k-\ell) 
\left( \bos{\Theta}_{n-k},m_{k-\ell+1}^\sharp \bos{\Theta}_{\ell}\right).
}
}
Now we claim that 
(i) the $2$-nd and the $3$-rd lines of the right hand side of the above cancel with each others
due to the Leibniz law,
(ii) the $4$-th and the $5$-th lines of the right hand side of the above cancel with each others
due to the Jacobi law,
(iii) the $6$-th and the $7$-th lines of the right hand side of the above cancel with each others
due to $m_k$ being a derivation of the BV bracket.
\begin{itemize}
\item
To check the claim (i), rewrite the $3$rd line of the right hand side of \zzx\ 
$$
\eqalign{
\sum_{k=2}^{n-1}&
\sum_{\ell=1}^{k-1} \ell(k-\ell)
\left( \bos{\Theta}_{n-k}, \bos{\Theta}_{\ell}\cdot   \bos{\Theta}_{k-\ell}\right)
\cr
&=
\sum_{k=2}^{n-1}
\sum_{\ell=1}^{k-1} \ell(k-\ell)\left\{
\left( \bos{\Theta}_{n-k}, \bos{\Theta}_{\ell}\right)\cdot   \bos{\Theta}_{k-\ell}
+ \bos{\Theta}_{\ell}\cdot \left( \bos{\Theta}_{n-k},  \bos{\Theta}_{k-\ell}\right)
\right\}
\cr
&=
\sum_{k=2}^{n-1}
\sum_{\ell=1}^{k-1} \ell(k-\ell)\left\{
\left( \bos{\Theta}_{n-k}, \bos{\Theta}_{\ell}\right)\cdot   \bos{\Theta}_{k-\ell}
+ \left( \bos{\Theta}_{n-k},  \bos{\Theta}_{k-\ell}\right)\cdot \bos{\Theta}_{\ell}
\right\}
\cr
&=
\sum_{k=2}^{n-1}\sum_{\ell=1}^{k-1} k(n-k)
\left( \bos{\Theta}_{\ell}, \bos{\Theta}_{k-\ell}\right)\cdot  \bos{\Theta}_{n-k},
}
$$
where we applied the Leibniz law for the first equality, used the 
super-commutativity of the product in the second equality
and  did a re-summation for the last equality. By 
comparing with the $2$-nd line of the right hand side of \zzx, we have proved  the claim.

\item
To check the claim (ii), rewrite the $5$-th line of the right hand side of \zzx\ 
$$
\eqalign{
-\sum_{k=2}^{n-1}&
\sum_{\ell=2}^{k-1}\ell(\ell-1)
\left( \bos{\Theta}_{n-k},
\left( \bos{\Theta}_{k-\ell},\La_{\ell}\right)\right)
\cr
&=-\Fr{1}{2}\sum_{k=2}^{n-1}
\sum_{\ell=2}^{k-1}\ell(\ell-1)
\left(\left( \bos{\Theta}_{n-k},
 \bos{\Theta}_{k-\ell}\right),\La_{\ell}\right)
\cr
&=-\Fr{1}{2}\sum_{k=2}^{n-2}\sum_{\ell=1}^{k-1}(n-k)(n-k-1)
\left(\left( \bos{\Theta}_{\ell}, \bos{\Theta}_{k-\ell}\right),\La_{n-k}\right),
}
$$
where we applied the Jacobi law for the first equality and did a re-summation for the
second equality. By
comparing with the $4$-th line of the right hand side of \zzx, we have proved the claim.
\item
To check the claim (iii), rewrite the $7$-th line of the right hand side of \zzx\ 
$$
\eqalign{
-\sum_{k=2}^{n-1}
&
\sum_{\ell=1}^{k-1}(k-\ell+1)(k-\ell) 
\left( \bos{\Theta}_{n-k},m_{k-\ell+1}^\sharp \bos{\Theta}_{\ell}\right)
\cr
&= 
-\Fr{1}{2}\sum_{k=2}^{n-1}
\sum_{\ell=1}^{k-1}(k-\ell+1)(k-\ell) 
m_{k-\ell+1}^\sharp
\left( \bos{\Theta}_{n-k}, \bos{\Theta}_{\ell}\right)
\cr
&= 
-\Fr{1}{2}\sum_{k=2}^{n-1}
\sum_{\ell=1}^{k-1}(k-\ell+1)(k-\ell) 
m_{k-\ell+1}^\sharp
\left( \bos{\Theta}_{\ell}, \bos{\Theta}_{n-k}\right)
\cr
&
=-\Fr{1}{2}\sum_{k=2}^{n-1}\sum_{\ell=1}^{k-1}
(n-k+1)(n-k) m_{n-k+1}^\sharp\left( \bos{\Theta}_{\ell}, \bos{\Theta}_{k-\ell}\right)
}
$$
where we used $m_{k-\ell+1}$ being a derivation of the bracket  for the first equality,
used the commutativity of the BV bracket for the second equality and
did a re-summation for the
last equality. By
comparing with the $6$-th line of the right hand side of \zzx, we have proved the claim.
\end{itemize}
Thus we are left with
$$
\eqalign{
\bos{K}\tilde\mathbb{M}_{n} = 
&
-\hbar\sum_{k=1}^{n-1} k(n-k)\left( \bos{\Theta}_{k},  \bos{\Theta}_{n-k}\right)
-\hbar\sum_{k=2}^{n-1}
 k(k-1)
\left( \bos{\Theta}_{n-k}, \bos{\Theta}_{k}\right),
}
$$
which gives, after a re-summation,
$$
\bos{K}\tilde\mathbb{M}_{n} 
= - \hbar n(n-1)\Fr{1}{2} \sum_{k=1}^{n-1}\left( \bos{\Theta}_{k}, \bos{\Theta}_{n-k}\right),
$$
which proves the first claim;
$$
\bos{K}\mathbb{M}_{n} 
= - \Fr{\hbar}{2} \sum_{k=1}^{n-1}\left( \bos{\Theta}_{k}, \bos{\Theta}_{n-k}\right).
$$
For the second claim, we have
$$
\eqalign{
\rd_0 \mathbb{M}_n
&= \rd_0\left(\Fr{2(n-1)}{n(n-1)}\bos{\Theta}_1 \cdot \bos{\Theta}_{n-1}
-\Fr{2}{n(n-1)}m_2\bos{\Theta}_{n-1} -\Fr{(n-1)(n-2)}{n(n-1)}\big(\bos{\Theta}_{1}, \La_{n-1}\big)
\right)
\cr
&= \Fr{2}{n}\bos{\Theta}_{n-1}
-\Fr{2}{n(n-1)}t^\a\rd_\a\bos{\Theta}_{n-1} -\Fr{(n-1)(n-2)}{n(n-1)}\big(1, \La_{n-1}\big)
\cr
&=\Fr{2}{n}\left(\bos{\Theta}_{n-1}
-\Fr{1}{(n-1)}t^\a\rd_\a\bos{\Theta}_{n-1}\right),
}
$$
where we have used the assumptions that 
$\rd_0 \bos{\Theta}_{k}=\rd_0 \La_k=0$ for $2\geq k\geq n-1$
and $[\rd_0,m_k^\sharp]=0$ for $3\geq k \geq n-1$ for the first equality, the fact that
$\rd_0 \bos{\Theta}_1=1$ and $[\rd_0,m_2^\sharp]=t^\a\rd_\a$ for the second equality. 
Finally we conclude that $\rd_0 \mathbb{M}_n=0$ since 
$t^\a \rd_\a \bos{\Theta}_{n-1}= (n-1)\bos{\Theta}_{n-1}$. \qed


\subsection{Algebras of Quantum Correlation Functions}

The solution $\bos{\Theta}$ of  quantum master equation can be used to
determine generating function of all
quantum correlators by the formula
$$
e^{-\bos{\Theta}/\hbar}
=1 + \sum_{n=1}^\infty \Fr{1}{n!}\Fr{(-1)^n}{\hbar^{n}}\bos{\Theta}^n
=1+\sum_{n=1}^\infty \Fr{(-1)^n}{\hbar^n} \bos{\Omega}_n,
$$
where the sequence $\bos{\Omega}_1,\bos{\Omega}_2, \cdots$  is defined 
by matching the word-lengths in $t$,
such that $\bos{\Omega}_{n}$ generates $n$-point quantum correlators;
$$
\bos{\Omega}_n=\Fr{1}{n!}t^{\a_n}\cdots t^{\a_1}\bos{\pi}_{\a_1\cdots \a_n}
\quad\hbox{where}\quad 
\bos{\pi}_{\a_1\cdots \a_n}\in \sC[[\hbar]]^{|\a_1|+\cdots+|\a_n|}.
$$
From the decomposition $\bos{\Theta} =\sum_{n=1}^\infty \bos{\Theta}_n$ of $\bos{\Theta}$
by the word-length in $t_H$, we have the following recursive formula
\eqn\qcfa{
\eqalign{
\bos{\Omega}_1&=\bos{\Theta}_1,\cr
\bos{\Omega}_n &= (-\hbar)^{n-1}\bos{\Theta}_{n}
+\Fr{1}{n}\sum_{j=1}^{n-1} j  (-\hbar)^{j-1}\bos{\Theta}_{j}\cdot\bos{\Omega}_{n-j}.
}
}
Equivalently,
\eqn\qcfa{
\bos{\pi}_{\a_1\cdots \a_n}:= (-\hbar)^n
\rd_{\a_1}\cdots \rd_{\a_n}e^{-\bos{\Theta}/\hbar}\biggl|_{t=0}.
}
Note that
$$
\bos{\pi}_{\a_1\cdots \a_n}\big|_{\hbar=0}= O_{\a_1}\cdots O_{\a_n}
$$
The quantum descendant equation implies that
$\bos{K}\bos{\Omega}_{n}=0$ for all $n=1,2,\ldots$ since
it is equivalent to $\bos{K}\,e^{-\bos{\Theta}/\hbar}=0$. 
Thus 
$\bos{\pi}_{\a_1\cdots \a_n}$ is the canonical quantum correlator -quantization of
classical correlator $ O_{\a_1}\cdots O_{\a_n}$.  
We define the generating
functional $\mb{\CZ}(t_H)$ of all correlation functions 
by the formula
$$
\eqalign{
\mb{\CZ}(t_H)
&:= <1> + \sum_{n=1}^\infty \Fr{(-1)^n}{\hbar^{n}}\left<\bos{\Omega}_n\right>
\cr
&=<1> + \sum_{n=1}^\infty \Fr{1}{n!}\Fr{(-1)^n}{\hbar^{n}}
t^{\a_{n}}\cdots t^{\a_{1}}\left<\bos{\pi}_{\a_{1}\cdots\a_{n}}\right>
}
$$
so that an arbitrary $n$-point correlation function 
$\left<\bos{\pi}_{\a_1\cdots\a_n}\right>$
is obtained as follows;
$$
\left<\bos{\pi}_{\a_1\cdots\a_n}\right>
\equiv (-\hbar)^n \rd_{\a_1}\cdots \rd_{\a_n}\mb{\CZ}(t_H)\bigl|_{t=0}.
$$

The quantum identity $\rd_0 \bos{\Theta}_1=1$ and $\rd_0 \bos{\Theta}_n=0$ for $n\geq 2$
implies that
\begin{corollary}\label{qids}
$\rd_0 \bos{\Omega}_1=1$ and $\rd_0 \bos{\Omega}_n= \bos{\Omega}_{n-1}$ for $n\geq 2$.
\end{corollary}

\begin{proof}
We use  induction.
It is obvious $\rd_0 \bos{\Omega}_1=1$ since $\bos{\Omega}_1= \bos{\Theta}_1$.
From $\bos{\Omega}_2 =\Fr{1}{2!}\bos{\Theta}_1^2 -\hbar \bos{\Theta}_2$, we have
$\rd_0\bos{\Omega}_2 =\rd_0\bos{\Theta}_1\cdot \bos{\Theta}_1 -\hbar \rd_0\bos{\Theta}_2
= \bos{\Theta}_1= \bos{\Omega}_1$. Fix $n > 3$ and assume that
$\rd_0\bos{\Omega}_k= \bos{\Omega}_{k-1}$ for $2\leq k\leq n-1$.
From \qcfa, we have
$$
\bos{\Omega}_n = (-\hbar)^{n-1}\bos{\Theta}_{n}
+\Fr{1}{n}\sum_{j=1}^{n-1} j  (-\hbar)^{j-1}\bos{\Theta}_{j}\cdot\bos{\Omega}_{n-j}.
$$
Then
$$
\eqalign{
\rd_0\bos{\Omega}_n &=
\Fr{1}{n}\rd_0\bos{\Theta}_{1}\cdot\bos{\Omega}_{n-1}
+\Fr{1}{n}\sum_{j=1}^{n-1} j  (-\hbar)^{j-1}\bos{\Theta}_{j}\cdot\rd_0\bos{\Omega}_{n-j}
\cr
&=
\Fr{1}{n}\bos{\Omega}_{n-1}
+\Fr{1}{n}\sum_{j=1}^{n-1} j  (-\hbar)^{j-1}\bos{\Theta}_{j}\cdot\rd_0\bos{\Omega}_{n-j}
\cr
&=
\Fr{1}{n}\bos{\Omega}_{n-1}
+\Fr{n-1}{n}(-\hbar)^{n-2}\bos{\Theta}_{n-1}
+\Fr{1}{n}\sum_{j=1}^{n-2} j  (-\hbar)^{j-1}\bos{\Theta}_{j}\cdot\bos{\Omega}_{n-1-j}
\cr
&=
\Fr{1}{n}\bos{\Omega}_{n-1}
+\Fr{n-1}{n}\left((-\hbar)^{n-2}\bos{\Theta}_{n-1}
+\Fr{1}{n-1}\sum_{j=1}^{n-2} j  (-\hbar)^{j-1}\bos{\Theta}_{j}\cdot\bos{\Omega}_{n-1-j}
\right)
\cr
&= \bos{\Omega}_{n-1}.
}
$$
\qed
\end{proof}

From the quantum master equation we shall show the following:
\begin{lemma}
for every $n >1$, we have
$$
\bos{\Omega}_n = \bos{p}_n^\sharp \bos{\Theta}_1 + \bos{K}\bos{x}_n
$$
where $\bos{p}_2^\sharp=m_2^\sharp$, $\bos{x}_2=\La_2$, and
$$
\eqalign{
\bos{p}_n^\sharp &= 
(-\hbar)^{n-2}m_n^\sharp
+\Fr{1}{n(n-1)}\sum_{k=2}^{n-1}(-\hbar)^{k-2}k(k-1) m_{k}^\sharp \bos{p}_{n+1-k}^\sharp
,\cr
\bos{x}_n &= 
(-\hbar)^{n-2}\l_n
+\Fr{1}{n(n-1)}\sum_{k=2}^{n-1}(-\hbar)^{k-2}k(k-1)
\Big(
m_{k}^\sharp \bos{x}_{n+1-k}
+\La_k\cdot\bos{\Omega}_{n-k}\Big).
}
$$
\end{lemma}

\begin{proof}
Consider the decomposition of $\bos{\Theta}$ in terms of the word-length of $t_H$;
$$
\bos{\Theta}=\bos{\Theta}_1+ \bos{\Theta}_2+\bos{\Theta}_3+\cdots 
$$
where $\bos{\Theta}_n=\Fr{1}{n!}t^{\a_n}\cdots t^{\a_1}\bos{O}_{\a_1\cdots\a_n}$ is a 
homogeneous polynomial of degree $n$ in $t_H$.
It follows that,  for all $n=1,2,3,\ldots$,
$$
t^\a \rd_\a \bos{\Theta}_n = n \bos{\Theta}_n,\qquad
t^\b t^\a \rd_\a\rd_\b \bos{\Theta}_n =n(n-1) \bos{\Theta}_n.
$$
It is convenient to introduce the following notations 
$$
A= \sum_{n=2}^\infty n(n-1)m_n^\sharp, \qquad \La = \sum_{n=2}^\infty n(n-1)\La_n.
$$
Now the quantum master equation can be rewritten in the following form;
\eqn\qmeq{
-\hbar t^\b t^\a \rd_\a\rd_\b \bos{\Theta} + t^\a \rd_\a \bos{\Theta}\cdot t^\b \rd_\b \bos{\Theta} 
- A\bos{\Theta} = \bos{K}\La
+\left(\bos{\Theta}, \La\right), 
}
since the above equation, after decomposing  in terms of the word-length of $t_H$,
is equivalent to the following  infinite sequence of equations;
$$
\mathbb{E}_2=0,\quad\mathbb{E}_3=0,\quad\mathbb{E}_4=0, \quad\cdots
$$
where
$$
\eqalign{
\mathbb{E}_n=
&
-\hbar t^\b t^\a \rd_\a\rd_\b \bos{\Theta}_n 
+\sum_{k=1}^{n-1} \left(t^\a \rd_\a \bos{\Theta}_k\right)\left( t^\b \rd_\b \bos{\Theta}_{n-k}\right) 
- \sum_{k=2}^n k(k-1) m_k^\sharp\bos{\Theta}_{n-k+1} 
\cr
&
-\bos{K}\La_n
-\sum_{k=2}^{n-1}k(k-1)\left(\bos{\Theta}_{n-k}, \La_k\right)
\cr
=&
-\hbar n(n-1) \bos{\Theta}_n 
+\sum_{k=1}^{n-1} k(n-k) \bos{\Theta}_k \bos{\Theta}_{n-k} 
- \sum_{k=2}^n k(k-1)m_k^\sharp\bos{\Theta}_{n-k+1} 
\cr
&
-\bos{K}\l_n
-\sum_{k=2}^{n-1}k(k-1)\left(\bos{\Theta}_{n-k}, \La_k\right).
}
$$
Now the equation \qmeq\ is equivalent to the followings;
\eqn\qmeqa{
\left(\hbar^2 t^\b t^\a\rd_\a\rd_\b +\hbar A\right)e^{-\bos{\Theta}/\hbar}
=\bos{K}\left(\La e^{-\bos{\Theta}/\hbar}\right),
}
since the LHS of the above is
$$
LHS=\left(-\hbar t^\b t^\a \rd_\a\rd_\b \bos{\Theta} + t^\a \rd_\a \bos{\Theta}\cdot t^\b \rd_\b \bos{\Theta} 
- A\bos{\Theta}\right)e^{-\bos{\Theta}/\hbar}
$$
while its RHS is
$$
\eqalign{
RHS
=&
\left(\bos{K}\La+\left(\bos{\Theta}, \La\right)\right)e^{-\bos{\Theta}/\hbar}
-\La \cdot\bos{K}e^{-\bos{\Theta}/\hbar}
\cr
=&
\left(\bos{K}\La+\left(\bos{\Theta}, \La\right)\right)e^{-\bos{\Theta}/\hbar},
}
$$
where we have used the quantum descendant equation 
$\bos{K}\bos{\Theta} +\Fr{1}{2}(\bos{\Theta},\bos{\Theta})=0$, which is equivalent to
$\bos{K}e^{-\bos{\Theta}/\hbar}=0$.
Substituting $e^{-\bos{\Theta}/\hbar}$ in equation \qmeqa\ by the formula;
$$
e^{-\bos{\Theta}/\hbar} = 1 +\sum_{n=1}^\infty \Fr{(-1)^n}{\hbar^n} \bos{\Omega}_n,
$$
we obtain the following infinite sequence of equations, for $n=2,3,4,\cdots$,
\eqn\caseg{
\eqalign{
\bos{\Omega}_n 
=
& 
\Fr{1}{n(n-1)}\sum_{k=2}^{n}(-\hbar)^{k-2}k(k-1)m_k^\sharp \bos{\Omega}_{n-k+1}
\cr
&
+\bos{K} \left(\Fr{1}{n(n-1)}\sum_{k=2}^{n-1}(-\hbar)^{k-2}k(k-1)\l_k \cdot \bos{\Omega}_{n-k} 
+(-\hbar)^{n-2}\La_n\right).
}
}
\begin{enumerate}
\item
For $n=2$, we have
$$
\eqalign{
\bos{\Omega}_2 
&=m_2 \bos{\Omega}_1+\bos{K}\l_2. 
}
$$
The above is just the same with the lowest quantum master equation
$\hbar\bos{\Theta}_2 =\Fr{1}{2}\bos{\Theta}_1^2-m_2 \bos{\Theta}_1-\bos{K}\l_2$,
since $\bos{\Omega}_1= \bos{\Theta}_1$ and $\bos{\Omega}_2 
= \Fr{1}{2}\bos{\Theta}_1^2 -\hbar\bos{\Theta}_2$.
Thus
\eqn\casew{
\eqalign{
\bos{\Omega}_1 &= \bos{\Theta}_1
,\cr
\bos{\Omega}_2 
&=\bos{p}_2^\sharp \bos{\Theta}_1+\bos{K}\bos{x}_2,
}
}
where  $\bos{p}^\sharp_2=m_2^\sharp$ and $\bos{x}_2=\La_2$.

\item
For $n=3$ we have
$$
\bos{\Omega}_3= \Fr{1}{3} m_2^\sharp \bos{\Omega}_2 -\hbar m_3^\sharp \bos{\Omega}_1
+ \bos{K}\left(\La_2 \cdot\bos{\Omega}_1 -\hbar \La_3\right)
$$
Using \casew, we conclude that
$$
\bos{\Omega}_3= \bos{p}_3^\sharp \bos{\Theta}_1 
+ \bos{K}\bos{x}_3,
$$
where
$$
\eqalign{
\bos{p}_3^\sharp &=\Fr{1}{3} m_2^\sharp \bos{p}_2^\sharp -\hbar m_3^\sharp
,\cr
\bos{x}_3&=\Fr{1}{3}m_2^\sharp \bos{x}_2 +\La_2\cdot \bos{\Omega}_1 -\hbar \La_3
}
$$
\item
Iterating the above procedure up to some $n > 3$ assume that,
for all $k=2,3,\ldots, n-1$,
$$
\bos{\Omega}_k =\bos{p}_k^\sharp \bos{\Theta}_1 +\bos{K} \bos{x}_k,
$$
where
$$
\eqalign{
\bos{p}_k^\sharp 
&= (-\hbar)^{k-2} m_k^\sharp 
+\Fr{1}{k(k-1)}\sum_{\ell=2}^{k-1} (-\hbar)^{\ell-2}\ell(\ell-1)m_\ell^\sharp \bos{p}_{n+1-\ell}^\sharp
,\cr
\bos{x}_k &= (-\hbar)^{k-2}\La_k +\Fr{1}{n(n-1)}\sum_{\ell=2}^{k-1}\ell(\ell-1)
\left( m_\ell^\sharp \bos{x}_{k+1-\ell} +\La_\ell \cdot \bos{\Omega}_{k-\ell}\right).
}
$$

Substituting above to \caseg\ we immediately obtain that
$\bos{\Omega}_n = \bos{p}_n^\sharp \bos{\Theta}_1 +\bos{K} \bos{x}_n$
as was claimed. 
\end{enumerate}
By induction the formula is true for every $n >1$.\qed
\end{proof}
It follows that $\left<\bos{\Omega}_1\right> 
= \left<\bos{\Theta}_1\right>$ and
$\left<\bos{\Omega}_n\right> = \bos{p}_n^\sharp\left<\bos{\Theta}_1\right>$ for $n=2,3,\cdots$,
more explicitly
$$
\eqalign{
\left<\bos{\Omega}_1\right>&=t^\g\left<\bos{O}_\g\right>
,\cr
\left<\bos{\Omega}_n\right> &=\Fr{1}{n!}t^{\a_n}\cdots t^{\a_1} \bos{p}_{\a_1\cdots\a_n}{}^\g
\left<\bos{O}_\g\right>.
}
$$
Combining with the corollary \ref{qids} we obtain that
\begin{corollary}\label{qidd}
$$
\left[\rd_0,\bos{p}_2^\sharp\right]=t^\a\rd_\a,\qquad
\left[\rd_0, \bos{p}_n^\sharp\right]= \bos{p}_{n-1}^\sharp
\hbox{ for } n\geq 3,
$$
or, equivalently $\bos{p}_{0\a_1\cdots\a_n}{}^\g=\bos{p}_{\a_1\cdots\a_n}{}^\g$
for $n\geq 2$.
\end{corollary}

\begin{example}
The first few quantum correlators are 
$$
\eqalign{
\bos{\Omega}_2 &=\Fr{1}{2!}\bos{\Theta}_1^2 -\hbar \bos{\Theta}_2,\cr
\bos{\Omega}_3 &=\Fr{1}{3!}\bos{\Theta}_1^3 -\hbar \bos{\Theta}_1\bos{\Theta}_2+\hbar^2\bos{\Theta}_3,
\cr
\bos{\Omega}_4 &=\Fr{1}{4!}\bos{\Theta}_1^4 -\Fr{\hbar}{2}\bos{\Theta}_1^2\bos{\Theta}_2
+\hbar^2\left(\bos{\Theta}_1\bos{\Theta}_3+\Fr{1}{2}\bos{\Theta}_2^2\right)
-\hbar^3 \bos{\Theta}_4
}
$$
or, in component
$$
\eqalign{
\bos{\pi}_\a:=&\bos{O}_\a
\cr
\bos{\pi}_{\a_{1}\a_{2}}:=&\bos{O}_{\a_{1}} \bos{O}_{\a_{2}} -\hbar \bos{O}_{\a_{1}\a_{2}}
,\cr
\bos{\pi}_{\a_{1}\a_{2}\a_{3}}=&
 \bos{O}_{\a_{1}} \bos{O}_{\a_{2}} \bos{O}_{\a_{3}}
-\hbar  \bos{O}_{\a_{1}\a_{2}}  \bos{O}_{\a_{3}}
-\hbar  \bos{O}_{\a_{1}} \bos{O}_{\a_{2}\a_{3}}
-\hbar (-1)^{|\a_{1}||\a_{2}|} \bos{O}_{\a_{2}}  \bos{O}_{\a_{1}\a_{3}}
\cr
&
+\hbar^2\bos{O}_{\a_{1}\a_{2}\a_{3}},
,\cr
\bos{\pi}_{\a_{1}\a_{2}\a_{3}\a_{4}}
=&
 \bos{O}_{\a_{1}} \bos{O}_{\a_{2}}  \bos{O}_{\a_{3}}  \bos{O}_{\a_{4}}
 \cr
 &
 -\hbar \bos{O}_{\a_{1}\a_{2}}  \bos{O}_{\a_{3}}  \bos{O}_{\a_{4}}
 -\hbar(-1)^{|\a_{1}||\a_{2}|}  \bos{O}_{\a_{2}}  \bos{O}_{\a_{1}\a_{3}}  \bos{O}_{\a_{4}}
 \cr
 &
 -\hbar  \bos{O}_{\a_{1}} \bos{O}_{\a_{2}\a_{3}}  \bos{O}_{\a_{4}}
 -\hbar (-1)^{|\a_{2}||\a_{3}|}  \bos{O}_{\a_{1}}\bos{O}_{\a_{3}}  \bos{O}_{\a_{2}\a_{4}}
\cr
&
-\hbar  \bos{O}_{\a_{1}} \bos{O}_{\a_{2}}  \bos{O}_{\a_{3}\a_{4}}
  -\hbar(-1)^{|\a_{1}|(|\a_{2}|+|\a_{3}|)} \bos{O}_{\a_{2}}  \bos{O}_{\a_{3}}  \bos{O}_{\a_{1}\a_{4}} 
   \cr
&
+\hbar^{2}  \bos{O}_{\a_{1}\a_{2}\a_{3}}  \bos{O}_{\a_{4}}
+\hbar^{2}(-1)^{|\a_{1}||\a_{2}|}  \bos{O}_{\a_{2}}  \bos{O}_{\a_{1}\a_{3}\a_{4}}
\cr
&
+\hbar^{2}  \bos{O}_{\a_{1}\a_{2}}  \bos{O}_{\a_{3}\a_{4}}
+\hbar^{2} (-1)^{|\a_{2}||\a_{3}|} \bos{O}_{\a_{1}\a_{3}}  \bos{O}_{\a_{2}\a_{4}}
+\hbar^{2} (-1)^{|\a_{1}|(|\a_{2}|+|\a_{3}|)}  \bos{O}_{\a_{2}\a_{3}}  \bos{O}_{\a_{1}\a_{4}}
\cr
&
+\hbar^2 \bos{O}_{\a_{1}}\bos{O}_{\a_{2}\a_{3}\a_{4}}
+\hbar^{2}(-1)^{(|\a_{1}|+|\a_{2}|)|\a_{3}|} \bos{O}_{\a_{3}}  \bos{O}_{\a_{1}\a_{2}\a_{4}}
\cr
&
-\hbar^3\bos{O}_{\a_{1}\a_{2}\a_{3}\a_{4}}
}
$$
Then
$$
\eqalign{
\left<\bos{\Omega}_2\right> =&\left<\bos{\Theta}_1\right>,\cr
\left<\bos{\Omega}_2\right> =& m_2^\sharp \left<\bos{\Theta}_1\right>,\cr
\left<\bos{\Omega}_3\right>
=&\left(\Fr{1}{3}m_2^{\sharp} m_2^{\sharp}  -\hbar m_3^\sharp\right)\left<\bos{\Theta}_1\right>
,\cr
\left<\bos{\Omega}_4\right> 
=& \left(\Fr{1}{18}m_2^{\sharp} m_2^{\sharp} m_2^{\sharp} -\Fr{\hbar}{6} m_2^\sharp m_3^\sharp
 -\Fr{\hbar}{2}m_3^\sharp m_2^\sharp +\hbar^2 m_4^\sharp  \right)
\left<\bos{\Theta}_1\right>
.
}
$$
The above examples illustrate some nature of quantum correlations.
\end{example}

Now the generating function $\mb{\CZ}(t_H)$ of all correlation functions
can be expressed as
$$
\mb{\CZ}(t_H) =\left<1\right> -\Fr{1}{\hbar} \mb{T}^\g(t_H)\left<\bos{O}_{\g}\right>
$$
where 
$$
\mb{T}^\g:=t^{\g} -\Fr{1}{2\hbar}t^\b t^\a m_{\a\b}{}^\g
+\sum_{n=3}^\infty \Fr{1}{n!}\Fr{(-1)^{n-1}}{\hbar^{n-1}}
t^{\a_{n}}\cdots t^{\a_{1}}\bos{p}_{\a_{1}\cdots\a_{n}}{}^{\g} \in \Bbbk\left[\!\left[t_H,\hbar^{-1}\right]\!\right]
$$
From the corollary \ref{qidd}, we have
$$
\rd_0 \mb{T}^\g = \d_0{}^\g -\Fr{1}{\hbar}\mb{T}^\g. 
$$
A detailed study of properties of $\{\mb{T}^\g\}$ is a subject of the next paper.

\subsection{Quantum versus flat coordinates}

An immediate consequence of the  theorem  $4.1.$ is that
that the classical limit $\Theta$ of $\bos{\Theta}$ is a solution
to the DGLA $\big(\sC, Q.\;\cdot\;\big)$ of very special kind.

\begin{corollary}
There exists a solution 
to  the classical descendant equation
\eqn\sdeform{
\eqalign{
Q{\Theta} +\Fr{1}{2}\big({\Theta},{\Theta}\big)=0,
\qquad
{\Theta} = t^\a {O}_\a 
+ \sum_{n=2}^\infty \Fr{1}{n!}t^{\a_n}\cdots t^{\a_1} {O}_{\a_1\cdots\a_n}
 \in \big(\Bbbk[[t_H]]\otimes\sC\big)^0
}
}
such that
\begin{enumerate}
\item (versality) the set of cohomology classes $[O_\a]$ 
form a basis of cohomology $H$ of the classical complex $(\sC, Q)$

\item (quantum coordinates) $\Theta$ is the classical limit of the solution to quantum master equation

\item (quantum identity) $\rd_0 \Theta=1$.
\end{enumerate}
\end{corollary}

We recall some standard relations between deformation theory, DGLA and  
$L_\infty$-algebra (see \cite{Kontsevich} and references therein for details). 
$L_\infty$-algebra is  natural homotopy generalization
of DGLA in the following sense. A morphism of DGLA is naturally a cochain map which is also a (graded)
Lie algebra map. However, a cochain map homotopic to a morphism of DGLA is not a Lie algebra map in
general but it can be viewed as a morphism of $L_\infty$-algebra. Thus it is natural to replace
the category of DGLA to more flexible category of $L_\infty$-algebra, which localize well under homotopy.
Also, on cohomology of DGLA, there is a structure of minimal $L_\infty$-algebra (an $L_\infty$-algebra with
zero-differential), which is quasi-isomorphic to the DGLA at the chain level as $L_\infty$-algebra. 
Furthermore such minimal $L_\infty$-structure on cohomology is the obstruction to have a
versal solution to Maurer-Cartan equation the DGLA. 
A DGLA is called formal if the minimal $L_\infty$-algebra on its cohomology is a graded
Lie algebra, and a formal DGLA has an associated smooth moduli space if and only if 
the graded Lie algebra on its cohomology is Abelian, i.e., the graded Lie bracket vanishes on $H$.
Then  a versal solution to the Maurer-Cartan equation is nothing but a quasi-isomorphism
from $H$ to the DGLA as $L_\infty$-algebra.\cite{Kontsevich}

The versal solution we have is an $L_\infty$-quasi-isomorphism of very special kind 
since not every versal solution  
of \sdeform\ arises as the classical limit of solution of quantum master equation.
Hence we say that an anomaly-free BV QFT has its natural family parametrized by
a smooth moduli space $\CM$, a formal super-manifold, in quantum coordinates.

Now we are going to demonstrate that
the notion of quantum coordinates is a natural generalization 
of that of flat or special coordinates on moduli spaces of topological strings 
in the context of Witten-Dijkgraaf-Verlinde-Verlinde (WDDV)
equation \cite{W,DVV} as well as the mirror map 
of Candelas-de la Ossa-Green-Parkes for Calabi-Yau quintic 
\cite{COGP,Witten}. 
For the mathematical sides, both the pioneering work of K.\ Saito on his flat structure
on moduli space of universal unfolding of simple singularities \cite{Saito}
and the flat coordinates in certain differential
BV algebra due to Barannikov-Kontsevich \cite{BK} are also examples of
quantum coordinates.  Those correspondences shall be discussed briefly in
three examples of this subsection.
We  remark that our result is {\it not} specific
to topological strings or $2d$-dimensional topological conformal theory.
It should be also noted that our general package does not include flat metric
over the moduli space $\CM$ in it. In the following three examples 
one may easily supply such a metric as additional data.

Consider a BV QFT algebra $\big(\sC[[\hbar]],\bos{K},\;\cdot\;\big)$ such that
$\bos{K}= Q +\hbar K^{(1)}$. Then the quadruple $(\sC, Q, \Delta,\;\cdot\;)$,
where $\Delta:=-K^{(1)}$, is a differential BV algebra. Conversely, let 
$(\sC, Q, \Delta,\;\cdot\;)$ be a differential BV algebra, then 
$\big(\sC[[\hbar]],\bos{K}=-\hbar\Delta +Q,\;\cdot\;\big)$ is a BV QFT algebra.
 We say such a BV QFT algebra
{\it semi-classical} if there is a $\Bbbk$-linear map $f$ on $H$ into $\sC$ such that
(i) $f$ is cochain map on $(H,0)$ into $(\sC, Q)$ which induces the identity map
on the cohomology $H$, (ii) $f(e)=1$ and (iii) $\Delta f =0$.  Then $\bos{K} f =0$,
i.e., $\bos{f}=f$, and $\bos{\k}=0$ on $H$ identically.
It follows that $\bos{\Theta}_1= t^a f(e_\a) =\Theta_1$, and the quantum master equation
is decomposed into the following set of equations
\eqn\semicl{
\eqalign{
&\left\{
\eqalign{
0=&Q{\Theta}_1
,\cr
0=&\Fr{1}{2}{\Theta}_1\cdot {\Theta}_1- m^\sharp_2 {\Theta}_1 - Q \La_2
,\cr
%
\vdots\;&
\cr
0=& 
\sum_{k=1}^{n-1}\Fr{ k(n-k)}{n(n-1)} {\Theta}_{k}\cdot  {\Theta}_{n-k}
-\sum_{k=2}^{n-1}\Fr{k(k-1)}{n(n-1)}\left(  m^\sharp_{k} {\Theta}_{n-k+1}
+ \big({\Theta}_{n-k},\La_{k}\big)\right)
\cr
&
- m^\sharp_n {\Theta}_1-Q\La_k
,\cr
\vdots\;&
}
\right.
\cr
&\left\{
\eqalign{
\Delta \Theta_1&=0,\cr
\Theta_n&=\Delta \La_n \hbox{ for } n\geq 2.
}\right.
}
}
The above shall be called semi-classical master equation.
Then the quantum descendant equation also has decomposition as follows
$$
\eqalign{
\Delta \Theta=0,\cr
Q\Theta +\Fr{1}{2}\big(\Theta,\Theta\big)=0,
}
$$
which shall be called semi-classical descendant equation.

Hence we have (see also \cite{Park2,Terilla})
\begin{corollary}\label{flatco}
Let $\big(\sC, Q, \Delta, \;\cdot\;\big)$ be a differential BV algebra with
associated BV bracket $(\hbox{ },\hbox{ })$ such that every $Q$-cohomology
class has a representative in $\Ker \Delta$. Then
there exists a solution $\Theta$
to  the MC equation of the DGLA $\big(\sC, Q, (\hbox{ },\hbox{ })\big)$;
$$
\eqalign{
Q{\Theta} +\Fr{1}{2}\big({\Theta},{\Theta}\big)=0,
\qquad
{\Theta} = t^\a {O}_\a 
+ \sum_{n=2}^\infty \Fr{1}{n!}t^{\a_n}\cdots t^{\a_1} {O}_{\a_1\cdots\a_n}
 \in \big(\Bbbk[[t_H]]\otimes\sC\big)^0
}
$$
such that
\begin{enumerate}
\item (versality) the set of cohomology classes $[O_\a]$ 
form a basis of cohomology $H$ of the complex $(\sC, Q)$

\item (flat coordinates) $O_\a \in \Ker \Delta$ and $O_{\a_1\cdots\a_n}\in \Im \Delta$ for $n\geq 2$

\item (flat identity) $\rd_0 \Theta=1$, where $\rd_0$ is the coordinate vector field corresponding to $[1]\in H^0$.
\end{enumerate}
\end{corollary}
In the above we've changed the adjective quantum to flat due to the
following famous examples.

\begin{example}   We say a differential BV-algebra $(\sC, \Delta, Q, \cdot)$ has the $\Delta
  Q$-{\it property} if
$$
\left(\Ker Q \cap \Ker \Delta\right) \cap \left(\hbox{Im }
 \Delta \oplus \hbox{Im } Q\right) = \hbox{Im }\Delta Q
= \hbox{Im }Q\Delta.
$$
The corresponding BV QFT algebra is semi-classical since 
the $\Delta Q$-{\it property} implies that every $Q$-cohomology class has a
representative in $\Ker \Delta$.  Then the corollary \ref{flatco} is exactly
the lemma $6.1$ of  Barannikov-Kontsevich in \cite{BK}.
The standard example of  such a differential BV algebra 
that controlling (extended)-deformation of complex structures of Calabi-Yau manifold,
corresponding to the extended moduli space of topological string B model \cite{Witten,Park0}. 
The  $\Delta Q$-property is a direct consequence the $\rd\bar\rd$-lemma of 
K\"{a}her manifold in \cite{DGMS}.
\end{example}

A semi-classical BV QFT can be also constructed from a differential BV algebra without the
$\Delta Q$-{\it property}. 

\begin{example}
  Let $\sC=\C[x^1,\ldots, x^m, \eta_1,\ldots, \eta_m]$ be a
  super-commutative polynomial algebra with free associative product subject to
  the super-commutative relations 
  $$x^i \cdot x^j= x^j\cdot x^i,\quad x^i\cdot\eta_j=\eta_j\cdot x^i,
  \quad \eta_i\cdot \eta_j=-\eta_j\cdot \eta_i
  $$
  Assign ghost number $0$ to $\{x^i\}$ and
  $-1$ to $\{\eta_i\}$. Then $\sC = \sC^{-m}\oplus\cdots\oplus
  \sC^{-1} \oplus \sC^0$. Note that $\sC^0=\Bbbk[x^1,\ldots, x^m]$.
  Define 
  $$
  \Delta:= \Fr{\rd^2}{\rd x^i \eta_i}:
  \sC^{k}\rightarrow \sC^{k +1}.
  $$
   It is obvious that
  $\Delta^2=0$, and the triple $(\sC, \Delta, \;\cdot\;)$ is a BV algebra
  over $\Bbbk$. We also note that $\sC^0 \in \hbox{Im }\Delta$.  To
  see this, it suffices to consider an arbitrary monomial
  $$(x^1)^{N_1}\cdots (x^m)^{N_m} \in \sC^0=\Bbbk[x^1,\ldots, x^m]$$ and
  observe that, for instance,
$$ (x^1)^{N_1}\cdots (x^m)^{N_m}
=\Fr{1}{(N_1+1)} \Delta\left(\eta_1\cdot (x^1)^{N_1+1}\cdots
  (x^m)^{N_m}\right).
$$
For any $S \in \sC^0$, we always have $\Delta S=(S,S)=0$. Fix $S$ and
define $Q=(S,\hbox{ })$, then  the quadruple $(\sC, \Delta, Q, \cdot)$ is a dBV algebra.
Denote by $H$ the cohomology of the complex $(\sC, Q)$.
Then
$$
H^0= \Bbbk[x^1,\ldots,x^m]\biggl/\left<\Fr{\rd S}{\rd x^1},\cdots, \Fr{\rd S}{\rd x^m}\right>,
$$
since $\sC^0\subset \hbox{Ker }Q$, and any element $R \in \sC^{-1}$
with the ghost number $-1$ is in the form $R =R^i \eta_i$, where
$\{R_i\}$ is a set of $m$ elements in $\sC^0$, such that $Q R = R^i
\Fr{\rd S}{\rd x^i}$.
Now we assume that $S$ is a polynomial (in
$x's$) with isolated singularities, so that the cohomology $H$ of the
complex $(\sC, Q)$ is concentrated in the ghost number zero part,
i.e., $H=H^0$. Then any representative of $H$ belongs to $\hbox{Ker
}\Delta$, since $\sC^0$ itself belongs to $\hbox{Ker }\Delta$.  Thus
the corresponding BV QFT  algebra is obviously semi-classical.  We already know that
$\sC^0 \subset \Ker Q \cap \hbox{Im} \Delta$, so that $\Delta
Q$-property would imply that $\sC^0 \subset \hbox{Im }Q\Delta$ and,
in particular, $H^0=0$, which is not generally true. The corollary \ref{flatco},
then, should be attributed to K.\ Saito \cite{Saito} and, independently, 
to Dijkgraaf-Verlinde-Verlinde \cite{DVV}.

\end{example}

\begin{remark}
The above example could be regarded as the simplest example of BV QFT
- a class of (0+0)-dimensional quantum field theories  without gauge symmetry,
where the polynomial $S$ is  classical action. In the next example
we will present a class of (0+0)-dimensional quantum field theories  with
Abelian gauge symmetry.
\end{remark}

\begin{example}
Let $\sC_{cl}=\C[z^\m |z^\bullet_\m]$, $\m=0,1,2,\cdots,n+2$, 
 a super-commutative polynomial algebra with free associative product subject to
  the super-commutative relations $z^\m \cdot z^\n= z^\m\cdot z^\n$,
  $z^\m\cdot z^\bullet_\n=z^\bullet_\n\cdot z^\m$ and $z^\bullet_\m\cdot
  z^\bullet_\n=-z^\bullet_\n\cdot z^\bullet_\m$. Assign ghost number $0$ to $\{z^\m\}$ and
  $-1$ to $\{z^\bullet_\n\}$. Then $\sC _{cl}= \sC_{cl}^{-n-2}\oplus\cdots\oplus
  \sC_{cl}^{-1} \oplus \sC_{cl}^0$. Note that $\sC^0=\C[z^\m]$.
Define 
$$
\Delta_{cl}:= \Fr{\rd^2}{\rd\! z^\m \rd\! z^\bullet_\m} :\sC_{cl}^k\rightarrow \sC_{cl}^{k+1},
$$
which satisfies $\Delta_{cl}^2=0$ such that $\big(\sC_{cl}, \Delta_{cl},\;\cdot\;\big)$
is a BV algebra over $\C$. The associated BV bracket $(\hbox{ },\hbox{ })_{cl}$ satisfies
that
$$
(z^\m, z^\n)_{cl}=0,\qquad (z^\n, z^\bullet_\m)_{cl}=-(z^\bullet_\m, z^\n)_{cl}=\d_\m{}^\n,
\qquad (z^\bullet_\m, z^\bullet_\n)_{cl}=0.
$$
We have $\sC_{cl}^0 \subset \Im \Delta_{cl}$ and $(\hbox{ },\hbox{ })_{cl}=0$ on $\sC^0_{cl}$.

Let 
$$
S(z)_{cl}= p\cdot G(x) \in \sC^0_{cl}
$$
where we denote $p=z^0$ and $x^i= z^i$ for $i=1,2,\cdots,n+2$ and $G(x)$ is 
a generic homogeneous polynomials in 
$\{x^i\}$ of degree $n+2$.
Let
$$
Q_{cl}:=(S_{cl},\hbox{ })_{cl}=G(x)\Fr{\rd}{\rd p^\bullet} 
+ p\left(\Fr{\rd G(x)}{\rd x^i}\right)\Fr{\rd}{\rd x^\bullet_i}.
$$
It follows that $Q_{cl}\Delta_{cl} +\Delta_{cl} Q_{cl}=Q_{cl}^2=0$ since  it is trivial that
$$
\eqalign{
\Delta S_{cl} =0,\cr
\big(S_{cl}, S_{cl}\big)_{cl}=0.
}
$$
Thus we have  constructed a BV QFT algebra $\big(\sC[[\hbar]], \bos{K}_{cl},\;\cdot\;\big)$
 with the descendant algebra 
$\big(\sC[[\hbar]], \bos{K}_{cl},(\hbox{ },\hbox{ })\big)$, where
$$
\bos{K}_{cl}:=-\hbar\Delta_{cl} + Q_{cl}.
$$
Let $H_{cl}$ denote cohomology of the cochain complex $\big(\sC_{cl}, Q_{cl}\big)$.
We note  that $\sC^0_{cl} \in \Ker Q_{cl}$.

There are two differences between the present case with the previous example;
(i) $H_{cl}^0$ is degenerated  and (ii) $H_{cl}^{-1}$ is
non-empty. Actually the property (i) is a consequence of (ii). We claim that
$H_{cl}=H_{cl}^{-1}\oplus H_{cl}^0$ and the BV QFT algebra is
semi-classical. The non-triviality of $H^{-1}$ corresponds to an obvious symmetry of
$S_{cl}$;
$$
\left(x^i\Fr{\rd}{\rd x^i} - (n+2) p\Fr{\rd}{\rd p}\right) S_{cl}=0,
$$
since $S_{cl}$ is a weighted homogeneous polynomial with degree $0$;
assign weight $1$ to $\{x^i\}$ and $-n-2$ to $x^0$. Or, equivalently $S_{cl}$
is invariant under the following $\C^*$-action on $\C^{n+3}$;
\eqn\cstar{
\eqalign{
x^i&\rightarrow \vr x^i,
\cr
p&\rightarrow \vr^{-n-2} p,
}
}
where $\vr\in \C^*$.
Let 
$$
R= x^i x^\bullet_i - (n+2) p p^\bullet, 
$$
Then $R=\in \sC^{-1}$ and
$Q_{cl} R =0$ 
since
$$
Q_{cl} R \equiv \big(S_{cl},R\big)_{cl}=-\big(R,S_{cl}\big)_{cl}= \left(x^i\Fr{\rd}{\rd x^i} 
- (n+2) p\Fr{\rd}{\rd p}\right) S_{cl},
$$
and $R$ can not be $Q_{cl}$-exact simply by the degree reason. Thus the 
$Q_{cl}$-cohomology class $[R]_{cl}$ of $R$ is a non-trivial element in $H_{cl}^{-1}$.
We claim that $H_{cl}^{-1}$ is generated by $R$ as a left $\sC_{cl}^0$-module.
Before we proceed further, here are some physics terminology:
\begin{itemize}
\item $\{z^\m\}$: classical fields.
\item $\{z^\bullet_\m\}$: anti-classical fields.
\item $S(z^\m)_{cl}$: classical action.
\item $R$: classical gauge symmetry vector.
\end{itemize}

We also note that $\Delta_{cl} R =0$ such that $\bos{K}_{cl}R=0$.
Thus we may say the classical symmetry vector is anomaly-free.
We also
note that the classical equation of motion, $\d S_{cl}/\d z^\m=0$, is
$$
\eqalign{
G(x)&=0,\cr
p \Fr{\rd G(x)}{\rd x^i}&=0, \quad i=1,2,\cdots,n+2.
}
$$
We may call the solution space of the above modulo the classical gauge symmetry
the space of classical observer. 
Assuming that $p\neq 0$, we must have $\Fr{\rd G(x)}{\rd x^i}=0$ for all $i$ to solve the
classical equation of motion.
Then $x^1=x^2=\cdots=x^{n+2}=0$ since $G$ is generic.
If $p=0$, then the solution space of classical equation is the zero set of homogeneous
polynomial $G(x^i)$ of degree $n+2$.
Then the space of solutions of classical equation of motion modulo the classical symmetry
is a $n$-dimensional Calabi-Yau hypersurface $X$ of $\mathbb{CP}^{n+1}$.
In general the space of classical observer can be viewed as the solution space of
classical equation of motion in the GIT quotient $\C^{n+3}//\C^*$, which depends on
choice of polarization imposing either $p\neq 0$ or $p=0$ -- see section $4$ in \cite{Witten1}.

Now we kill the classical gauge symmetry as follows. Introduce the dual basis $c$ of $H_{cl}^{-1}$
with ghost number $1$. Let $c^\bullet$ denote the corresponding basis of $H_{cl}^*[-2]$
with ghost number $-2$.
Let $\sC=\C[z^\m, c|z^\bullet_\m, c^\bullet]$
and
$$
\eqalign{
\Delta &:= \Delta_{cl}
-\Fr{\rd^2}{\rd c \rd c^\bullet}
,\cr
 S &:= S_{cl} + c R= p\cdot G(x) + c x^i x^\bullet_i -(n+2)c p p^\bullet.
 }
 $$
 Then
$$
\eqalign{
\Delta S =0,\cr
(S,S)=0.\cr
}
$$
Let $Q=(S,\hbox{ })$, then
we have another BV QFT algebra $\big(\sC[[\hbar]], \bos{K}=-\hbar \Delta + Q,\;\cdot\;\big)$
with the descendant algebra $\big(\sC[[\hbar]], \bos{K},(\hbox{ },\hbox{ })\big)$
where $(\hbox{ },\hbox{ })\big|_{\sC_{cl}}=(\hbox{ },\hbox{ })_{cl}$ and
$$
(c,c^\bullet)=-(c^\bullet, c)=1, \qquad
(c^\bullet, c^\bullet)=(c,c)=(c, z^\m)=(c,z^\bullet_\m)=(c^\bullet, z^\m)=(c^\bullet,z^\bullet_\m)=0,
$$
such that
$$
\eqalign{
Q =&(S_{cl}+c R, \hbox{ })
\cr
=&Q_{cl} + c(R, \hbox{ }) - R \Fr{\rd}{\rd c^\bullet}
\cr
=&G(x^j)\Fr{\rd}{\rd p^\bullet} + p\left(\Fr{\rd G(x)}{\rd x^i}\right)\Fr{\rd}{\rd x^\bullet_i}
-R \Fr{\rd}{\rd c^\bullet}
\cr
&+c\left(x^i\Fr{\rd}{\rd x^i} - (n+2) p\Fr{\rd}{\rd p}\right)
-c\left(x^\bullet_i\Fr{\rd}{\rd x^\bullet_i} - (n+2) p^\bullet\Fr{\rd}{\rd p^\bullet}\right)
}
$$ 
In particular $Q c^\bullet = R$ and, hence, we just have killed the classical symmetry vector.
Before we proceed further, here are some more physics terminology:
\begin{itemize}
\item $c$: Faddev-Popov ghost field.
\item $c^\bullet$: anti-field of Faddev-Popov ghost field.
\item $S$: BV quantum master action which is semi-classical.
\item $\d_{BRST}:=Q\big|_{z^\bullet_\m=c^\bullet=0}=
c\left(x^i\Fr{\rd}{\rd x^i} - (n+2) p\Fr{\rd}{\rd p}\right)$: the  BRST operator
which corresponds to the Euler vector field associated with the $\C^*$-action.
\end{itemize}

Now the $Q$-cohomology $H$ is concentrate on $H^0$ and
$$
H^0 =\C[p, x^i]^{inv}\biggl/\left<G(x), p\Fr{\rd G(x)}{\rd x^i}\right>^{inv},
$$
where the superscript {\it inv} means the invariant part under the $\C^*$-action \cstar.\foot{
The $\Ker Q$ in $\sC^0$ is  in the form $f + c g^\m z^\bullet_\m$ for any $f\in \C[p, x^i]^{inv}$
and $g\in\C[p, x^i]$, since $cc=0$. On the other hand it can be shown that any expression
$c g^\m z^\bullet_\m$ belongs to $\Im Q$.
}
We note that any  element in $\C[p, x^i]^{inv}$ is a $\C$-linear combinations of
monomials in the form $ p^{k}M(x)_{k(n+2)}$, $k=0,1,2,\cdots$,
where $M(x)_{k(n+2)}$ are monomials in $\{x^i\}$ of degree $k(n+2)$. 
 It is a standard exercise in commutative algebra to show  that
$Q$-cohomology class $\left[p^k M_{k(n+2)}\right]$ 
of $p^k M_{k(n+2)}$ is trivial
 for $k > n$. Then the following isomorphism is obvious;
$$
H^0 
\simeq \bigoplus_{k=0}^{n}\C\left[x^i\right]^{k(n+2)}\biggl/\left<\Fr{\rd G(x)}{\rd x^i}\right>^{k(n+2)}
$$
where the superscript $k(n+2)$ denote homogeneous polynomial of degree $k(n+2)$.
What is not obvious is the following isomorphism
$$
H^0 \simeq \bigoplus_{k=0}^{n} H_{prim}^{n-k, k}(X)
$$
where $H_{prim}^{n-k,k}(X)$ denote the primitive part of Dolbeault cohomology of the Calabi-Yau
$n$-fold $X$. This is due to the residue map Griffths \cite{Griffiths}.
 Now choose a basis $\{e_\a\}$ of $H^0$ as a finite-dimensional $\C$-vector space 
 such that its set  $\{O_\a\}$of representatives  
are invariant monomials among
$$
1, p M_{n+2}, p^2 M_{2n+2}, \cdots, p^{n} M_{n(n+2)},
$$
and $O_\a=1$. Then we have a $\Bbbk$-linear map $f:H=H^0\rightarrow \sC^0$
such that $f(e_\a)=O_\a$, $f(1)=1$, $[f(e_\a)]=e_\a$ and $\bos{K}f=0$. Thus
the BV QFT $\big(\sC[[\hbar]], \bos{K}=-\hbar \Delta + Q,\;\cdot\;\big)$ is semi-classical
and the semi-classical master equation \semicl\ as well as the corollary \ref{flatco} applies.
It is also a matter of computation to find explicit solution once a particular form of $S_{cl}$
is given. For example, consider the following classical action
$$
S_{cl}= p \sum_{i=1}^{n+2} (x^i)^{n+2}.
$$ 
The problem of solving the semi-classical master equation reduces to  a sequence of 
ideal membership problem, which can be implanted as a code for an algebraic package
dealing Gr\"obner basis. For $n=3$, thus for the Fermat quintic hypersurface the
dimension of $H=H^0$ is $204=1+101+101+1$, which is a large number requiring some CPU time.
Candelas et al. in \cite{COGP} originally studied one parameter family of Calabi-Yau Quintic,
i.e., $\Theta_1 = 5 t x^1\cdots x^5$ and determined the special coordinates and Picard-Fuchs
equation for period. 

\end{example}

\begin{remark}

The classical action $S(z)_{cl}= p\cdot G(x) \in \sC^0_{cl}$ has been borrowed from a holomorphic
superpotential in the gauged linear sigma model of Witten \cite{Witten1}. 
The above example may be generalized to any Calabi-Yau space $X$ based on toric geometry - toric
hypersurface, complete intersection of toric hypersurfaces.
The generating functional of quantum correlation functions for the present case is closely related 
with certain extended variations of Hodge structure on $X$, which details is beyond the scope of this
paper \cite{Comp1}. Solution of semi-classical master equation implies the
Picard-Fuchs equation. 
\end{remark}

In our philosophy  the quantum coordinates is nothing but a geometrical avatar of 
quasi-isomorphism as QFT algebra. All these seem suggesting that mirror symmetry 
may be stated in the similar fashion via QFT algebra with certain additional structure. 
The goal of this on going series is to file up evidences that quantum field theory is 
a study of morphism of QFT algebras such that two quantum field theories are physically equivalent
 if the associated QFT algebras are quasi-isomorphic, while "renormalizing" the notion 
 of QFT algebra as we proceed further.

%

\end{document}